\documentclass[journal,10pt]{IEEEtran}
\usepackage{graphicx}
\usepackage{setspace}
\usepackage{amssymb,amsfonts}
\usepackage[tbtags]{amsmath}
\usepackage[utf8]{inputenc}
\usepackage{mathtools}
\usepackage{bbm}
\usepackage{cite}
\usepackage{bm}
\usepackage{siunitx}
\usepackage{algorithm}
\usepackage{algpseudocode}
\usepackage{stfloats}
\usepackage{subfigure}
\usepackage{color}
\usepackage{xcolor}
\usepackage{amsthm}
\usepackage{enumerate}
\usepackage{multirow}
\usepackage{array}
\usepackage{url}

\usepackage{tabularx,booktabs}
\newtheorem{theorem}{Theorem}
\newtheorem{lemma}{Lemma}

\newtheorem{definition}{Definition}
\newtheorem{proposition}{Proposition} 
\newtheorem{assumption}{Assumption}

\allowdisplaybreaks[4]

\algblock{ParFor}{EndParFor}
\algnewcommand\algorithmicparfor{\textbf{parfor}}
\algnewcommand\algorithmicpardo{\textbf{do}}
\algnewcommand\algorithmicendparfor{\textbf{end\ parfor}}
\algrenewtext{ParFor}[1]{\algorithmicparfor\ #1\ \algorithmicpardo}
\algrenewtext{EndParFor}{\algorithmicendparfor}

\begin{document}

% \setstretch{1}
\linespread{1}

	\bibliographystyle{IEEEtran} 
	\title{
% Privacy-Enhanced Wireless Personalized Federated Learning: Quantization-Aided Differential Privacy and Fairness-Driven Scheduling 
Enhancing Convergence, Privacy and Fairness for Wireless Personalized Federated Learning: Quantization-Assisted Min-Max Fair Scheduling
}

\author{Xiyu~Zhao,
Qimei~Cui,~\IEEEmembership{Senior~Member,~IEEE}, Ziqiang Du, Wei~Ni, Weicai~Li,~\IEEEmembership{Graduate~Student~Member,~IEEE},\\ Xi~Yu,~\IEEEmembership{Graduate~Student~Member,~IEEE}, Ji Zhang, 
Xiaofeng Tao,~\IEEEmembership{Senior Member,~IEEE},\\ and
Ping Zhang,~\IEEEmembership{Fellow,~IEEE}
 %\\and H. Vincent Poor,~\IEEEmembership{Life Fellow,~IEEE}  
		% \thanks{X. Zhao is with the Beijing University of Posts and Telecommunications, Beijing 100876, China. }
        \thanks{Manuscript received 28 October 2024; revised 18 December 2024; accepted 22 April 2025. 
        This work was supported by the National Key Research and Development Program of China under Grant No. 2020YFB1806804, and the Beijing Natural Science Foundation Program under Grand No.L232002.
        % The work was supported by the Joint funds for Regional Innovation and Development of the National Natural Science Foundation of China (Grant No. U21A20449) and the National Key Research and Development Program of China (Grant No. 2020YFB1806804).
        }
  \thanks{X. Zhao, Q. Cui, W. Li, X. Yu, X. Tao, and P. Zhang are with the School of Information and Communication Engineering, Beijing University of Posts and Telecommunications, Beijing 100876, China. 
  X. Zhao is also with the School of Computing, Macquarie University, Sydney, NSW 2109, Australia.
  Q. Cui, X. Tao, and P. Zhang are also with the Department of Broadband Communication, Peng Cheng Laboratory, Shenzhen 518055, China (e-mail: \{zxy, cuiqimei, liweicai, yusy, taoxf, pzhang\}@bupt.edu.cn).}	
  \thanks{W.~Ni is with the School of Computing, Macquarie University, Sydney, NSW 2109, Australia. (e-mail: wei.ni@mq.edu.au).}
  \thanks{Z.~Du and J.~Zhang are with China Telecom, Sichuan Branch (15328856888@189.cn, zhangji@sctel.com.cn).}
% \thanks{Y.~Xin is with the College of Big Data and Information Engineering, Guizhou University, Guiyang, China (e-mail: mec.xxu22@gzu.edu.cn).}
  % \thanks{Q. Z. Sheng is with the School of Computing, Macquarie University, Sydney, NSW 2109, Australia (e-mail: michael.sheng@mq.edu.au).}  
\thanks{ \textit{Corresponding authors: Q.~Cui, J.~Zhang} (e-mail: cuiqimei@bupt.edu.cn, zhangji@sctel.com.cn)}}
 
	\maketitle
	\begin{abstract}
 Personalized federated learning (PFL) offers a solution to balancing personalization and generalization by conducting federated learning (FL) to guide personalized learning (PL). Little attention has been given to wireless PFL (WPFL), where privacy concerns arise. Performance fairness of PL models is another challenge resulting from communication bottlenecks in WPFL. 
This paper exploits quantization errors to enhance the privacy of WPFL and proposes a novel quantization-assisted Gaussian differential privacy (DP) mechanism. 
We analyze the convergence upper bounds of individual PL models by considering the impact of the mechanism (i.e., quantization errors and Gaussian DP noises) and imperfect communication channels on the FL of WPFL. 
By minimizing the maximum of the bounds, we design an optimal transmission scheduling strategy that yields min-max fairness for WPFL with OFDMA interfaces.
This is achieved by revealing the nested structure of this problem to decouple it into subproblems solved sequentially for the client selection, channel allocation, and power control, and for the learning rates and PL-FL weighting coefficients.
Experiments validate our analysis and demonstrate that our approach substantially outperforms alternative scheduling strategies by $87.08\%$, $16.21\%$, and $38.37\%$ in accuracy, the maximum test loss of participating clients, and fairness (Jain's index), respectively.
	\end{abstract}

	\section{Introduction}
  Personalized federated learning (PFL) has been recently proposed to account for both generalization and personalization. It can strike a balance between personalized models and the global model, e.g., via a global-regularized multi-task framework~\cite{li2021ditto}. {Only several studies~\cite{sami2023over,mestoukirdi2023user,zhao2023ensemble,you2023hierarchical} have attempted to integrate PFL under a wireless setting, compared to significant efforts on wireless federated learning (FL). Moreover, existing works on PFL, e.g.,~\cite{li2021ditto,li2019fedmd,fallah2020personalized,wei2023personalized,you2022semi,t2020personalized,li2020federated,huang2021personalized,luo2022adapt,zhang2023fedala,zhang2023federated}, have focused on the accuracy or performance distribution fairness under the assumption of ideal (wired) communication environments.}

 %  	\begin{figure}[!t]
	% \centering
 %        \includegraphics[width=0.45\textwidth]{pic/diagram.pdf}
	% \caption{The diagram of DP-Ditto: In each round, every client trains its local model and its PL model based on its local dataset and the FL global model broadcast by the server in the last round. Then, the clients perturb and upload their FL local models to the server, and the server aggregates the perturbed FL local models into the FL global model and broadcasts the global model.}
	% \label{fig:PFL_model}
	% \end{figure}

A challenge arising from wireless PFL (WPFL) is the communication bottleneck, especially when dynamic, noisy, and resource-constrained wireless channels are considered. The channels of geographically dispersed clients can differ significantly and change frequently. 
While more participating clients and, hence, larger overall training datasets are conducive to the convergence of WPFL, they could congest wireless resources and hinder FL model uploading~\cite{cui2024overview}. 
{While wireless FL (WFL) focuses primarily on optimizing the performance of the global model across all clients~\cite{chen2020joint,9798757,9815289,10064038}, WPFL emphasizes the performance of each individual personalized model. In WPFL, the challenge arises from how heterogeneous wireless communication conditions affect the performance of different personalized models. Therefore, designing WPFL scheduling and parameter adjustment methods to efficiently handle resource constraints and imperfect channels becomes more challenging compared to traditional WFL, which optimizes a single global model shared across all clients.}

{Another challenge is that WPFL is prone to privacy leakage due to the incorporation of FL under the wireless setting. Differential privacy (DP)~\cite{abadi2016deep} can be applied to protect the privacy of WPFL. While the majority of the existing studies, e.g.,~\cite{wei2020federated,zhao2020local,truex2020ldp,yuan2023amplitude,liu2024differentially,chen2022feddual}, have straight-forwardly added artificial noise to implement DP in FL/PFL, {several pioneering studies e.g.,~\cite{lang2023joint,lyu2024secure,wang2024p2cefl}}, have explored privacy protection through quantization since quantization inherently brings errors and can help perturb FL local models. 
Stochastic quantization with random quantization lattices~\cite{lang2023joint} or random mappings~\cite{lyu2024secure,wang2024p2cefl} has been considered for mathematical tractability, which is unfortunately less compatible with practical wireless systems. 
{
Specifically, random mapping requires additional random mapping functions~\cite{lyu2024secure,wang2024p2cefl}. For stochastic dithering quantization, synchronized random seeds for the generation of dithering and scaling factors are required during decoding~\cite{lang2023joint}.
These processes introduce complexities that can be difficult to manage within the real-time constraints of communication networks.} 
% {These methods introduce complexities that are difficult to manage within the real-time constraints of communication networks, which is incompatible with practical wireless systems.}
Moreover, fairness is critical to WPFL and can be difficult to achieve. It is affected by scheduling strategies in wireless channels with limited bandwidths. 
Neither of~\cite{lang2023joint} and~\cite{lyu2024secure} has given thought to fairness. }

This paper presents a novel quantization-assisted Gaussian DP mechanism and transmission scheduling strategy for WPFL with orthogonal frequency division multiple access (OFDMA) interfaces, where quantization errors are exploited to enhance the privacy of WPFL while fairness is strengthened among personalized models through scheduling. Specifically, we analyze the convergence upper bound of WPFL in the presence of errors caused by quantization, DP, and imperfect communication channels. Based on the convergence upper bound, the transmission scheduling strategy is designed to achieve the min-max fairness of WPFL by jointly optimizing the client selection, channel allocation, power control, and the weighting coefficients between personalized learning (PL) and FL models, adapting to the channel conditions, as well as the privacy requirement of WPFL.

The contributions of this paper are summarized as follows:
\begin{itemize}
     \item We propose the exploitation of quantization errors to enhance the privacy of WPFL and develop the new quantization-assisted Gaussian mechanism. We analyze the cumulative privacy loss of the mechanism. 
    \item 
     A convergence upper bound of WPFL is derived, characterizing the impact of the new quantization-assisted Gaussian mechanism and the imperfect channel conditions on the convergence of WPFL.
     \item 
     While the impact of the PFL learning rates and PL-FL weighting coefficients on the PL model convergence is intricate, a new min-max problem is formulated to enhance the convergence of wireless {PFL} and maintain fairness 
    by retaining the consistency of the convergence rates among the clients and minimizing the maximum convergence bias of all clients. 
    \item
    A new scheduling strategy is developed to solve the min-max problem by revealing the nested structure of the problem and decoupling the problem into two subproblems solved sequentially for the client selection, channel allocation, and power control, and for the learning rate and PL-FL weighting coefficients.
\end{itemize}

Extensive experiments validate our convergence analysis of the WPFL under the new quantization-assisted Gaussian mechanism.
Three image classification tasks are performed using a deep neural network (DNN), multi-class linear regression (MLR), and convolutional neural network (CNN) on the Federated MNIST, Federated FMNIST, and Federated CIFAR10 datasets. 
Under the CNN model, our approach substantially outperforms its alternative scheduling schemes, i.e., round-robin, random selection, and non-adjustment, by at least $87.08\%$, $16.21\%$, and $38.37\%$ in accuracy, the maximum test loss of participating clients, and fairness (measured by Jain's index), respectively.
Under DNN and MLR models, while our approach slightly outperforms the alternatives in fairness, it is at least
$52.26\%$ and $15.99\%$ better in accuracy and the maximum test loss, respectively.

% Our approach substantially outperforms its benchmarks FedAMP~\cite{huang2021personalized}, pFedMe~\cite{t2020personalized}, APPLE~\cite{luo2022adapt}, and FedALA~\cite{zhang2023fedala}, {by at least $11.30\%$ and $38.46\%$ in accuracy and maximum test loss of all participating clients, respectively. (measured by Jain's index). }

The rest of this paper is structured as follows. Section II reviews the related works. Section III outlines the system and threat models. 
Section~IV elaborates on the new quantization-assisted Gaussian mechanism and analyzes its privacy budget.
In Section~V, the convergence upper bound of WPFL is established under the new mechanism. In Section~VI, we develop a min-max fair scheduling strategy to accelerate the convergence {in a fair fashion}. Experimental results are presented in Section~VII. Conclusions are drawn in Section~VIII.

\textit{Notation:} $\parallel\cdot\parallel$ denotes the $L_2$-norm of a vector or matrix; $|\cdot|$ stands for cardinality; $\nabla(\cdot)$ takes gradient; $\circ$ takes the element-wise product of two vectors or matrices.

\section{{Related Work}}     

\subsection{Personalization}
PFL has been explored to combat statistical heterogeneity among participants through transfer learning (TL) \cite{li2019fedmd}, meta-learning \cite{fallah2020personalized,wei2023personalized,you2022semi}, and multitask learning (MTL) \cite{t2020personalized,li2020federated,huang2021personalized,luo2022adapt,zhang2023fedala,zhang2023federated}. 
TL-based FL enhances personalization by diminishing domain discrepancy of the global and local models~\cite{tan2022towards}. 
FedMD \cite{li2019fedmd} is an FL structure grounded in TL and knowledge distillation (KD), enabling clients to formulate autonomous models utilizing their private data. 

% Meta-learning can be applied to FL to optimize the global model for fast personalization. 
Meta-learning finds utility in enhancing the global model for rapid personalization.
In~\cite{fallah2020personalized}, a 
% variant of FedAvg called Per-FedAvg built on top of 
variation of FedAvg, named Per-FedAvg, was introduced, leveraging
the Model-Agnostic Meta-Learning (MAML). 
% The goal was to learn a good initial global model that performs well on a new heterogeneous task after updating it within just a few gradient descent steps.
It acquired a proficient initial global model that is effective on a novel heterogeneous task and can be achieved through only a few gradient descent steps.
You \textit{et al.} \cite{you2022semi} proposed Semi-Synchronous Personalized {FederatedAveraging (PerFedS$^2$)} based on MAML.
In~\cite{wei2023personalized}, a privacy budget allocation scheme based on R\'{e}nyi DP composition theory was designed to address information leakage arising from two-stage gradient descent of meta-learning-based PFL.

MTL trains a model to simultaneously execute several related tasks. 
% By treating each FL client as a task in MTL, there is the potential to capture the relationships among the clients exhibited by their heterogeneous local data.
In~\cite{t2020personalized}, pFedMe employing Moreau envelopes as the regularized loss functions for clients was recommended to disentangle the optimization of personalized models from learning the global model.
The global model aggregates the local models updated based on the personalized models. Each client's personalized model maintains a bounded distance from the global model.
In \cite{li2020federated}, 
% FedProx was designed by introducing a proximal term to the local subproblem. As a result, the dissimilarity was captured between the global FL model and local models to facilitate adjusting the impact of local updates. 
FedProx was formulated by incorporating a proximal term into the local subproblem. Contrast was delineated between the global and local models to ease the influence of local updates.
In \cite{zhang2023federated}, a federated multi-task learning (FMTL) framework was developed, where the server broadcasts a set of global models aggregated based on the local models of different clusters of clients. Each client selects one of the global models for local model updating.

Huang \textit{et al.} \cite{huang2021personalized} integrated PFL with supplementary terms and employed a federated attentive message passing (FedAMP) strategy to mitigate the impact of diverse data.
A protocol named APPLE \cite{luo2022adapt} was proposed to improve the personalized model of each client based on the others' models. Clients obtain the personalized models locally by aggregating the core models of other clients downloaded from the server. The aggregation weights and the core models are locally learned from the personalized model by adding a proximal term to the local objectives. Instead of overwriting the old local model with the downloaded global model, FedALA \cite{zhang2023fedala} aggregates the downloaded global model and the old local model for local model initialization. 

Some recent studies~\cite{sami2023over,mestoukirdi2023user,zhao2023ensemble,you2023hierarchical} have started integrating PFL in wireless networks. In \cite{sami2023over}, over-the-air clustered FL was designed to enable spectrum sharing across different clusters by employing a coordinated precoder design. 
In \cite{mestoukirdi2023user}, user-centric aggregation was designed, where the server aggregates personalized models based on collaboration coefficients heuristically determined at each round. K-means clustering was applied to cluster users based on their similarity and serve each group of similar users with one personalized model. 
Ensemble FL~\cite{zhao2023ensemble} was proposed by integrating intra-cluster FL models via model ensemble. {Clusters were formed} to improve data distribution similarity and expected energy consumption using a coalition formation game solved by a Nash-stable algorithm.
In \cite{you2023hierarchical}, three-layer FL was adopted, where edge servers aggregate local updates in multiple clusters, and a cloud server implements global aggregation. Scheduling and bandwidth allocation were optimized to balance training loss minimization and round latency minimization.

However, these existing studies~\cite{li2019fedmd,fallah2020personalized,wei2023personalized,you2022semi,t2020personalized,li2020federated,huang2021personalized,luo2022adapt,zhang2023fedala,zhang2023federated,sami2023over,mestoukirdi2023user,zhao2023ensemble,you2023hierarchical} have focused primarily on model accuracy. None has taken fairness among the PL models of different participants.

\subsection{Privacy}
Privacy has been increasingly valued in FL studies~\cite{wei2020federated,zhao2020local,truex2020ldp,yuan2023amplitude,liu2024differentially,chen2022feddual} have explored ways to integrate privacy techniques into FL to provide a demonstrable assurance of safeguarding privacy.
However, little to no consideration has been given to the personalization of learning models and their fairness under imperfect communications and privacy techniques.
In \cite{wei2020federated}, a DP-based framework was 
% proposed to prevent information leakage by injecting noise to protect the privacy of the local model parameters.  
suggested to avert privacy leakage by introducing noise to obfuscate the local model parameters.
In \cite{zhao2020local}, three local DP (LDP) techniques 
% were developed to preserve privacy in data analysis tasks. The LDP mechanisms were integrated into FL to predict traffic status, alleviate privacy threats, and reduce communication overhead in crowd-sourcing applications.
were devised to uphold the privacy of FL, and diminish communication overhead in crowd-sourcing scenarios.
% The authors of \cite{truex2020ldp} 
% % proposed FL with LDP, where LDP-based perturbation was performed upon model updating and sharing, according to the local privacy budget. 
% suggested FL with LDP, wherein LDP was applied during model uploading, adhering to individual privacy budgets.
Liu \textit{et al.} \cite{liu2024differentially} proposed a transceiver protocol to maximize the convergence rate under privacy constraints in a MIMO-based DP FL system, where a server performs over-the-air model aggregation and parallel private information extraction from the uploaded
local gradients with a DP mechanism.

In \cite{yuan2023amplitude}, DP noises were adaptively added to local model parameters to preserve user privacy during FL. The amplitude of DP noises was adjustable to preserve privacy and encourage convergence. 
FedDual \cite{chen2022feddual} was designed to add DP noises locally and aggregate asynchronously via a gossip protocol. Noise-cutting was adopted
to alleviate the impact of the DP noise on the global model. 
% Hu \textit{et al.} \cite{hu2020personalized} proposed privacy-preserving PFL using the Gaussian mechanism, which provides a privacy guarantee by adding Gaussian noise to the uploaded local updates. 
In \cite{liu2022privacy}, the Gaussian mechanism was considered in a mean-regularized MTL framework, and the accuracy was analyzed for single-round FL using a Bayesian framework.
In \cite{okegbile2023differentially}, differentially private federated MTL was designed for human digital twin systems with computationally efficient blockchain-enabled validation.

Some studies~\cite{lang2023joint,lyu2024secure,wang2024p2cefl} have utilized stochastic quantization in support of DP. In~\cite{lang2023joint}, devices utilize vector quantization based on random lattices to compress their noise-perturbed local models, achieving a predefined privacy level by adding noise and exploiting quantization errors. 
In \cite{lyu2024secure}, a secure and efficient FL framework was proposed by adding a stochastic quantization module at the client to quantize the local gradients for global aggregation. A new metric was designed to analyze the privacy and a trade-off between communication overhead, convergence rate, and privacy concerning the quantization interval.
{In~\cite{wang2024p2cefl}, an FL algorithm preserving privacy and efficiency of communication (P2CEFL) was proposed, where a subtractive dithering approach was employed to reduce communication overhead under DP guarantee.}

% {Some studies~\cite{lang2023joint,lyu2024secure} have started to utilize stochastic quantization to implement DP. In~\cite{lang2023joint}, devices utilize vector quantization based on random lattices to compress their noise-perturbed local models, achieving a predefined privacy level by adding noise and exploiting quantization errors. 
% In \cite{lyu2024secure}, a secure and efficient FL framework was proposed by adding a stochastic quantization module at the client to quantize the original local gradients for global aggregation. A new metric was designed to analyze the privacy performance and a trade-off between communication overhead, convergence rate, and privacy protection concerning the quantization interval.
% % {These studies~\cite{lang2023joint,lyu2024secure} designed stochastic quantization and are incompatible with practical wireless systems.}
% }

However, none of the above works~\cite{wei2020federated,zhao2020local,truex2020ldp,yuan2023amplitude,chen2022feddual,liu2024differentially,hu2020personalized,lang2023joint,lyu2024secure,liu2022privacy,wei2023personalized,wang2024p2cefl} have considered fairness among the participants in PFL.
%especially in the presence of statistical heterogeneity.

\subsection{Fairness}

Some existing studies, e.g., \cite{li2020fair,hu2022federated,li2021ditto}, have attempted to improve performance distribution fairness, i.e., by 
%reducing the variance of model accuracy across clients. 
mitigating the variability in model accuracy among different clients.
Yet, none has taken user privacy into account.
In \cite{li2020fair}, $q$-FFL was proposed to achieve a more uniform accuracy distribution across clients. A parameter $q$ was used to re-weight the aggregation loss by assigning bigger weights to clients undergoing more significant losses.
In \cite{hu2022federated}, FedMGDA+ 
% was proposed to improve the model’s robustness while maintaining good-intent fairness. 
was suggested to enhance model robustness while upholding fairness with positive intentions.
A multi-objective problem 
% was constructed to minimize the loss functions of all clients, and solved using Pareto-stationary solutions to find a common descent direction for all selected clients.
was structured to diminish the loss functions across all clients, tackled by employing Pareto-steady resolutions to pinpoint a collective descent direction suitable for all chosen clients.
Li \textit{et al.} \cite{li2021ditto} designed a scalable federated MTL framework, Ditto, which learns personalized and global models in a global-regularized framework. 
% A regularization term was added to make the personalized models close to the optimal global model.
Regularization was introduced to bring the personalized models in proximity to the optimal global model.
The optimal weighting coefficient of Ditto was designed in terms of fairness and robustness.
%In a meta-learning-based PFL framework, a privacy budget allocation scheme based on R\'{e}nyi DP composition theory was designed considering the challenge of information leakage due to two-stage gradient descent \cite{wei2023personalized}.
Unfortunately, these studies~\cite{li2020fair,hu2022federated,li2021ditto} have overlooked privacy risks or failed to address the influence of DP and imperfect communications on fairness.

\section{System Model and Problem Statement}
In this section, we present the PFL system, channel model, threat model, and the preliminary of DP.
The PFL system consists of a server and $N$ clients. $\mathcal{N}$ denotes the set of clients. $\mathcal{D}_n$ denotes the local dataset at client $n \in \mathcal{N}$. $\mathcal{D}$ collects all data samples. $\left|\mathcal{D}\right|={\sum}_{n=1}^{N}\left|\mathcal{D}_{n}\right|$.
% with $\left|\cdot\right|$ standing for cardinality. 
% To adapt to the heterogeneity under federated settings, we consider 
The PFL has both global and personalized objectives for FL and PL, respectively. 

\subsection{PFL Model}

{
At every communication round~$t$, the server selects a subset of clients $\mathcal{N}_{t} \subset \mathcal{N}$, and quantizes and sends the latest global model to the clients.  
Upon the receipt of the noisy global model, i.e., through imperfect downlink channels, each client $n$, $\forall n \in \mathcal{N}_{t}$, executes local FL training, and updates its FL local model. {Each client $n$, $\forall n \in \mathcal{N}$, executes PL training, and updates its PL model.}
% The FL local learning rate and the PL learning rate are $\eta_{\mathrm{F},n}^{t}$ and $\eta_{\mathrm{P},n}^{t}$, respectively. 
After clipping, DP perturbation, and quantization, client $n$ uploads its local models to the server. Based on the received local models  
 %$\hat{\boldsymbol{\omega}}_{n}^{t+1}$ 
from client $n$, the global model
 %$\overset{\sim}{\boldsymbol{\omega}}_{\mathrm{L}}^{t+1}$
 is obtained by aggregation at the server. }

\subsubsection{FL}
As for FL, the global objective of FL is to learn an FL global model with the minimum global training loss, i.e., 
{
\begin{equation} \small
    \label{GlobalObj}
    \underset{\boldsymbol{\omega}}{\min}\, F(\boldsymbol{\omega})={\sum}_{n=1}^{N}p_nF_n(\boldsymbol{\omega})  \,,
\end{equation}}%
where $\boldsymbol{\omega}\in\mathbb{R}^{ |\boldsymbol{\omega}|}$ is the model parameter with $|\boldsymbol{\omega}|$ elements; $F(\cdot)$ is the global loss function, 
% as given by
% \begin{equation}
%     \label{g_loss}
%     F(\boldsymbol{\omega})={\sum}_{n=1}^{N}p_nF_n(\boldsymbol{\omega}) \,,
% \end{equation}
and $F_{n}(\cdot)$ is the local loss function of client $n \in \mathcal{N}$; $p_{n}\triangleq\frac{\left|\mathcal{D}_{n}\right|}{\left|\mathcal{D}\right|}$ is the aggregation coefficient for client $n$, with ${\sum}_{n=1}^{N}p_n=1$. 
For illustration convenience, we assume the size of each client's local dataset is the same, i.e., $p_n=\frac{1}{N}$.

According to 
% (\ref{objective_p})--
(\ref{GlobalObj}), on each communication round, an FL local model, denoted by $\boldsymbol{u}_n^t$, is trained at every selected client $n$, followed by clipping, DP perturbation, and quantization, before the client uploads the FL local model to the server.

\textbf{Clipping}: 
The FL local models are clipped as
{
 	\begin{equation} \small
    \boldsymbol{u}_{n}^{t}=\boldsymbol{u}_{n}^{t}\Big/\max\big(1,{\| \boldsymbol{u}_{n}^{t}\|}\big/{C}\big)  \,,
    \label{clipping} 
	\end{equation}}%
 where $C$ is the pre-determined clipping threshold ensuring that the local model parameter $\parallel \boldsymbol{u}_{n}^{t}\parallel \leq C$~\cite{yuan2023amplitude}.

 \textbf{DP perturbation:}
 % The DP noise is added after a client clips its FL local model.
  Let $\mathbf{z}_n^t$ denote the independent and identically distributed (i.i.d.) Gaussian noise added by client $n$ to its local model $\boldsymbol{u}_{n}^{t}$ at the $t$-th communication round. Each element in $\mathbf{z}_n^t$ follows $\mathbb{N}(0,\sigma_\mathrm{DP}^2)$.

  \begin{definition}[$(\epsilon,\delta)$-DP]
A privacy preserving mechanism $\mathcal{M}:\mathcal{X}\rightarrow \mathcal{R}$ is $(\epsilon,\delta)$-DP if, for any two adjacent datasets $\mathcal{X}_{0}$, $\mathcal{X}_{1}\in \mathcal{X}$ and any subset of outputs $S\subseteq \mathcal{R}$, it holds that
\begin{equation} \small
    \label{DefDP}
    \Pr\left[\mathcal{M}\left(\mathcal{X}_{0}\right)\in\mathcal{S}\right]\leq e^{\epsilon}\Pr\left[\mathcal{M}\left(\mathcal{X}_{1}\right)\in\mathcal{S}\right]+\delta  \,,
\end{equation}
where $\epsilon>0$ specifies the difference beyond which the outputs concerning $\mathcal{X}_{0}$ and $\mathcal{X}_{1}$ can be differentiated, and $\delta \in[0,1]$ is the probability with which
the ratio between the probabilities of $\mathcal{X}_{0}$ and $\mathcal{X}_{1}$ is no smaller then $e^\epsilon$.
\end{definition}
\begin{definition}[Max Divergence]
The Max Divergence, also known as the $\infty$-th order of $R\acute{e}nyi$ divergence, between two random variables $Y$ and $Z$ taking values from the same sample space $\mathcal{Y}$ is defined as
{\begin{equation}  \small
D_{\infty}(Y||Z)=\max_{y\in \mathcal{Y}}\left[\ln\frac{\Pr\left[Y=y\right]}{\Pr\left[Z=y\right]}\right].
    \end{equation}}%
The $\delta$-Approximate Max Divergence 
% between $Y$ and $Z$ 
is defined as
{\begin{equation} \small
D_{\infty}^{\delta}(Y||Z)=\max_{y\in \mathcal{Y}}\left[\ln\frac{\Pr\left[Y=y\right]-\delta}{\Pr\left[Z=y\right]}\right].
    \end{equation}}%
The randomized mechanism $\mathcal{M}:\mathit{\mathcal{D\rightarrow R}}$ satisfies $\epsilon$-DP if $D_{\infty}\left[\mathcal{M}\left({\mathcal{X}_{0}}\right)||\mathcal{M}\left({\mathcal{X}_{1}}\right)\right]\leq \epsilon$ for all measurable sets $\mathcal{Y}\subseteq\mathcal{R}$ and any two adjacent datasets ${\mathcal{X}_{0}},{\mathcal{X}_{1}}\in\mathcal{D}$. 
Moreover, $\mathcal{M}$ satisfies $(\epsilon,\delta)$-DP if 
% for all measurable sets $\mathcal{Y}\subseteq\mathcal{R}$ and for any two adjacent databases ${\mathcal{X}_{0}},{\mathcal{X}_{1}}\in\mathcal{D}$, 
$D_{\infty}^{\delta}\left[\mathcal{M}\left({\mathcal{X}_{0}}\right)||\mathcal{M}\left({\mathcal{X}_{1}}\right)\right]\leq \epsilon$\cite{dwork2014algorithmic}.
\end{definition}

\textbf{Quantization}:
The clients quantize and transmit their FL local models to the server for FL global aggregation. 
The server quantizes and broadcasts the FL global model. 
Suppose that each element of the FL local and global models is quantized into $R$ bits. 
%, with $1$ sign bit and $R-1$ integer bits.
%The ranges of the fixed-point number for the local models and the global models are $[-C+3\sigma_{\mathrm{DP}},C+3\sigma_{\mathrm{DP}}]$ and $[-\Lambda_{\mathrm{G}},\Lambda_{\mathrm{G}}]$, with $C+3\sigma_{\mathrm{DP}}=C+3\sigma_{\mathrm{DP}}$ and $\Lambda_{\mathrm{G}}=C$. 
% The floating-point numbers of the local models and the global models are constrained within $[-C+3\sigma_{\mathrm{DP}},C+3\sigma_{\mathrm{DP}}]$ and $[-\Lambda_{\mathrm{G}},\Lambda_{\mathrm{G}}]$, \textcolor{blue}{with $C+3\sigma_{\mathrm{DP}}=C+3\sigma_{\mathrm{DP}}$ and $\Lambda_{\mathrm{G}}=C$}, and then mapped to corresponding fixd-point numbers.
We set the quantization range of the FL local models to $[-C-3\sigma_{\mathrm{DP}},C+3\sigma_{\mathrm{DP}}]$, capturing 99.7\% of the local models perturbed by the Gaussian mechanism.
With no perturbation, the quantization range of the FL global models is $[-C,C]$.
% \textcolor{blue}{This is because the DP noise is only added for model uploading, and about 99.7\% of values drawn from a Gaussian distribution are within three standard deviations from the mean.}
Then, the respective quantization intervals of the FL local and global models are given by
	\begin{align} 
    &\triangle_{\mathrm{L}}=\frac{2(C+3\sigma_{\mathrm{DP}})}{2^{R}-1}; \,
    \triangle_{\mathrm{G}}=\frac{2C}{2^{R}-1} .
    \label{LSB} 
	\end{align}
The respective maximum quantization errors of the FL local and global models are given by
 	\begin{align} 
    & E_{\mathrm{L}}^{\mathrm{max}}=\frac{\triangle_{\mathrm{L}}}{2}\triangleq\beta_{\mathrm{L}}(C+3\sigma_{\mathrm{DP}}); \,
    E_{\mathrm{G}}^{\mathrm{max}}=\frac{\triangle_{\mathrm{L}}}{2}\triangleq\beta_{\mathrm{G}}C \,, \label{Quantization_error} 
	\end{align}
where, for conciseness, $\beta_{\mathrm{L}}\triangleq\frac{1}{2^{R}-1}$ and $\beta_{\mathrm{G}}\triangleq\frac{1}{2^{R}-1}$. 
%We define $\mathcal{Q}_{\mathrm{L}}$ and $\mathcal{Q}_{\mathrm{G}}$ to collect the respective quantization levels of the FL local and global models, respectively.

Let $\mathcal{Q}(\cdot)$ denote the multi-dimensional quantization with every element rounded towards the closest quantization level.
The clipped, perturbed, and quantized FL local model is 
\begin{align}
    \label{multi_quan}
    \tilde{\boldsymbol{u}}_{n}^{t}=\mathcal{Q}({\boldsymbol{u}}_{n}^{t}+{\boldsymbol{z}}_{n}^{t}), 
    % \\
    % \tilde{\boldsymbol{\omega}}_{\mathrm{G}}^{t}&=Q(\tilde{\boldsymbol{\omega}}_{\mathrm{L}}^{t}),
\end{align}
which 
% $\boldsymbol{u}_{n}^{t}$ is the local model of client $n$ in round $t$, and 
% $\tilde{\boldsymbol{u}}_{n}^{t}$ is the noisy local model
is 
uploaded by client $n$ in the $t$-th round.

\subsubsection{PL}
A PL model $\boldsymbol{\varpi}_{n}$ is trained locally at client $n$ concerning the FL global model. 
The training of the FL global model and that of the PL models are synchronized on the basis of FL rounds. Client $n$ updates its PL model $\boldsymbol{\varpi}_{n}^{t}$ based on the FL global model $\boldsymbol{\omega}^t$ updated at the $t$-th round. 
For the sake of model generalization, we encourage the PL model to be close to the optimal FL global model, i.e.,

\vspace{-\baselineskip}
	\begin{subequations} \small
        \label{objective_p}
	\begin{align}
     \underset{\boldsymbol{\varpi}_{n}}{\min} \quad f_{n}(\boldsymbol{\varpi}_{n};\boldsymbol{\omega}^{\ast})\!&=\!\!\big(1\!\!-\!\!\frac{\lambda}{2}\big)\!F_{n}(\boldsymbol{\varpi}_{n})\!\!+\!\!\frac{\lambda}{2}\!\parallel\!\boldsymbol{\varpi}_{n}\!\!-\!\!\boldsymbol{\omega}^{\ast}\!\parallel^{2}, \label{fn} \\
     \label{fn_1}
     \textrm{s.t.}  \,\,\boldsymbol{\omega}^{\ast}&=\underset{\boldsymbol{\omega}}{\arg\min}\, \frac{1}{N} {\sum}_{n=1}^{N} F_{n}\left(\boldsymbol{\omega}\right) \,, 
	\end{align}
	\end{subequations}
 where $f_n(\cdot)$ is the loss function of the PL model at client $n$; $\lambda \in [0,2]$ is a PL-FL {weighting coefficient} that controls {the trade-off between the FL and PL models.} When $\lambda=0$, PFL trains a {PL} model for each client based on its local datasets. When $\lambda=2$, there is no personalization.

\subsection{Communication Model}
% The quantized parameters of the FL local and global models are modulated into $N_{\boldsymbol{\omega}}$ number of {$M_{\boldsymbol{\omega}}$-ary} QAM symbols (i.e., $R=N_{\boldsymbol{\omega}}\log_2 M_{\boldsymbol{\omega}}$).
% % \textcolor{blue}{where $N_{\boldsymbol{\omega}}$ is a pre-determined constant}. 
{The quantized parameters of the FL local and global models are modulated into {$M_{\boldsymbol{\omega}}$-ary} QAM symbols.}
For a model with {$|\boldsymbol{\omega}|$ elements,} the minimum uplink data rate (in bits/s) is 
\begin{equation}
    \label{max_data_rate}
    r_{\mathrm{min}}={|\boldsymbol{\omega}|R}/{\tau_{\mathrm{max}}} \,,
\end{equation}
where $\tau_{\mathrm{max}}$ is the maximum transmission delay (in seconds). 

\subsubsection{Channelization}
The server connects $N$ clients wirelessly over $K$ orthogonal subchannels. $\mathcal{K}=\{1,\ldots,K\}$ denotes the set of subchannels. In each round, at most $K$ clients are selected for local model uploading. Let $\mathbf{p}^{t}=\{P_{1}^{t},\ldots,P_{N}^{t}\}\in \mathbb{R}^{N}$ and $\mathbf{\boldsymbol{c}}^{t}=\{\mathbf{c}_{1}^{t},\ldots,\mathbf{c}_{N}^{t}\}\in \mathbb{R}^{N \times K}$ collect the transmit powers of all clients and their selected subchannels in the $t$-th round. $P_{n}^{t}=0$ if client $n$ selects no subchannel in the round. $\mathbf{c}_{n}^{t}=\{c_{n,1}^{t},\ldots,c_{n,K}^{t}\}$ collects the channel selection indicators, with $c_{n,k}^{t}=1$ if subchannel $k$ is selected for client $n$, and $c_{n,k}^{t}=0$, otherwise. 

Suppose that client $n$ uploads its clipped, perturbed and quantized local model, $\tilde{\boldsymbol{u}}_{n}^{t}$, through subchannel $k$ in the $t$-th round. The uplink data rate is given by
\begin{equation}
    \label{data_rate}
    r_{n,k}^{t}=B\log_{2}(1+\gamma_{n,k,\mathrm{L}}^{t}) \,,
\end{equation}
where $B$ is the bandwidth of each subchannel, and $\gamma_{n,k,\mathrm{L}}^{t}$ is the receive signal-to-noise ratio (SNR) at the server from client $n$ in subchannel $k$ during the $t$-th round, as given by
\begin{equation}
    \label{SNR}
\gamma_{n,k,\mathrm{L}}^{t}={P_{n}^{t}}\left|h_{n,k,\mathrm{L}}^{t}\right|^{2}\Big/{\sigma_{0}^{2}} \,,
\end{equation}
 where $h_{n,k,\mathrm{L}}^{t}$ is the channel of client $n$ in subchannel $k$, and ${\sigma_{0}^{2}}$ is the variance of the additive white Gaussian noise (AWGN).
 %Let $\mathbf{e}^{t}=\{e_{1}^{t},\ldots,e_{N}^{t}\}$ collect the symbol error rate (SER) of each client, as given by

 \subsubsection{Received Models}
 {Let {$e_{n,k,\mathrm{L}}^{t}$} denote the bit error rate (BER) of client $n$ in subchannel $k$~\cite{cho2002general}:
 \begin{equation} 
    \label{SER}
  \!\!\!  \scalebox{0.98}{$
e_{n,k,\mathrm{L}}^{t}\!\!=\!\!\frac{2\sqrt{M_{\boldsymbol{\omega}}}\!-\!2}{\sqrt{M_{\boldsymbol{\omega}}}\log_2\!\!\sqrt{M_{\boldsymbol{\omega}}}}Q\!\!\left(\!\!\sqrt{\frac{3\gamma_{n,k,\mathrm{L}}^{t}\!\log_2\!M_{\boldsymbol{\omega}}}{M_{\boldsymbol{\omega}}-1}}\!\right) \!,$}
\end{equation}}
where $Q(x)=\frac{1}{\sqrt{2\pi}}\int_{x}^{\infty}e^{-\frac{x^{2}}{2}}dx$ is the $Q$ function.
%Note that $\sum_{k=1}^{K}c_{n,k}^{t} \leq 1$, i.e., $0\leq e_{n}^{t}\leq1$. 

Then, for client $n$, the error probability of each element of its FL local model is given by
{
 \begin{equation} \small
    \label{element_error_pr}
    \rho_{n,\mathrm{L}}^{t}={\sum}_{k=1}^{K}c_{n,k}^{t}\left(1-(1-e_{n,k,\mathrm{L}}^{t})^{R}\right) \,.
\end{equation}}
Similarly, we obtain the error probability $\rho_{n,\mathrm{G}}^{t}$ of each element of the FL global model received at client $n$ in the $t$-th round.

In the $t$-th round, we use 
% $\hat{\boldsymbol{\omega}}_{n}^{t}$ denote the received local model of client $n$ at the server in the $t$-th communication round, and 
$\mathbf{s}_{n}^{t}=\{s_{n,i}^{t},\forall i=1,\ldots,{|\boldsymbol{\omega}|}\}\in\mathbb{R}^{|\boldsymbol{\omega}|}$ to denote the error indicator vector for the received local model $\hat{\boldsymbol{\omega}}_{n}^{t}$.  {$s_{n,i}^{t}=0$ if the $i$-th element of $\hat{\boldsymbol{\omega}}_{n}^{t}$ is error-free; otherwise, $s_{n,i}^{t}=1$.}
At the server, the received FL local model of client $n$ and the aggregated FL global model are given by
{\small
    \begin{align} 
        \label{received_local_model}
    \hat{\boldsymbol{\omega}}_{n}^{t}&= \mathbf{s}_{n}^{t}\circ\boldsymbol{\hat{u}}_{n}^{t}+\left(1-\mathbf{s}_{n}^{t}\right)\circ\tilde{\boldsymbol{u}}_{n}^{t}  \,;\\
    \label{aggregated_glb}
    \tilde{\boldsymbol{\omega}}_{\mathrm{L}}^{t}&=\frac{1}{|\mathcal{N}_{t}|} {\sum}_{n\in\mathcal{N}_{t}}\hat{\boldsymbol{\omega}}_{n}^{t} \,,
\end{align}}%
where  
% $\mathcal{N}_{t}$ is the set collecting the selected clients at round $t$, $\tilde{\boldsymbol{u}}_{n}^{t}$ is the local model of client $n$ at the communication round $t$ after noise addition and quantization, and 
$\boldsymbol{\hat{u}}_{n}^{t}$ is the erroneous version of $\tilde{\boldsymbol{u}}_{n}^{t}$ resulting from imperfect, noisy wireless channels.
% {$|\mathcal{N}_{t}|$ is the number of participating clients in the FL.}
Considering the errors caused by DP noise, quantization, and transmission in imperfect channels, $\tilde{\boldsymbol{u}}_{n}^{t}$ and ${\boldsymbol{\hat{u}}}_{n}^{t}$ are given by
{
% {\small
% \begin{align} 
\begin{equation} \small
    \label{noisy_local_model}
    \tilde{\boldsymbol{u}}_{n}^{t}=\boldsymbol{u}_{n}^{t}+\mathbf{z}_{n}^{t}+\mathbf{E}_{n,\mathrm{L}}^{t} \,;\quad
    % \\
    % \label{error_local_model} 
\boldsymbol{\hat{u}}_{n}^{t}=\boldsymbol{u}_{n}^{t}+\boldsymbol{\zeta}_{n,\mathrm{L}}^{t} \,,
% \end{align}}%
\end{equation}}%
where $\mathbf{E}_{n,\mathrm{L}}^{t}=\{{E}_{n,i,\mathrm{L}}^{t}, i=1,\ldots,{|\boldsymbol{\omega}|}\}\in \mathbb{R}^{|\boldsymbol{\omega}|}$ is the quantization error vector of ${\boldsymbol{u}}_{n}^{t}$. $|{E}_{n,i,\mathrm{L}}^{t}| \leq E_{\mathrm{L}}^{\mathrm{max}}$, $i=1,\ldots,{|\boldsymbol{\omega}|}$. $\boldsymbol{\zeta}_{n,\mathrm{L}}^{t}=\{{\zeta}_{n,i,\mathrm{L}}^{t}, i=1,\ldots,{|\boldsymbol{\omega}|}\} \in \mathbb{R}^{|\boldsymbol{\omega}|}$ is the error between $\boldsymbol{\hat{u}}_{n}^{t} $ and $\boldsymbol{u}_{n}^{t}$ caused by DP noise, quantization, and transmission errors. $|{\zeta}_{n,i,\mathrm{L}}^{t}|\leq |{u}_{n,i}^{t}|+C+3\sigma_{\mathrm{DP}}$, $i=1,\ldots,{|\boldsymbol{\omega}|}$.

% \subsubsection{...}
In the $(t+1)$-th round, at client $n$, the received FL global model is given by
\begin{equation}
    \label{error_glb}
    \hat{\boldsymbol{\omega}}_{n,\mathrm{G}}^{t+1}=\mathbf{s}_{n,\mathrm{\mathrm{G}}}^{t+1}\circ\hat{\boldsymbol{\omega}}_{n,\mathrm{L}}^{t+1}+(\mathbf{1}_{|\boldsymbol{\omega}|}-\mathbf{s}_{n,\mathrm{G}}^{t+1})\circ\tilde{\boldsymbol{\omega}}_{\mathrm{G}}^{t} \,,
\end{equation}
where $\mathbf{s}_{n,\mathrm{G}}^{t+1}=\{s_{n,i,\mathrm{G}}^{t+1},\forall i=1,\ldots,{|\boldsymbol{\omega}|}\}\in\mathbb{R}^{|\boldsymbol{\omega}|}$ is the error indicator vector for the transmission of the FL global model.
% for ${\boldsymbol{\omega}}_{\mathrm{G}}^{t+1}$. 
{$s_{n,i,\mathrm{G}}^{t+1}=0$ if the $i$-th element of $\mathbf{s}_{n,\mathrm{G}}^{t+1}$ is received error-free; otherwise, $s_{n,i,\mathrm{G}}^{t+1}=1$. }
Moreover, $\tilde{\boldsymbol{\omega}}_{\mathrm{G}}^{t}$ is the global model after quantization. 
Let {$\hat{\boldsymbol{\omega}}_{\mathrm{L}}^{t+1}$} be the erroneous version of $\tilde{\boldsymbol{\omega}}_{\mathrm{G}}^{t}$. Then,
% \begin{subequations}
% \begin{align}
\begin{equation}
    \label{noisy_glb}
{\tilde{\boldsymbol{\omega}}_{\mathrm{G}}^{t}=\tilde{\boldsymbol{\omega}}_{\mathrm{L}}^{t}+\mathbf{E}_{\mathrm{G}}^{t} \,;\quad
    % \\
    % \label{error_glb_model}
\hat{\boldsymbol{\omega}}_{n,\mathrm{L}}^{t+1}=\tilde{\boldsymbol{\omega}}_{\mathrm{L}}^{t}+\boldsymbol{\zeta}_{n,\mathrm{G}}^{t+1} \,,}
% \end{align}
% \end{subequations}
\end{equation}%
where $\mathbf{E}_{\mathrm{G}}^{t}=\{{E}_{i,\mathrm{G}}^{t},\forall i=1,\ldots,{|\boldsymbol{\omega}|}\}\in \mathbb{R}^{|\boldsymbol{\omega}|}$ is the quantization error vector of $\tilde{\boldsymbol{\omega}}_{\mathrm{L}}^{t}$, with $|{E}_{i,\mathrm{G}}^{t}| \leq E_{\mathrm{G}}^{\mathrm{max}}$, $\forall i$; $\boldsymbol{\zeta}_{n,\mathrm{G}}^{t+1}=\{{\zeta}_{n,i,\mathrm{G}}^{t+1},\forall i=1,\ldots,{|\boldsymbol{\omega}|}\}\in \mathbb{R}^{|\boldsymbol{\omega}|}$ is the error between $\tilde{\boldsymbol{\omega}}_{\mathrm{L}}^{t}$ and $\hat{\boldsymbol{\omega}}_{n,\mathrm{L}}^{t+1}$ caused by downlink quantization and transmission errors, with $|{\zeta}_{n,i,\mathrm{G}}^{t+1}|\leq|{\omega}_{i,\mathrm{L}}^{t}| +C, \,\forall i$. 

\subsubsection{Local Model Update}
The FL local model and the PL model of client $n$ are updated at the $(t+1)$-th communication round based on the received FL global model, as given by
\begin{subequations} \small
\begin{align}
    \label{update_lc_model}
   \boldsymbol{u}_{n}^{t+1}&=\hat{\boldsymbol{\omega}}_{n,\mathrm{G}}^{t+1}-\eta_{\mathrm{F},n}^{t+1}\nabla F_{n}\Big(\hat{\boldsymbol{\omega}}_{n,\mathrm{G}}^{t\!+\!1}\Big) \,;\\
   \tilde{\boldsymbol{\varpi}}_{n}^{t+1}&\!\!=\!\!\tilde{\boldsymbol{\varpi}}_{n}^{t}\!\!-\!\!\eta_{\mathrm{P},n}^{t\!+\!1}\!\!\left[\!\Big(\!1\!\!-\!\!\frac{\lambda_{n}^{t\!+\!1}}{2}\!\Big)\!\nabla F_{n}(\tilde{\boldsymbol{\varpi}}_{n}^{t})\!\!+\!\!\lambda_{n}^{t+1}\!\Big(\!\tilde{\boldsymbol{\varpi}}_{n}^{t}\!\!-\!\hat{\boldsymbol{\omega}}_{n,\mathrm{G}}^{t+1}\!\Big)\!\right] .
   \label{update_pl_model}
\end{align}
\end{subequations}
where $\eta_{\mathrm{F},n}^{t+1}$ and $\eta_{\mathrm{P},n}^{t+1}$ are the FL local learning rate and the PL learning rate of client $n$ at the $(t+1)$-th round, respectively. 

Define $\boldsymbol{u}_{n}^{\ast}\in\mathbb{R}^{|\boldsymbol{\omega}|}$ as the optimal {FL} local model of client $n$, and $\boldsymbol{\varpi}_{n}^{\ast}\in\mathbb{R}^{|\boldsymbol{\omega}|}$ as the optimal PL model, i.e.,
\begin{equation}\small
% \begin{subequations}\small
% \begin{align}
{\boldsymbol{u}_{n}^{\ast}=\underset{\boldsymbol{u}_{n}}{\arg\min}\,F_{n}(\boldsymbol{u}_{n})\,;\quad
% \label{Optimal_fmodel} \\
\boldsymbol{\varpi}_{n}^{\ast}=\underset{\boldsymbol{\varpi}_{n}}{\arg\min}\,f_{n}(\boldsymbol{\varpi}_{n};\boldsymbol{\omega}^{\ast}). }
\label{optimal_pmodel} 
% \end{align}
% \end{subequations}
\end{equation}

\subsection{Threat Model and Problem Statement}

The server may be honest-but-curious, and attempt to recover the training datasets of the clients or infer their private features based on the FL local models uploaded by the clients \cite{nasr2019comprehensive}. There may also be external attackers who intend to breach the privacy of the clients. Although the clients train their FL local models locally, the local models are shared with the server and can be analyzed to potentially compromise their privacy under inference attacks during learning \cite{nasr2019comprehensive} and model-inversion attacks during testing \cite{fredrikson2015model}. 

{In this paper, we wish to exploit the inherent privacy-preserving capability of quantization to enhance the privacy of WPFL.
We also wish to optimize the scheduling policy and power control for the transmissions of the local models to facilitate the convergence of WPFL with privacy enhancement while maintaining performance fairness between the PL models of the participating clients.}

\section{Privacy Analysis of WPFL}
% \footnote{DP mechanism with parameters $\epsilon$ and $\delta$ has developed into a strong standard for privacy guarantees in data processing systems. It provides a rigorous framework for privacy guarantees under various adversarial attacks. Here, $\epsilon > 0$ is the distinguishable bound of all outputs on neighboring datasets $d$, $d'$ in a database, and $\delta$ represents the probability of the event that the ratio of the probabilities for two adjacent datasets $d$, $d'$ cannot be bounded by $\epsilon$ after adding a privacy-preserving mechanism.}
% To preserve data privacy from the uploaded local models, a Gaussian mechanism can be used to guarantee $(\epsilon,\delta)$-DP by adding artificial Gaussian noises \cite{dwork2014algorithmic}. 
% In this section, we delineate the new quantization-assisted Gaussian mechanism, and the procedure of the WPFL under this proposed DP mechanism. 
This section analyzes the impact of quantization on the privacy of WPFL, establishes the new quantization-assisted Gaussian mechanism, and analyzes its privacy budget.
% In this paper, 
This starts with the following proposition.
\begin{proposition}[Quantization-Assisted Gaussian Mechanism]
    \label{joint_DP}
    For a local model $\boldsymbol{u}_{n}(\cdot)$ and its local dataset $\mathcal{D}_n$, the quantization-assisted Gaussian mechanism is defined as
    \begin{equation}
    \label{DP_mechanism}
        \mathcal{M}_\mathrm{Q}(\boldsymbol{u}_{n}(\cdot),\mathcal{D}_n)=\mathcal{Q}(\boldsymbol{u}_{n}(\mathcal{D}_{n})+\boldsymbol{z}_{n}),
    \end{equation}
    where $\mathbf{z}_n$ is the Gaussian noise added by client $n$ to its local models before quantization, and its elements follow $\mathbb{N}(0,\sigma_\mathrm{DP}^2)$.
\end{proposition}

By evaluating the probability distribution of $\mathcal{M}_\mathrm{Q}$ for any quantization level and its Max Divergence, we establish the upper bound of the privacy loss (i.e., $\epsilon$ and $\delta$) for the quantization-assisted Gaussian mechanism in \textbf{Theorem \ref{privacy_budget}}. 
% This is achieved by first analyzing the probability distribution of $\mathcal{M}_\mathrm{Q}$ for any quantization level, followed by obtaining the Max Divergence and utilizing Composition Theorem~\cite[Thm 3.16]{dwork2014algorithmic}. 
{ 
\begin{theorem}
\label{privacy_budget}
Given the privacy budget $\epsilon_\mathrm{Q}$, the quantization-assisted Gaussian mechanism $\mathcal{M}_\mathrm{Q}$ (\ref{DP_mechanism}) satisfies $(\epsilon_\mathrm{Q},\delta_\mathrm{Q})$-DP: 
    \begin{align}
        \label{delta_Q}
        \delta_\mathrm{Q}&=T_0\cdot\max \big\{\psi-\psi_{1}e^{\frac{\epsilon_\mathrm{Q}}{T_0}},\,\psi'-\psi_{1}'e^{\frac{\epsilon_\mathrm{Q}}{T_0}}\big\},
        % \epsilon_\mathrm{Q}=\ln\frac{}{},\\
        \end{align}
    where, for conciseness,
        \begin{subequations}\small
        \begin{align}
        \label{psi}
        \psi&=(1-q)\psi_{1}+q\left(1-2Q\left(\frac{E_{\mathrm{L}}^{\mathrm{max}}}{\sigma_{\mathrm{DP}}}\right)\right) ;
        \\
\label{psi1}\psi_{1}&\!=\!Q\left(\!\frac{2C\!+\!3\sigma_{\mathrm{DP}}\!-\!E_{\mathrm{L}}^{\mathrm{max}}}{\sigma_{\mathrm{DP}}}\right)\!-\!Q\left(\frac{2C\!+\!3\sigma_{\mathrm{DP}}\!+\!E_{\mathrm{L}}^{\mathrm{max}}}{\sigma_{\mathrm{DP}}}\right);
\\
      \label{psi'}
        \psi'&=(1-q)\psi_{1}'+qQ\left(\frac{3\sigma_{\mathrm{DP}}-E_{\mathrm{L}}^{\mathrm{max}}}{\sigma_{\mathrm{DP}}}\right);
        \\
\label{psi1'}
\psi_{1}'&=Q\left(\frac{2C+3\sigma_{\mathrm{DP}}-E_{\mathrm{L}}^{\mathrm{max}}}{\sigma_{\mathrm{DP}}}\right),
    \end{align}
    \end{subequations}
    where $q$ is the mini-batch sampling rate, and $T_0$ is the maximum number of rounds in which each client can upload its FL local model due to privacy concerns.
    % , {and $E$ is the maximum quantization error.} %$E_{\mathrm{L}}^{\mathrm{max}}$ for brevity.}
    % (due to its privacy concern), 
    % and $\Phi(z)=\frac{1}{\sqrt{2\pi}}\int_{-\infty}^{z}e^{-\frac{z^{2}}{2}}dz$ is the cumulative distribution function (CDF) of a standard normal distribution with zero mean and standard deviation of 1.
    % By letting $\delta_\mathrm{Q}=0$, $\mathcal{M}_\mathrm{Q}$ satisfies $T_0\ln\frac{\psi}{\psi_{1}}$-DP.
\end{theorem}
\begin{proof}
    See \textbf{Appendix \ref{privacy_budget_proof}}.
\end{proof}

When $\delta_\mathrm{Q}=0$, $\mathcal{M}_\mathrm{Q}$ satisfies $T_0\max\{\ln\frac{\psi}{\psi_{1}},\ln\frac{\psi'}{\psi_{1}'}\}$-DP, which can be readily proved by substituting $\delta_\mathrm{Q}=0$ in \eqref{delta_Q} and obtaining $\epsilon_{\mathrm{Q}}=T_0\max\{\ln\frac{\psi}{\psi_{1}},\ln\frac{\psi'}{\psi_{1}'}\}$.

% According to \textbf{Theorem \ref{privacy_budget}}, 
By plugging (\ref{psi}) and (\ref{psi'}) into (\ref{delta_Q}), we have 
\begin{subequations}\small
\label{delta_Q0}
\begin{align}
    % \label{delta_Q1}
    \delta_\mathrm{Q}&=T_0\left[\psi_{1}\left(1-q-e^{\frac{\epsilon_\mathrm{Q}}{T_0}}\right)+q\left(1-2Q\left(\frac{E_{\mathrm{L}}^{\mathrm{max}}}{\sigma_{\mathrm{DP}}}\right)\right)\right] \nonumber\\
    \normalsize\text{or  }
    % \label{delta_Q1'}
\delta_\mathrm{Q}&=T_0\left[\psi_{1}'\left(1\!\!-\!\!q\!\!-\!\!e^{\frac{\epsilon_\mathrm{Q}}{T_0}}\right)\!\!+\!\!qQ\left(\frac{3\sigma_{\mathrm{DP}}-E_{\mathrm{L}}^{\mathrm{max}}}{\sigma_{\mathrm{DP}}}\right)\right]. \nonumber
\end{align}
\end{subequations}
With (\ref{psi1}) and (\ref{psi1'}), $\delta_\mathrm{Q}$ decreases with the increase of $\sigma_{\mathrm{DP}}$, because both $1-2Q(\frac{E_{\mathrm{L}}^{\mathrm{max}}}{\sigma_{\mathrm{DP}}})$ and $Q(\frac{3\sigma_{\mathrm{DP}}-E_{\mathrm{L}}^{\mathrm{max}}}{\sigma_{\mathrm{DP}}})$ decrease with the increase of $\sigma_{\mathrm{DP}}$ {(with $\psi_{1}$ and $\psi_{1}'$ close to zero). 
{
% Based on \textbf{Theorem \ref{privacy_budget}}, 
% $\delta_\mathrm{Q}$ decreases with the increase of $\sigma_{\mathrm{DP}}$. Therefore, the 
{The privacy budget $\epsilon_\mathrm{Q}$, the number of quantization bits $R$, and the clipping threshold $C$ are determined based on the specific requirements of the applications and the structure of models.
For applications where model performance is critical (e.g., autonomous vehicles or industrial systems), $R$ should be chosen for small communication overhead while maintaining an acceptable level of accuracy. The requirement of privacy can be relaxed. 
For applications where privacy protection is prioritized (e.g., healthcare or financial systems), $\epsilon_\mathrm{Q}$ and $\delta_\mathrm{Q}$ should be small.
$C$ can be chosen adaptively based on the model type to ensure accuracy.}
Given the maximum number of uploading rounds for each client $T_{0}$, the clipping threshold $C$, and the number of quantization bits $R$, $\sigma_{\mathrm{DP}}$ can be obtained through a one-dimensional search to satisfy the required privacy budget (i.e., $\epsilon_\mathrm{Q}$ and $\delta_\mathrm{Q}$).
}
}
\section{Convergence Analysis of WPFL}

 \begin{algorithm}[t]
	\caption{{WPFL with Privacy Protection}}
        \label{algorithm}
	\begin{algorithmic}[1]
	\Require $T_0$, $t=1$, $\{t_n=0\}_{n\in\mathcal{N}}$, $\lambda_n^0$, $\boldsymbol{\omega}^0$, $\{\mathbf{\boldsymbol{\varpi}}_{n}^0\}_{n\in\mathcal{N}}$, $N$, $\eta_{\mathrm{F},n}^{0}$, $\eta_{\mathrm{P},n}^{0}$, $\sigma_\mathrm{DP}$.
	\Ensure $\boldsymbol{\omega}^T$, $\{\mathbf{\boldsymbol{\varpi}}_{n}^T\}_{n\in\mathcal{N}}$.
	\While{$\{n|t_n\leq T_0,n\in \cal{N}\}\neq \emptyset$}
%    \STATE/ / oal training process for the global model;
        % \State $//$ \textit{Local training process for the FL model}
        \For {$n \in \mathcal{N}_{t}$}
                \State $//$ \textit{Local training process for the FL model}
        \State {Receive the FL global model $\hat{\boldsymbol{\omega}}_{n,\mathrm{G}}^{t}$;}
	\State Update the FL local model $\boldsymbol{u}_{n}^{t}$ by (\ref{update_lc_model}); 
        % \STATE $\boldsymbol{u}_{n}^{t}=\boldsymbol{u}_{n}^{t}-\eta_{\mathrm{F},n}^{t}\nabla F_{n}(\boldsymbol{u}_{n}^{t})$.
        \State Clip the local model {by \eqref{clipping};}
        % \State $\boldsymbol{u}_{n}^{t}=\boldsymbol{u}_{n}^{t}\big/\max\big(1,{\parallel \boldsymbol{u}_{n}^{t}\parallel}/{C}\big)$;
        \State Perturb, quantize, and upload the FL local models:
        \State $\tilde{\boldsymbol{u}}_{n}^{t}=\boldsymbol{u}_{n}^{t}+\mathbf{z}_{n}^{t}+\mathbf{E}_{n,\mathrm{L}}^{t}$;

        \State $//$ \textit{Local training process for the PL model}
        \State Update the PL model $\boldsymbol{\varpi}_{n}^{t}$ by (\ref{update_pl_model});
        % \STATE $\boldsymbol{\varpi}_{n}^{t}=\boldsymbol{\varpi}_{n}^{t-1}-\eta_{\mathrm{P},n}^{t}(\left(1-\frac{\lambda_{n}^{t}}{2}\right)\nabla F_{n}(\boldsymbol{\varpi}_{n}^{t-1})+\lambda_{n}^{t}(\boldsymbol{\varpi}_{n}^{t-1}-\hat{\boldsymbol{\omega}}_{n,\mathrm{G}}^{t}))$;
        {\State $t_n\leftarrow t_n+1$;}
        \EndFor
        \State $//$ \textit{FL model aggregation}
        \State Update the FL global model $\tilde{\boldsymbol{\omega}}_{\mathrm{L}}^{t}$ based on the received local models $\hat{\boldsymbol{\omega}}_{n}^{t} (\forall n \in \mathcal{N}_{t})$:
$\tilde{\boldsymbol{\omega}}_{\mathrm{L}}^{t}=\frac{1}{|\mathcal{N}_{t}|}{\sum}_{n\in\mathcal{N}_{t}}\hat{\boldsymbol{\omega}}_{n}^{t}$;
        \State Quantize the FL global model:
        $\tilde{\boldsymbol{\omega}}_{\mathrm{G}}^{t}=\tilde{\boldsymbol{\omega}}_{\mathrm{L}}^{t}+\mathbf{E}_{\mathrm{G}}^{t}$;
        % \STATE Receive the downloaded global model $\hat{\boldsymbol{\omega}}_{n,\mathrm{G}}^{t+1}$ by each client $n$;
        \State {$t\leftarrow t+1$}.
        \EndWhile
	\end{algorithmic}
	\end{algorithm}

 	\begin{figure}[!t]
	\centering
        \includegraphics[width=0.5\textwidth]{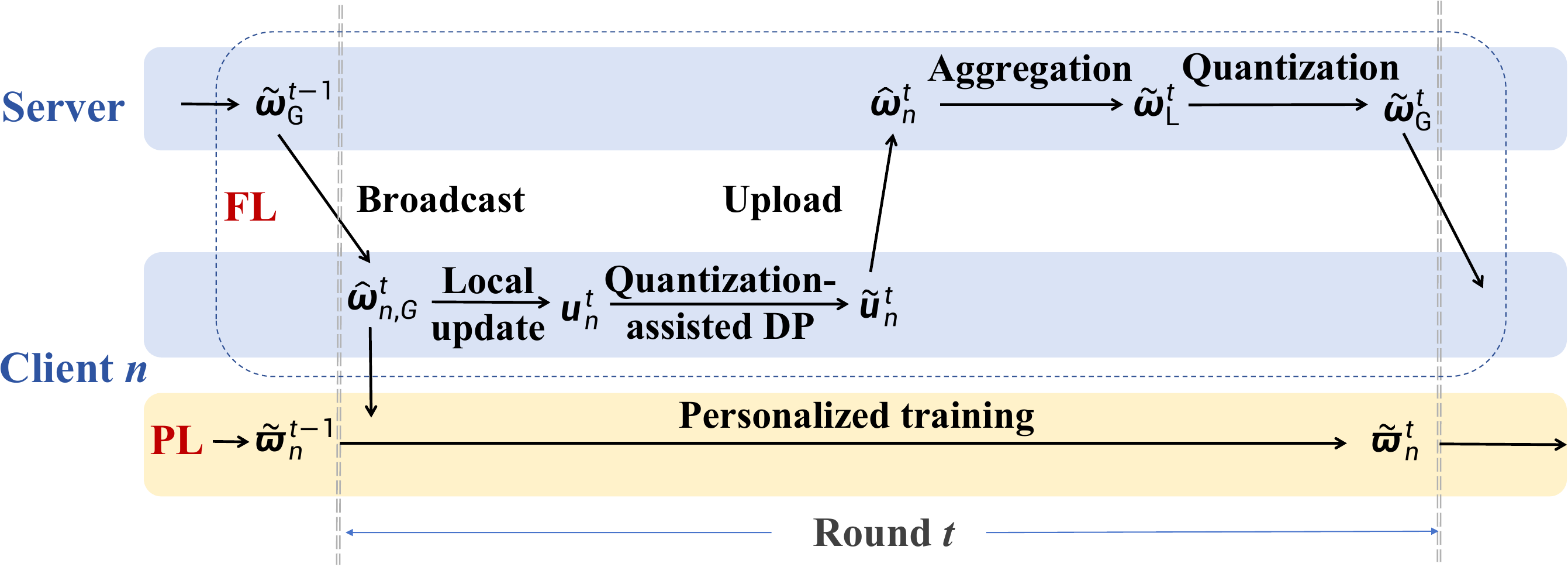}
	\caption{The timeline of WPFL in the $t$-th round.}
	\label{fig:timeline}
	\end{figure}

% \subsection{Modulation and Channel Model}

% \subsection{Convergence Analysis of Wireless PFL}

% This section analyzes the convergence of the WPFL under the following assumptions.
{This section analyzes the convergence of WPFL. 
% We start by presenting the key assumptions. 
Specifically, we first establish the convergence upper bound of the FL global model with the consideration of DP (see Section~V-A), followed by analyzing the per-round convergence of PL (see Section~V-B). Finally, the overall convergence of PL is attained (see Section~V-C). This starts with the following assumptions.}
\begin{assumption}
\label{assumption1}
$\forall n \in \mathcal{N}$, 
\begin{itemize}
    \item 
The local loss function of each client $n$, i.e., $F_{n}\left(\cdot\right)$, is $\mu$-strongly convex \cite{karimi2016linear} and $L$-smooth \cite{o2006metric}, i.e., $F\left(\boldsymbol{\omega}\right)-F\left(\boldsymbol{\omega}^{\ast}\right)\leq\frac{1}{2\mu}\parallel\nabla F\left(\boldsymbol{\omega}\right)\parallel^{2}$ and $\parallel\nabla F\left(\boldsymbol{\omega}\right)-\nabla F\left(\boldsymbol{\omega}'\right)\parallel\leq L\parallel\boldsymbol{\omega}-\boldsymbol{\omega}'\parallel$. Here, $\mu$ and $L$ are constants~\cite{wei2020federated};
%\item 
% $F\left(\boldsymbol{\omega}^{0}\right)-F\left(\boldsymbol{\omega}^{\ast}\right)=\Psi_{1}$, and $\parallel\boldsymbol{\varpi}_{n}^{0}-\boldsymbol{\varpi}_{n}^{\ast}\parallel^{2}=\Psi_{2}$;
% \item 
% The global learning rate $\eta_{\mathrm{F},n}^{t}\leq\frac{2}{L}$, and $\mu>\frac{2-2C+3\sigma_{\mathrm{DP}}_{n}^{t+1}}{2-\lambda_{n}^{t+1}}$;
\item The expectation of stochastic gradients is uniformly bounded at each client $n$ and each round $t$, i.e.,  
$\mathbb{E}\left[\parallel\nabla F_{n}(\boldsymbol{\omega}^{t})\parallel^{2}\right]\leq G_{0}^{2}$ \cite{li2021ditto};
\item The $L_2$-norm between the optimal FL local models and the optimal FL global model is bounded, i.e.,
$\| \boldsymbol{u}_{n}^{\ast}-\boldsymbol{\omega}^{\ast}\|\leq M$, where $M$ is a constant \cite{li2021ditto}.

\end{itemize}
\end{assumption}

%\subsection{Convergence Analysis}
\subsection{Convergence of FL}
The convergence upper bound of the FL global model with DP is established under \textbf{Assumption \ref{assumption1}}, as follows.

\begin{lemma}
\label{Lemma0_0}
Given the FL learning rate $\eta_{\mathrm{F},n}^{t}$ of client $n$, under \textbf{Assumption \ref{assumption1}}, the expected difference between the FL global model $\tilde{\boldsymbol{\omega}}_{\mathrm{L}}^{t}$ aggregated at the server and the optimal FL global model $\boldsymbol{\omega}^{\ast}$ 
% after updating the personalized model 
at the $t$-th round is upper-bounded by
       %\begin{equation} 
       {\small\begin{align} 
        \label{Glb_Up_OneStep}
\mathbb{E}&\left[\parallel\tilde{\boldsymbol{\omega}}_{\mathrm{L}}^{t}-\boldsymbol{\omega}^{\ast}\parallel^{2}\right]\!\!\leq\!\! \left(1\!+\!\frac{1}{\phi_{1}}\!+\!\frac{1}{\phi_{2}}\right) \left[\Theta_{\mathrm{L}}^{t}\!+\!|\boldsymbol{\omega}|\!\left(\sigma_{\mathrm{DP}}^{2}\!+\!\left(E_{\mathrm{L}}^{\mathrm{max}}\right)^{2}\right)\!\right] \nonumber \\
&{+\!\frac{1}{|\mathcal{N}_{t}|}\!\!\underset{n\in\mathcal{N}_{t}}{\sum} \!\!\!\left[\!\left(1\!\!+\!\!\phi_{2}\right)\!\!+\!\!\left(1\!\!+\!\!\phi_{1}\right)\!\!L^{2}\!\!\left(\!\eta_{\mathrm{F},n}^{t}\!\right)^{2}\!\!\!-\!\!u\eta_{\mathrm{F},n}^{t}\!\right]\!\!\mathbb{E}\!\left[\!\parallel\!\!\hat{\boldsymbol{\omega}}_{n,\mathrm{G}}^{t}\!\!-\!\!\boldsymbol{\omega}^{\ast}\!\!\parallel^{2}\right]\!\!,}
        \end{align}}%
	%\end{equation}
 where $\phi_{1}>0$ and $\phi_{2}>0$ can be any positive constants, and $\Theta_{\mathrm{L}}^{t}=\frac{2C^{2}+\left(2-\beta_{\mathrm{L}}^{2}\right)|\boldsymbol{\omega}|(C+3\sigma_{\mathrm{DP}})^{2}-|\boldsymbol{\omega}|\sigma_{\mathrm{DP}}^{2}}{|\mathcal{N}_{t}|}{\sum}_{n\in\mathcal{N}_{t}}\rho_{n,\mathrm{L}}^{t}$.
\end{lemma}
\begin{proof}
    See \textbf{Appendix \ref{Lemma0_0_proof}}.
\end{proof}

\begin{theorem}
    \label{theorem1}
    Given the FL local learning rate $\eta_{\mathrm{F},n}^{t}$ of client $n$, under \textbf{Assumption \ref{assumption1}}, the expected difference between the FL global model $\hat{\boldsymbol{\omega}}_{n,\mathrm{G}}^{t+1}$ received by client $n$
at the $(t+1)$-th round and the optimal FL global model $\boldsymbol{\omega}^{\ast}$ is upper-bounded:
{
\begin{subequations} \small
    \begin{align}
        \label{down_glb_di}
        \mathbb{E}&\left[\parallel\!\!\hat{\boldsymbol{\omega}}_{\mathrm{n,G}}^{t+1}\!\!-\!\!\boldsymbol{\omega}^{\ast}\!\!\parallel^{2}\right]\!\! \leq \!\!\frac{1}{|\mathcal{N}_{t}|}\!\!{\sum}_{n\in\mathcal{N}_{t}} \!\!\varepsilon_{\mathrm{F},n}^{t}\mathbb{E}\left[\parallel\!\!\hat{\boldsymbol{\omega}}_{n,\mathrm{G}}^{t}\!\!-\!\!\boldsymbol{\omega}^{\ast}\!\!\parallel^{2}\right] \!\!+\!\!\Gamma_{t+1}\\
        &\leq \!\!\left(\varepsilon_{\mathrm{F}}^{\max}\right)^{t\!+\!1}\mathbb{E}\!\left[\parallel\hat{\boldsymbol{\omega}}_{\mathrm{n,G}}^{0}\!\!-\!\!\boldsymbol{\omega}\parallel^{2}\right]\!\!+\!\!\frac{(\varepsilon_{\mathrm{F}}^{\max})^{t\!+\!1}\!-\!1}{\varepsilon_{\mathrm{F}}^{\max}-1}\Gamma^{\max},
    \label{down_glb_overall}
        % \nonumber\\
        % &\!\!+\!\!h_1(\rho_{n,\mathrm{G}}^{t+1})\Theta_{\mathrm{L}}^{t}\!+\!\Gamma_{0}\rho_{n,\mathrm{G}}^{t+1}+\Gamma_{1} 
    \end{align}
\end{subequations}}%
where, for brevity, we define
\begin{subequations} \small
    \label{eq_con_FL}
    \begin{align}
    \label{Gamma}
    \Gamma_{t+1}&\triangleq h_{1}(\rho_{n,\mathrm{G}}^{t+1})\Theta_{\mathrm{L}}^{t}+\Gamma_{0}\rho_{n,\mathrm{G}}^{t+1}+\Gamma_{1},
    \\
        \label{vareps_F}
\varepsilon_{\mathrm{F},n}^t\!&\triangleq\!\!\left(\!1\!\!+\!\!\varphi_{1}\right)\!\!\left(\!\left(1\!\!+\!\!\phi_{2}\right)\!\!+\!\!\left(1\!\!+\!\!\phi_{1}\right)\!L^{2}\!\left(\eta_{\mathrm{F},n}^{t}\right)^{2}\!\!-\!\!\mu\eta_{\mathrm{F},n}^{t}\!\right)\!,\\
\label{h1}
        {h_1}(\rho_{n,\mathrm{G}}^{t+1})&\!\triangleq\!2\!\left(\!1\!\!+\!\!\frac{1}{\varphi_{1}}\!\right)\!\left(\!1\!\!+\!\!\varphi_{2}\!\right)\!\rho_{n,\mathrm{G}}^{t+1}\!\!+\!\!\left(\!1\!\!+\!\!\varphi_{1}\!\right)\!\left(\!1\!\!+\!\!\frac{1}{\phi_{1}}\!\!+\!\!\frac{1}{\phi_{2}}\!\right) \!,
 \\
 \begin{split}
\Gamma_0&\triangleq\left(1+\frac{1}{\varphi_{1}}\right)\left(2\left(1+\frac{1}{\varphi_{2}}\right)C^{2}+2|\boldsymbol{\omega}|\left(1+\varphi_{2}\right)\cdot\right.\\
&\left.\left(\sigma_{\mathrm{DP}}^{2}\!\!+\!\!\left(E_{\mathrm{L}}^{\mathrm{max}}\right)^{2}\right)\!\!+\!\!2|\boldsymbol{\omega}|(C^2\!\!-\!\!\left(E_{\mathrm{L}}^{\mathrm{max}}\right)^{2})\right),
 \end{split}
     \\
     \begin{split}
\Gamma_{1}&\triangleq|\boldsymbol{\omega}|\left(1+\varphi_{1}\right)\left(1+\frac{1}{\phi_{1}}+\frac{1}{\phi_{2}}\right)\left(\sigma_{\mathrm{DP}}^{2}+\left(E_{\mathrm{L}}^{\mathrm{max}}\right)^{2}\right)\\
&+2|\boldsymbol{\omega}|\left(1+\frac{1}{\varphi_{1}}\right)\left(E_{\mathrm{G}}^{\mathrm{max}}\right)^{2},
     \end{split}
    \end{align}
\end{subequations}
where {$\varepsilon_{\mathrm{F}}^{\max}\in (0,1)$} and  $\Gamma^{\max}$ are the maxima of $\varepsilon_{\mathrm{F},n}^{t}$ and $\Gamma_{t+1}$, respectively; $\varphi_{1}$ and $\varphi_{2}$ are any positive constants.
\end{theorem}
\begin{proof}
    See \textbf{Appendix \ref{theorem1_proof}}.
\end{proof}

    According to \textbf{Theorem \ref{theorem1}}, the right-hand side (RHS) of (\ref{down_glb_di}) increases with $\rho_{n,\mathrm{G}}^{t+1}$, $\Theta_{\mathrm{L}}^{t}$, $\sigma_{\mathrm{DP}}$, and $E_{\mathrm{L}}^{\mathrm{max}}$.
    In other words, the DP noise, imperfect wireless channel, and quantization degrade the convergence of the FL global model.

\subsection{Per-Round Convergence of PL}
Next, we analyze the convergence of the PL models under \textbf{Assumption \ref{assumption1}}, as follows. 

\begin{theorem}
    \label{theorem_PL_con}
    Given the PL learning rate $\eta_{\mathrm{P},n}^{t+1}$, the FL local learning rate $\eta_{\mathrm{F},n}^{t}$, and the weighting coefficient $\lambda_{n}^{t+1}$ of client $n$, under \textbf{Assumption \ref{assumption1}}, the expected difference between the PL model $\tilde{\boldsymbol{\varpi}}_{n}^{t+1}$ at the $(t+1)$-th round and the optimal PL model $\boldsymbol{\varpi}_n^{\ast}$ 
% after updating the personalized model 
is upper-bounded by
{
\begin{align}
\label{PL_con}
% \begin{aligned}
    \mathbb{E}&\left[\parallel\tilde{\boldsymbol{\varpi}}_{n}^{t+1}\!\!-\!\!\boldsymbol{\varpi}_{n}^{\ast}\parallel^{2}\right]\leq \varepsilon_{\mathrm{P},n}^{t+1}\mathbb{E}\left[\parallel\tilde{\boldsymbol{\varpi}}_{n}^{t}\!\!-\!\!\boldsymbol{\varpi}_{n}^{\ast}\parallel^{2}\right] \!\!+\!\! \Phi_n^{t+1}
,
% \end{aligned}
\end{align}}
where, for brevity, we define  
\begin{subequations}\small
\begin{align}
\label{varepsilon_P}
\varepsilon_{\mathrm{P},n}^{t+1}&\triangleq1\!\!-\!\!\eta_{\mathrm{P},n}^{t+1}\left(\left(1\!\!-\!\!{\lambda_{n}^{t+1}}\big/{2}\right)\mu\!\!+\!\!\lambda_{n}^{t+1}\right)\!\!+\!\!\left(\eta_{\mathrm{P},n}^{t+1}\right)^{2};\\
% \label{h0}
% {h_0}\!&\left(\eta_{\mathrm{F},n}^{t},\eta_{\mathrm{P},n}^{t+1},\lambda_{n}^{t+1}\right)\!=\!\!\Psi_{n}^{t+1}\varepsilon_{\mathrm{F},n}^{t},\\
\label{Psi_n}
\Psi_{n}^{t+1}&\triangleq\left(\left(\eta_{\mathrm{P},n}^{t+1}\right)^{2}+1\right)\left(\lambda_{n}^{t+1}\right)^{2}+{\left(\eta_{\mathrm{P},n}^{t+1}\right)^{3}}\big/{\lambda_{n}^{t+1}};
\\
     \begin{split}  
\Phi_n^{t+1}&\triangleq\left(1\!+\!\left(\lambda_{n}^{t+1}\right)^{3}\right)\!\left(\eta_{\mathrm{P},n}^{t+1}\right)^{2}\!\!G_n^{t+1} \!+ \!\!\Psi_{n}^{t+1}\!\Big(\Gamma_{t+1}\!\\
&{+\frac{(G_{0}^{2}+M\mu)^{2}}{|\mathcal{N}_{t}|\mu^{2}}\!\!{\sum}_{n\in\mathcal{N}_{t}} \varepsilon_{\mathrm{F},n}^t\Big)} ;
     \end{split}
     \label{Phi}\\
\label{G_n}
G_n^{t+1}&\triangleq\left(\left(1-{\lambda_{n}^{t+1}}\big/{2}\right)G_{0}+\lambda_{n}^{t+1}({G_{0}}\big/{\mu}+M)\right)^{2}.
    \end{align}
    \label{eq30}
\end{subequations}
\end{theorem}
\begin{proof}
    See \textbf{Appendix \ref{theorem_PL_con_proof}}.
\end{proof}

It is revealed in \textbf{Theorem \ref{theorem_PL_con}} that the convergence of PL is degraded by the DP noise (i.e., $\sigma_{\mathrm{DP}}$), imperfect wireless channel (i.e., $\rho_{n,\mathrm{G}}^{t+1}$ and $\Theta_{\mathrm{L}}^{t}$), and quantization errors (i.e., $E_{\mathrm{L}}^{\mathrm{max}}$ and $E_{\mathrm{G}}^{\mathrm{max}}$).
{ 
Based on \textbf{Theorems \ref{privacy_budget}} and \textbf{\ref{theorem_PL_con}}, the DP noise (i.e., $\sigma_{\mathrm{DP}}$), while ensuring the privacy budget, compromises the convergence of PL. The effect of the DP noise on the convergence of PL (i.e., $\Phi_n^{t+1}$, $\forall n \in \mathcal{N}$) depends on the uplink and downlink channel conditions of the selected clients (i.e., $\rho_{n,\mathrm{L}}^{t}$ and $\rho_{n,\mathrm{G}}^{t+1}$), and can vary among clients. This also degrades the min-max fairness of WPFL.}
% {Based on \textbf{Theorem \ref{privacy_budget}}, 
% $\delta_\mathrm{Q}$ decreases with the increase of $\sigma_{\mathrm{DP}}$. Therefore, the privacy budget $\epsilon_\mathrm{Q}$, the number of quantization bits $R$, and the clipping threshold $C$ are determined based on the specific requirements of the applications and the structure of models.
% For applications where model performance is critical (e.g., autonomous vehicles or industrial systems), $R$ should be chosen for small communication overhead while maintaining an acceptable level of accuracy. The requirement of privacy can be relaxed. 
% For applications where privacy protection is prioritized (e.g., healthcare or financial systems), $\epsilon_\mathrm{Q}$ and $\delta_\mathrm{Q}$ should be small.
% $C$ should be chosen adaptively based on the model type to ensure accuracy. }

{
The impact of $\lambda_{n}^{t+1}$ and $\eta_{\mathrm{P},n}^{t+1}$ on the convergence of PL is intricate. 
    {
    In particular, $\varepsilon_{\mathrm{P},n}^{t+1}$ decreases monotonically with $\eta_{\mathrm{P},n}^{t+1}\in(0,1)$ when $\mu\geq 2$, and first decreases and then increases when $\mu < 2$. 
    Moreover, $\varepsilon_{\mathrm{P},n}^{t+1}$ increases with $\lambda_{n}^{t+1}\in (0,2)$ when $\mu>2$, decreases when $\mu<2$, and is unaffected by $\lambda_{n}^{t+1}$ when $\mu=2$.
Across the entire possible range of $\mu$, $\Psi_{n}^{t+1}$ increases with $\eta_{\mathrm{P},n}^{t+1}$, and first increases and then decreases with $\lambda_{n}^{t+1}$.     
     $\Phi_n^{t+1}$ increases with $\eta_{\mathrm{P},n}^{t+1}$, but its dependence on $\lambda_{n}^{t+1}$ is more complex, subject to the values of $M$, $G_0$, $\mu$, and $h_1(\rho_{n,\mathrm{G}}^{t+1})\Theta_{\mathrm{L}}^{t}\!+\!\Gamma_{0}\rho_{n,\mathrm{G}}^{t+1}+\Gamma_{1}$; see~(\ref{Phi}). Nevertheless, there is an opportunity to minimize the maximum of $\Phi_n^{t+1}$, $\forall n \in \mathcal{N}$, while keeping $\varepsilon_{\mathrm{P},n}^{t+1}$ consistent across all clients by optimizing $\mathbf{c}^{t}$, $\mathbf{p}^{t}$, $\boldsymbol{\eta}_{\mathrm{F}}^{t}$, $\boldsymbol{\eta}_{\mathrm{P}}^{t+1}$, and $\boldsymbol{\lambda}^{t+1}$. This encourages the convergence of PL while maintaining performance fairness among the clients.
     }
     }

\subsection{Overall Convergence of PL}
Let $T$ be the maximum number of communication rounds satisfying ${\{n|t_n\leq T_0,n\in \cal{N}\}\neq \emptyset}$. We analyze the overall convergence upper bound of PL under imperfect channels after $T$ aggregation rounds, as follows.
\begin{theorem}
\label{Convergence_t2}
Under \textbf{Assumption \ref{assumption1}}, the convergence upper bound of PL under imperfect channels after $T$ aggregation rounds is given by
{
% \begin{subequations}
    \begin{equation} \small
\mathbb{E}\!\left[\!\parallel\!\!\tilde{\boldsymbol{\varpi}}_{n}^{T}\!\!-\!\!\boldsymbol{\varpi}_{n}^{\ast}\!\!\parallel^{2}\!\right] \!\leq\! 
\left(\!\varepsilon_{\mathrm{P}}^{\max}\!\right)^{T}\!\mathbb{E}\!\!\left[\!\parallel\!\!{\boldsymbol{\varpi}}_{n}^{0}\!\!-\!\!\boldsymbol{\varpi}_{n}^{\ast}\!\!\parallel^{2}\!\right]\!\!+\!\!\frac{(\varepsilon_{\mathrm{P}}^{\max})^{T}\!\!-\!\!1}{\varepsilon_{\mathrm{P}}^{\max}\!-\!1}\Phi^{\max},
    \label{T_convergence}
    \end{equation}}%
where {$\varepsilon_{\mathrm{P}}^{\max} \in (0,1)$} and
$\Phi^{\max}$ are the upper bounds of $\varepsilon_{\mathrm{P},n}^{t}$ and $\Phi_{n}^{t}$, $\forall n \in \cal{N}$, $t=0,1,\ldots,T$, respectively.
% {
% $A_{i}$ is defined for brevity, as given by
%     \begin{align}
%         \label{AT}
% A_{i}\!\!=\!\!\left\{
% \begin{aligned}
% &\frac{\left(\varepsilon_{\mathrm{P}}^{\max}\right)^{i+1}\!\!-\!\!\left(\varepsilon_{\mathrm{F}}^{\max}\right)^{i+1}}{\varepsilon_{\mathrm{P}}^{\max}-\varepsilon_{\mathrm{F}}^{\max}}, \text{if } \varepsilon_{\mathrm{F}}^{\max}\neq\varepsilon_{\mathrm{P}}^{\max};\\
%         % \label{A'T}
%         &(i+1)\left(\varepsilon_{\mathrm{P}}^{\max}\right)^{i}, \text{if } \varepsilon_{\mathrm{F}}^{\max}=\varepsilon_{\mathrm{P}}^{\max}.
%         \end{aligned}
%         \right.
%     \end{align}
% }
\end{theorem}

\begin{proof}
    See \textbf{Appendix \ref{Convergence_t2_proof}}.
\end{proof}

% { 
% Based on \textbf{Theorem \ref{privacy_budget}}, the DP noise (i.e., $\sigma_{\mathrm{DP}}$), while ensuring the privacy budget, compromises the convergence of PL. The effect of the DP noise on the convergence of PL (i.e., $\Phi_n^{t+1}$, $\forall n \in \mathcal{N}$) depends on the uplink and downlink channel conditions of the selected clients (i.e., $\rho_{n,\mathrm{L}}^{t}$ and $\rho_{n,\mathrm{G}}^{t+1}$), and can vary among clients. This also degrades the min-max fairness of the WPFL system under consideration.}

% According to \textbf{Theorem \ref{Convergence_t2}}, the overall convergence upper bound of PL is degraded by the upper bound of $\Gamma_{t+1}$, which increases with the upper bounds of $\rho_{n,\mathrm{G}}^{t+1}$, $\Theta_{\mathrm{L}}^{t}$, $\sigma_{\mathrm{DP}}$, and $E_{\mathrm{L}}^{\mathrm{max}}$. On the other hand, the impact of $\lambda_{n}^{t+1}$, $\eta_{\mathrm{F},n}^{t+1}$, and $\eta_{\mathrm{P},n}^{t+1}$ on the convergence is much more complex due to the complex coupling among them. 
% { Nevertheless, this convergence upper bound of PL can be further optimized by minimizing the RHS of (\ref{PL_con}) at each communication round, as will be discussed in Section VI.}

\section{Optimal Configuration and Scheduling Policy}
To accelerate the convergence of WPFL in a fair fashion, this section minimizes the maximum per-round convergence upper bound of all PL models, as the channels and subsequently the device selections change randomly across rounds.
{This starts by formulating a min-max problem and converting it to a max-min problem (see Section VI-A). By revealing its nested structure, the max-min problem is solved first through client selection, channel allocation, and power control (see Section~VI-B), followed by learning rate and weighting coefficient adjustment (see Section VI-C). The proposed algorithm and its complexity analysis are presented in Section~VI-D.}

\subsection{Problem Formulation}

According to (\ref{PL_con}), the convergence upper bound is dominated by $\Phi_n^{t+1}$, while the convergence rate is determined by $\varepsilon_{\mathrm{P},n}^{t+1}$. For this reason, we minimize $\underset{n\in \mathcal{N}}{\max}\,\Phi_n^{t+1}$, while keeping $\varepsilon_{\mathrm{P},n}^{t+1}$, $\forall n \in \cal{N}$ consistent across the clients, as follows.
\begin{subequations}\small
\label{optimizationP}
    \begin{align}
&\textbf{P}:  \min_{\mathbf{c}^{t},\mathbf{p}^{t},\boldsymbol{\eta}_{\mathrm{F}}^{t},\boldsymbol{\eta}_{\mathrm{P}}^{t+1},\boldsymbol{\lambda}^{t+1}} \max_{n\in\mathcal{N}} \Phi_n^{t+1} \\
\label{P_C1}
\textrm {s.t. }  \quad&\textbf{C1}: \varepsilon_{\mathrm{P},n}^{t+1}=\varepsilon_{\mathrm{P}}^{t+1},\forall n\in{\cal {N}}, \\
% \label{P_C2}
% & \textbf{C2}: h_{0}\left(\eta_{\mathrm{F},n}^{t},\eta_{\mathrm{P},n}^{t+1},\lambda_{n}^{t+1}\right)=h_{0}^{t},\forall n\in{\cal {N}}, \\
\label{P_C3}
& \textbf{C2}: {\sum}_{k=1}^{K}c_{n,k}^{t}\leq1,\,\forall n\in\mathcal{N}, \\
& \textbf{C3}: {\sum}_{n=1}^{N}c_{n,k}^{t}\leq1,\,\forall k\in\mathcal{K}, \\
& \textbf{C4}:
P_{n}^{t}\leq P_{n}^{\mathrm{th}},\,\forall n\in\mathcal{N}, \\
\label{P_C6}
& \textbf{C5}:
r_{n,k}^{t}\geq r_{\mathrm{min}},\forall k\in\mathcal{K},\\
& \textbf{C6}:
c_{n,k}^{t}\in\{0,1\},\forall n\in\mathcal{N},k\in\mathcal{K},\\
\label{P_C8}
& \textbf{C7}:
{\sum}_{t'=1}^{t}{\sum}_{k=1}^{K}c_{n,k}^{t'}\leq T_{0},\,\forall n\in\mathcal{N},\\
% \label{P_C9}
% & \textbf{C9}:
% 0<\varepsilon_{\mathrm{P},n}^{t}<1,\forall n\in\mathcal{N},\\
\label{P_C9}
& \textbf{C8}:
0<\lambda_{n}^{t+1}<2,\forall n\in\mathcal{N},\\
\label{P_C10}
& \textbf{C9}:
0<\eta_{\mathrm{P},n}^{t+1}<1,\forall n\in\mathcal{N},\\
\label{P_C11}
& \textbf{C10}:
0<\eta_{\mathrm{F},n}^{t}<1,\forall n\in\mathcal{N},\\
\label{P_C12}
& \textbf{C11}:
0<\varepsilon_{\mathrm{F},n}^{t}<1,\forall n\in\mathcal{N},
    \end{align}
\end{subequations}
where $\boldsymbol{\eta}_{\mathrm{F}}^{t}=\{{\eta}_{\mathrm{F},1}^{t},\ldots,{\eta}_{\mathrm{F},N}^{t}\}$ collects the FL learning rates of the clients at the $t$-th round, $\boldsymbol{\eta}_{\mathrm{P}}^{t+1}=\{{\eta}_{\mathrm{P},1}^{t+1},\ldots,{\eta}_{\mathrm{P},N}^{t+1}\}$ collects the PL learning rates of the clients at the $(t+1)$-th round, and $\boldsymbol{\lambda}^{t+1}=\{{\lambda}_1^{t+1},\ldots,{\lambda}_N^{t+1}\}$ collects the weighting coefficients of the clients at the $(t+1)$-th round; $\varepsilon_{\mathrm{P}}^{t+1}$ can differ between slots; $P_{n}^{\mathrm{th}}$ is the maximum transmit power of client~$n$. 

Constraints \textbf{C1} guarantee the consistent convergence rates among the clients. \textbf{C2} and \textbf{C3} specify that at most one subchannel is allocated to a client, and a subchannel is only allocated to one client. \textbf{C4} indicates that the transmit power of client $n$ is upper bounded by $P_{n}^{\mathrm{th}}$. \textbf{C5} indicates that the data rate of each client needs to exceed $r_{\mathrm{min}}$, {to ensure that its transmission delay does not exceed the maximum transmission delay $\tau_{\mathrm{max}}$.}
\textbf{C7} specifies that the number of rounds each client can participate in is no more than $T_0$ to meet the privacy requirement.
\textbf{C6} and \textbf{C8} -- \textbf{C11} are self-explanatory. 

{According to~\eqref{eq_con_FL}, \eqref{eq30}, and $\Theta_{\mathrm{L}}^{t}>0$ in \textbf{Lemma \ref{Lemma0_0}}, the objectives $\Phi_n^{t+1}$, $\forall n \in \mathcal{N}$, increase monotonically with $\Theta_{\mathrm{L}}^{t}$, which is consistent across all clients and only affected by the selected subchannels for the clients, $\mathbf{c}^{t}$, and the transmit powers of all clients, $\mathbf{p}^{t}$. The remaining terms in $\Phi_n^{t+1}$ are affected by the FL learning rate $\eta_{\mathrm{F},n}^{t}$, the PL learning rate $\eta_{\mathrm{P},n}^{t}$ and the weighting coefficient $\lambda_{n}^{t+1}$. In other words, $\eta_{\mathrm{F},n}^{t}\in\boldsymbol{\eta}_{\mathrm{F}}^{t}$, $\eta_{\mathrm{P},n}^{t}\in\boldsymbol{\eta}_{\mathrm{P}}^{t+1}$ and $\lambda_{n}^{t+1}\in\boldsymbol{\lambda}^{t+1}$ are independent of each other, and have no impact on $\Phi_{n'}^{t+1}$, $n'\neq n$. This} min-max problem can be converted into the following max-min problem:
\begin{align}
\label{P1}
    % \begin{align}
\textbf{P1}:  \quad\max_{n\in\mathcal{N}} \min_{\mathbf{c}^{t},\mathbf{p}^{t},\boldsymbol{\eta}_{\mathrm{F}}^{t},\boldsymbol{\eta}_{\mathrm{P}}^{t+1},\boldsymbol{\lambda}^{t+1}} \Phi_n^{t+1},\quad
\textrm {s.t. } \,\eqref{P_C1}-\eqref{P_C12},\nonumber
% \end{align}
\end{align}
which is a mixed integer program and is still challenging.
% In the following subsections, we decouple this problem {losslessly} into two sub-problems for client selection, channel allocation, and power control, and for learning rate and weighting coefficient adjustment, and obtain the optimal solutions to the subproblems.

According to {(\ref{Gamma}) and }(\ref{Phi}),  $\Theta_{\mathrm{L}}^{t}$ is consistent across all clients and depends only on client selection, channel allocation, and power control. Given $\Theta_{\mathrm{L}}^{t}$, for each client $n$, the second term on the RHS of (\ref{PL_con}), $\Phi_n^t$, only depends on the FL and PL learning rates {$\eta_{\mathrm{F},n}^{t}$} and $\eta_{\mathrm{P},n}^{t+1}$, and the weighting coefficient $\lambda_n^{t+1}$. Therefore, Problem \textbf{P1} is a nested optimization problem and can be equivalently decoupled into two subproblems solved sequentially. The first subproblem minimizes $\Theta_{\mathrm{L}}^{t}$ through client selection, channel allocation, and power control. Given $\Theta_{\mathrm{L}}^{t}$, the second subproblem maximizes the minimum of $\Phi_{n}^{t+1}$, {through the learning rate and weighting coefficient adjustment for optimizing $\Phi_{n}^{t+1}$ for any client $n$, followed by taking the maximum $\Phi_{n}^{t+1}$ among all clients.}

\subsection{Client Selection, Channel Allocation, and Power Control}
\label{client_selection}
%\subsection{Robustness of Privacy-Preserving FL}
We first minimize $\Phi_n^t$ per client $n \in \cal{N}$ by minimizing $\Theta_{\mathrm{L}}^{t}$:
% \begin{subequations}
% \label{P2}
    \begin{align}
\textbf{P2}:  \quad\min_{\mathbf{c}^{t},\mathbf{p}^{t}}  \;\Theta_{\mathrm{L}}^{t},\quad
\textrm {s.t. }\,\eqref{P_C3}-\eqref{P_C8} \nonumber.
\end{align}
% \end{subequations}
Clearly, the data rate $r_{n,k}^{t}$ increases monotonically with $P_{n}^{t}\leq P_{n}^{\mathrm{th}}$, while the element error probability $\rho_{n,\mathrm{L}}^{t}$ decreases.
Therefore, the optimal transmit power of client $n \in \mathcal{N}_{t}$ is $P_{n}^{t}=P_{n}^{\mathrm{th}}$. Let $\Upsilon_{n}^{t}$ denote the number of rounds that client $n$ is selected for uploading before round $t$, and $\mathcal{N}_{t}^{\mathrm{a}}$ collect the clients allowed to upload local models at round $t$, i.e., $\mathcal{N}_{t}^{\mathrm{a}}=\{n|n\in \mathcal{N}, \Upsilon_{n}^{t}<T_0\}$. Problem \textbf{P2} can be rewritten as
\begin{subequations}\small
\label{P3}
    \begin{align}
\textbf{P3}: \quad & \min_{\mathbf{c}^{t}}\;{\sum}_{n\in\mathcal{N}}\rho_{n,\mathrm{L}}^{t}\\
\textrm {s.t. }  \quad & {\sum}_{k=1}^{K}c_{n,k}^{t}\leq1,\,\forall n\in\mathcal{N}_{t}^{\mathrm{a}},\\
% &\sum_{k=1}^{K}c_{n,k}^{t}=0,\,\forall n\in\mathcal{N}/\mathcal{N}_{t+1}^{\mathrm{a}},\\
&{\sum}_{n=1}^{N}c_{n,k}^{t}\leq1,\,\forall k\in\mathcal{K},\\
&c_{n,k}^{t}\in\{0,1\},\forall n\in\mathcal{N}_{t}^{\mathrm{a}},k \in \mathcal{K},\\
&\eqref{P_C6}. \nonumber
\end{align}
\end{subequations}
Problem \textbf{P3} can be interpreted as a maximum-weight matching problem in bipartite graphs, which can be optimally solved using the Kuhn-Munkres (KM) algorithm\cite{kuhn1955hungarian}. At round $t$, the minimum of $\Theta_{\mathrm{L}}^{t}$, denoted as $\Theta_{\mathrm{L},\min}^{t}$, is obtained with the optimal client selection, channel allocation, and power control.

\subsection{Learning Rate and Weighting Coefficient Adjustment}
\label{co_adjustment}
Given $\Theta_{\mathrm{L},\min}^{t}$, $\Phi_{n}^{t+1}$ is rewritten as
\begin{equation} \small
    % \begin{split}
        \label{Phi_n}
        \scalebox{0.92}{$\Phi_{n}^{t\!+\!1}\!\!=\!\!\left(\!1\!\!+\!\!\left(\!\lambda_{n}^{t\!+\!1}\!\right)^{3}\!\right)\!\!\left(\!\eta_{\mathrm{P},n}^{t+1}\!\right)^{2}\!\!G_{n}^{t\!+\!1}\!\!+\!\!\Psi_{n}^{t\!+\!1}\!\Big(\!\Gamma_{2}\rho_{n,\mathrm{G}}^{t\!+\!1}\!\!+\!\!\Gamma_{3}\!\!+\!\!\frac{(\!G_{0}^{2}\!+\!M\mu\!)^{2}}{|\mathcal{N}_{t}|\mu^{2}}\!\!\!\underset{n\in\mathcal{N}_{t}}{\sum} \!\!\varepsilon_{\mathrm{F},n}^t\!\Big)\!,$}
    % \end{split}
\end{equation}%
where, for the brevity of notation, $\Gamma_{2}$ and $\Gamma_{3}$ are defined as 
\begin{subequations}\small
\begin{align}
    \label{Gamma2}
    \Gamma_{2}&\triangleq2\left(1+\frac{1}{\varphi_{1}}\right)\left(1+\varphi_{2}\right)\Theta_{\mathrm{L},\min}^{t}+\Gamma_{0};\\
% \end{equation}
% \begin{equation}
    \label{Gamma3}
    \Gamma_{3}&\triangleq\left(1+\varphi_{1}\right)\left(1+\frac{1}{\phi_{1}}+\frac{1}{\phi_{2}}\right)\Theta_{\mathrm{L},\min}^{t}+\Gamma_{1}.
    \end{align}
\end{subequations}
{
According to (\ref{Phi_n}), $\Phi_{n}^{t+1}$ increases with $\underset{n\in \mathcal{N}} {\sum}\varepsilon_{\mathrm{F},n}^t$, where the latter depends only on $\eta_{\mathrm{F},n}^t$, $\forall n \in \mathcal{N}$ and is consistent {among the clients}. 
For any client $n$, $\Phi_{n}^{t+1}$ is independent of $\boldsymbol{\eta}_{\mathrm{P}}^{t+1}\setminus{\eta}_{\mathrm{P},n}^{t+1}$ and $\boldsymbol{\lambda}^{t+1}\setminus{\lambda}_{n}^{t+1}$. 
The learning rates and weighting coefficient of client $n$ can be optimized at the server by solving
\begin{subequations}\small
\label{P4}
    \begin{align}
\textbf{P4}:  \quad&\min_{\lambda_{n}^{t+1},\eta_{\mathrm{P},n}^{t+1},\boldsymbol{\eta}_{\mathrm{F}}^{t}} \Phi_{n}^{t+1} \\
\label{P4_1}
\textrm {s.t. }  \quad &\varepsilon_{\mathrm{P},n}^{t+1}=\varepsilon_{\mathrm{P}}^{t+1},\\
% \label{P4_2}
% &h_{0}\left(\eta_{\mathrm{F},n}^{t},\eta_{\mathrm{P},n}^{t+1},\lambda_{n}^{t+1}\right)=h_{0}^{t},\\
% \label{P4_3}
% &0<\varepsilon_{\mathrm{P},n}^{t}<1,\\
\label{P4_3}
&0<\lambda_{n}^{t+1}<2,\\
\label{P4_4}
&0<\eta_{\mathrm{P},n}^{t+1}<1,\\
&\eqref{P_C11},\eqref{P_C12}.\nonumber
\end{align}
\end{subequations}
Since \eqref{P_C11} and \eqref{P_C12} are independent of $\eta_{\mathrm{P},n}^{t+1}$ and $\lambda_n^{t+1}$ while $\boldsymbol{\eta}_{\mathrm{F}}^{t}$ impacts $\Phi_{n}^{t+1}$ through ${\sum}_{n\in \mathcal{N}} \varepsilon_{\mathrm{F},n}^t$, we can solve Problem \textbf{P4} in two steps. We first determine the optimal FL local learning rate $\eta_{\mathrm{F},n}^{t}$, followed by the PL learning rate $\eta_{\mathrm{P},n}^{t+1}$ and the weighting coefficient $\lambda_n^{t+1}$.

\subsubsection{FL Local Learning Rate}
To minimize $\Phi_{n}^{t+1}$, we optimize the FL local learning rates $\boldsymbol{\eta}_{\mathrm{F}}^{t}$ to minimize ${\sum}_{n\in \mathcal{N}} \varepsilon_{\mathrm{F},n}^t$:
% \label{P_FLlocal}
    \begin{align}
\textbf{P5}:  \quad\min_{\boldsymbol{\eta}_{\mathrm{F}}^{t}} {\sum}_{n\in \mathcal{N}} \varepsilon_{\mathrm{F},n}^t,\quad
\textrm {s.t. }  \, \eqref{P_C11},\eqref{P_C12} \nonumber.
\end{align}
According to (\ref{vareps_F}),  Problem \textbf{P5} can be solved by setting $\eta_{\mathrm{F},n}^{t}=\underset{\eta_{\mathrm{F},n}^{t}\in (0,1)}{\arg\min}\varepsilon_{\mathrm{F},n}^t$, which is consistent among the clients and rounds. $\eta_{\mathrm{F},n}^{t}=\frac{\mu}{2(1+\phi_{1})L^{2}}$. Under properly designed $\phi_{1}$, $\phi_{2}$, and $\varphi_{1}$, $\min_{\boldsymbol{\eta}_{\mathrm{F}}^{t}}\varepsilon_{\mathrm{F},n}^t\in (0,1)$.

\subsubsection{PL learning Rate and Weighting Coefficient}
Given $\varepsilon_{\mathrm{P}}^{t+1}$ and $\varepsilon_{\mathrm{F},n}^{t}$, $\forall n\in \mathcal{N}$, 
% the personalized learning rate $\eta_{\mathrm{P},n}^{t+1}$ and weighting coefficient adjustment sub-problem of
Problem \textbf{P4} can be rewritten as}
% \begin{subequations}
\label{P5}
    \begin{align}
\textbf{P6}:  \quad\min_{\lambda_{n}^{t+1}\in(0,2),\eta_{\mathrm{P},n}^{t+1}\in(0,1)} \Phi_{n}^{t+1},\quad
\textrm {s.t. }  \,\eqref{P4_1},\eqref{P4_3},\eqref{P4_4}. \nonumber
\end{align}
% \end{subequations}
Problem \textbf{P6} is a non-convex problem because the equality constraint (\ref{P4_1}) is not an affine function; i.e., the feasible set is not a convex set. Nevertheless, we can write $\lambda_{n}^{t+1}=\lambda_{n}^{t+1}(\eta_{\mathrm{P},n}^{t+1})$, and specify the convex/concave region of Problem \textbf{P6} in $\eta_{\mathrm{P},n}^{t+1}$. Specifically, according to (\ref{P4_1}),
\begin{equation} \small
\label{lambda}
    \lambda_{n}^{t+1}=\left(1-\frac{\mu}{2}\right)^{-1}\left(\left(1-\varepsilon_{\mathrm{P}}^{t+1}\right)\big/{\eta_{\mathrm{P},n}^{t+1}}+\eta_{\mathrm{P},n}^{t+1}-\mu\right).
\end{equation}
Then, we can convert Problem \textbf{P6} to an unconstrained problem about $\eta_{\mathrm{P},n}^{t+1}$ under the typical situation with $\mu<2$ (as empirically measured in our experiments described in Section~\ref{experiments}).
{
We notice from \eqref{lambda} that, if $0<\varepsilon_{\mathrm{P}}^{t+1}<1-\frac{\mu^{2}}{4}$, then $\lambda_{n}^{t+1}>0$;
% is always greater than a certain positive value, which means that the PL learning rate has to be large enough to obtain a relatively small $\lambda_{n}^{t+1}$, or 
in other words, the regularization term on the RHS of \eqref{objective_p} can never be suppressed, even when the FL model $\tilde{\boldsymbol{\omega}}_{\mathrm{G}}^{t}$ deviates dramatically from an individual PL model $\tilde{\boldsymbol{\varpi}}_{n}^{t+1}$, penalizing personalization. 
% This is likely to cause the PL models to diverge or fail to converge to the optimal solution~\cite{bengio2012practical,goodfellow2016deep}.
}
% This might lead to instability, especially when the PL models get stuck in different local minima during each run.
For this reason, we design $\varepsilon_{\mathrm{P}}^{t+1}\geq1-\frac{\mu^{2}}{4}$, in which case $\lambda_{n}^{t+1}=\lambda_{n}^{t+1}(\eta_{\mathrm{P},n}^{t+1})$ intersects with $\lambda_{n}^{t+1}=0$, i.e., at $(\eta_2,0)$ and $(\eta_3,0)$.
By solving $\lambda_{n}^{t+1}(\eta_{\mathrm{P},n}^{t+1})=0$, the feasible set of $\eta_{\mathrm{P},n}^{t+1}$ can be written as $\Omega_{0}^{t+1}\cup\Omega_{1}^{t+1}$, with
\begin{subequations}\small
\label{feasible_set}
    \begin{align}
        \Omega_{0}^{t+1}&=\{\eta_{\mathrm{P},n}^{t+1}|\eta_1<\!\eta_{\mathrm{P},n}^{t+1}<\eta_2\};\\
        \Omega_{1}^{t+1}&=\left\{  
        \begin{aligned}
        &\left\{\eta_{\mathrm{P},n}^{t+1}|\,\eta_3\!<\!\eta_{\mathrm{P},n}^{t+1}\!<\!1\right\},\,    \text{ if } \,\varepsilon_{\mathrm{P}}^{t+1}\leq2-\mu; \\
        &\,\, \emptyset,  \quad\quad\quad\quad\quad\quad\quad\quad\quad \text{if }\, \varepsilon_{\mathrm{P}}^{t+1}\!>\!2\!-\!\mu  ,
        \end{aligned}
        \right.
    \end{align}
\end{subequations}
where $\eta_1=1-\sqrt{\varepsilon_{\mathrm{P}}^{t+1}}<1$ {is the $x$-coordinate of the intersection of $\lambda_{n}^{t+1}=\lambda_{n}^{t+1}(\eta_{\mathrm{P},n}^{t+1})$ and $\lambda=2$,} $\eta_2=\frac{\mu-\sqrt{\mu^{2}-4\left(1-\varepsilon_{\mathrm{P}}^{t+1}\right)}}{2}$, and $\eta_3=\frac{\mu+\sqrt{\mu^{2}-4\left(1-\varepsilon_{\mathrm{P}}^{t+1}\right)}}{2}$.
%The learning rate is usually set close to 0 rather than 1. Therefore, we take $\Omega_{0}^{t+1}$ as the feasible set, i.e., $\eta_{\mathrm{P},n}^{t+1} \in \left(1-\sqrt{\varepsilon_{\mathrm{P}}^{t+1}},\frac{\mu-\sqrt{\mu^{2}-4\left(1-\varepsilon_{\mathrm{P}}^{t+1}\right)}}{2}\right)$. 

% {We set the PL convergence rate }$\varepsilon_{\mathrm{P}}^{t+1}\geq1-\frac{\mu^{2}}{4}$. 
Now, constraints \eqref{P4_1}, \eqref{P4_3}, and \eqref{P4_4} are fulfilled in the feasible set specified by \eqref{feasible_set}. Problem \textbf{P6} can be rewritten as
\begin{equation}
\label{P6}
    % \begin{align}
\textbf{P7}:  \quad \min_{\eta_{\mathrm{P},n}^{t+1} \in \Omega_{0}^{t+1}\cup\Omega_{1}^{t+1} } \Phi_{n}^{t+1}, %\\
% \textrm {s.t. }  \quad& 1\!-\!\sqrt{\varepsilon_{\mathrm{P}}^{t+1}}\!<\!\eta_{\mathrm{P},n}^{t+1}\!<\!\frac{\mu\!-\!\sqrt{\mu^{2}\!-\!4\left(1\!-\!\varepsilon_{\mathrm{P}}^{t+1}\right)}}{2}.
% \end{align}
\end{equation}
which is an unconstrained optimization problem about $\eta_{\mathrm{P},n}^{t+1}$ with a nonconvex feasible set. 

\begin{theorem}
\label{theorem4}
    When $\varepsilon_{\mathrm{P}}^{t+1}\in[1-\frac{\mu^{2}}{4},2-\mu]$, $\Phi_{n}^{t+1}$ is convex in both $\eta_{\mathrm{P},n}^{t+1}\in \Omega_{0}^{t+1}$ and $\eta_{\mathrm{P},n}^{t+1}\in \Omega_{1}^{t+1}$. When $\varepsilon_{\mathrm{P}}^{t+1}\in (2-\mu,1)$ with $\mu\in (1,2)$, $\Phi_{n}^{t+1}$ is convex in $\eta_{\mathrm{P},n}^{t+1}\in \Omega_{0}^{t+1}$.
\end{theorem}

\begin{proof}
    See \textbf{Appendix \ref{theorem4_proof}}.
\end{proof}

{By \textbf{Theorem \ref{theorem4}}, 
%$\frac{\partial^{2}\Phi_{n}^{t+1}}{\partial\eta^{2}}>0$ on $\Omega_0$ and $\Omega_1$, i.e, 
Problem \textbf{P7} is convex in $\Omega_0^{t+1}$ and $\Omega_1^{t+1}$ when they are non-empty, under $\mu <2$ and $\varepsilon_{\mathrm{P}}^{t+1}\geq1-\frac{\mu^{2}}{4}$.} 
When $\varepsilon_{\mathrm{P}}^{t+1}\in[1-\frac{\mu^{2}}{4},2-\mu]$, the optimal $\boldsymbol{\eta}_{\mathrm{P}}^{t+1,*}$ can be obtained by comparing the respective solutions in $\Omega_0^{t+1}$ and $\Omega_1^{t+1}$. 
When $\varepsilon_{\mathrm{P}}^{t+1}\in (2-\mu,1)$ (i.e., $1<\mu<2$), Problem~\textbf{P7} is convex.
The solutions can be obtained using convex optimization tools.
{After minimizing $\Phi_{n}^{t+1}$ for every client $n \in \mathcal{N}$, the maximization problem (i.e., $\underset{n\in \mathcal{N}}{\max}\,\Phi_n^{t+1}$) is solved by comparing the minimized $\Phi_{n}^{t+1}$ among all clients.}
%methods, such as Newton's Method \cite{bertsekas1997nonlinear}.
% On the other hand, when $\varepsilon_{\mathrm{P}}^{t+1}\in(0,1-\frac{\mu^{2}}{4})$, the convexity of Problem \textbf{P7} is hard to assess. However, when $\varepsilon_{\mathrm{P}}^{t+1}\in(0,1-\frac{\mu^{2}}{4})$, the obtained solutions of Problem \textbf{P7} could be sub-optimal. 
% While a small $\varepsilon_{\mathrm{P}}^{t+1}$ results in a large convergence rate of PL models, excessively a small $\varepsilon_{\mathrm{P}}^{t+1}$ can cause the PL models to be trapped in local minima or fail to generalize adequately. 
% Therefore, we design $\varepsilon_{\mathrm{P}}^{t+1}$ within the range of $[1-\frac{\mu^{2}}{4},1)$, e.g., $\varepsilon_{\mathrm{P}}^{t+1}=1-\frac{\mu^{2}}{4}$.}

\subsection{Algorithm Description and Discussion}

The overall algorithm is summarized in Algorithm \ref{algorithm2}. 
% and illustrated in Fig.~\ref{fig:flowchart}.
The complexity of the algorithm is dominated by the KM algorithm used to solve Problem \textbf{P3} and the convex optimization used to solve Problem \textbf{P7}.
The worst-case complexity of using the KM algorithm to solve Problem \textbf{P3} is $\mathcal{O}(|N|^3)$ \cite{jungnickel2005graphs}, as $N$ specifies the number of vertices in the bipartite graph. 

According to \textbf{Theorem \ref{theorem4}}, in the case where Problem \textbf{P7} is convex in the feasible sets $\Omega_0^{t+1}$ and $\Omega_1^{t+1}$ or is convex in the feasible set $\Omega_0^{t+1}$, the algorithm converges to the global optimum.
The complexity of using a typical convex optimization solver, e.g., interior point method, to solve Problem \textbf{P7} is $\mathcal{O}(V_{1}^{4.5}\log(\frac{1}{\alpha}))$ \cite{sun2021joint}, where $V_{1}$ is the number of variables and $\alpha$ is the convergence accuracy. Here, $V_{1}=1$. The problem is solved for $N$ clients in parallel. 
As a result, the overall complexity of Algorithm \ref{algorithm2} is $\mathcal{O}(N^3+N\log(\frac{1}{\alpha}))$.

Under $\varepsilon_{\mathrm{P}}^{t+1} \in [1-\frac{\mu^{2}}{4},1)$, we confirm the optimality of the solution to Problem \textbf{P1}. Due to the consistency of $\Theta_{\mathrm{L}}^{t}$ across the clients and the independence of $\Phi_n^t$ under given $\Theta_{\mathrm{L}}^{t}$, $\varepsilon_{\mathrm{P}}^{t+1}$, and $\boldsymbol{\eta}_{\mathrm{F}}^{t}$, Problem \textbf{P1} is a nested problem and decoupled into two subproblems. The first subproblem, i.e., client selection, channel allocation, and power control, is optimally solved using the KM algorithm. The second subproblem, i.e., learning rate and weighting coefficient adjustment, is further divided between the FL local learning rate adjustment, and the PL learning rate and weighting coefficient adjustment (i.e., Problem \textbf{P6}).
% , since $\Phi_{n}^{t+1}$ is independent of $\boldsymbol{\eta}_{\mathrm{P}}^{t+1}\setminus{\eta}_{\mathrm{P},n}^{t+1}$ and $\boldsymbol{\lambda}^{t+1}\setminus{\lambda}_{n}^{t+1}$, and is impact by $\boldsymbol{\eta}_{\mathrm{F}}^{t}$ with $\underset{n\in \mathcal{N}} {\sum}\varepsilon_{\mathrm{F},n}^t$.
Under $\varepsilon_{\mathrm{P}}^{t+1} \in [1-\frac{\mu^{2}}{4},1)$, Problem \textbf{P6} is convex within the specified convex region of $\eta_{\mathrm{P},n}^{t+1}$, optimally solved through convex optimization methods. 
The solution to Problem \textbf{P1} is optimal, {under the specified $\boldsymbol{\eta}_{\mathrm{F}}^{t}$, $\boldsymbol{\eta}_{\mathrm{P}}^{t+1}$, and $\boldsymbol{\lambda}^{t+1}$. 
}

{According to \eqref{down_glb_di} and \eqref{Gamma}, the per-round convergence upper bound of the FL global model depends on $\Gamma_{t+1}$ and ${\sum}_{n\in\mathcal{N}_{t}}\varepsilon_{\mathrm{F},n}^{t}$, which are minimized by solving Problems~\textbf{P2} and \textbf{P5}, respectively. In this sense, through client selection, channel allocation, power control, and FL local learning rate configuration, the per-round convergence upper bound and subsequently the overall convergence upper bound of the FL global model are minimized.} 

 \begin{algorithm}[t]
    \caption{Proposed Transmission Scheduling Policy}
    \label{algorithm2}
    \begin{algorithmic}[1]
        \Require $T_0$, $\{t_n=0\}_{n\in\mathcal{N}}$, $r_{\mathrm{min}}$, $\{P_n^{\mathrm{th}}\}_{n \in \mathcal{N}}$, $\varepsilon_{\mathrm{P}}$, $h_0^{0}$.
        \Ensure $\mathbf{c}^{t}$, $\mathbf{p}^{t}$, $\boldsymbol{\eta}_{\mathrm{F}}^{t}$, $\boldsymbol{\eta}_{\mathrm{P}}^{t+1}$, $\boldsymbol{\lambda}^{t+1}$, $\forall t$.
        \While{$\{n \mid t_n \leq T_0, n \in \mathcal{N}\} \neq \emptyset$}
            \State {$//$Client selection, channel allocation, power control}
            \State Let $P_n^t = P_n^{\mathrm{th}}$, $\forall n \in \mathcal{N}$;
            \State Obtain the set of candidate clients $\mathcal{N}_{t}^{\mathrm{a}}$ and $\{\rho_{n,\mathrm{L}}^{t}\}_{n \in \mathcal{N}}$ according to (\ref{SNR})-(\ref{element_error_pr});
            \State Obtain the optimal client selection $\mathcal{N}_{t}$ and channel allocation $\mathbf{C}^{t}$ by solving Problem \textbf{P3} using KM algorithm; 
            \State $t_n \leftarrow t_n + 1$, $\forall n \in \mathcal{N}_{t}$;
            \State $//$ Learning rate and weighting coefficient adjustment
            \State Set $\eta_{\mathrm{F},n}^{t}=\frac{\mu}{2(1+\phi_{1})L^{2}}$ and $\varepsilon_{\mathrm{P}}^{t+1}=\varepsilon_{\mathrm{P}}$;
            \ParFor{$n \in \mathcal{N}$}
                % \If{$\varepsilon_{\mathrm{P}}^{t+1} \in (2 - \mu, 1)$}
                    \State Obtain the optimal PL learning rate $\eta_{\mathrm{P},n}^{t+1,*}$ by solving Problem \textbf{P7}; 
                % \ElsIf{$\varepsilon_{\mathrm{P}}^{t+1} \in [1 - \frac{\mu^{2}}{4}, 2 - \mu]$}
                %     \State Obtain the optimal PL rate $\eta_{\mathrm{P},n,0}^{t+1}$ and $\eta_{\mathrm{P},n,1}^{t+1}$ by minimizing $\Phi_{n}^{t+1}$ on $\Omega_0^{t+1}$ and $\Omega_1^{t+1}$, using Newton's Method separately
                %     \State Obtain the optimal $\eta_{\mathrm{P},n}^{t+1,*} = \underset{\eta_{\mathrm{P},n}^{t+1}}{\arg\min}\{\Phi_{n}^{t+1}\mid_{\eta_{\mathrm{P},n,0}^{t+1}}, \Phi_{n}^{t+1}\mid_{\eta_{\mathrm{P},n,1}^{t}}\}$
                % \EndIf
                \State Obtain the optimal weighting coefficient $\lambda_{n}^{t+1,*}$ based on (\ref{lambda}) and $\eta_{\mathrm{P},n}^{t+1,*}$;
            \EndParFor
            % \State Set $h_{0}$ as the lower bound of (\ref{vL_c2}), and obtain the optimal local learning rates $\eta_{\mathrm{F},n}^{t}$ by solving (\ref{eta_G}) in parallel;
            \State $t \leftarrow t + 1$.
        \EndWhile
    \end{algorithmic}
\end{algorithm}

 % 	\begin{figure}[!t]
	% \centering
 %        \includegraphics[width=0.48\textwidth]{pic/flowchart5.pdf}
	% \caption{The flowchart of the proposed configuration and scheduling policy.}
	% \label{fig:flowchart}
	% \end{figure}

\section{Experiments and Results}
\label{experiments}
% In this section, we assess the convergence, accuracy, and fairness of DP-Ditto experimentally. The impact of privacy considerations on those aspects of DP-Ditto is discussed.

% \subsection{Experimental Settings}
\begin{table}[t]
\caption{Simulation parameter configuration}
\label{parameters}
\centering
\small 
\begin{tabular}{|*{2}{l|}}
\hline
\textbf{Parameter} & \textbf{Value} \\
\hline
% Cell radius & 150 m \\
% % \hline
% D2D pair distance & $\textless$40 m \\
% % \hline
% RB bandwidth $W$& 180 KHz \\
% \hline
Total bandwidth & 10 MHz \\
BS's maximum transmit power & 30 dBm \\
% \hline
client's maximum transmit power $P_{n}^{\mathrm{th}}$& 23 dBm \\
% % % \hline
Noise spectral density & -169 dBm/Hz \\
% % \hline
% Cellular link path loss & $128.1+37.6\log_{10}{\left(d\right)} $ \\
% \hline
Path loss at 1 m & -30 dB \\
% \hline
Path loss exponent & 2.8 \\
% \hline
 % Maximum transmission delay $\tau_{\mathrm{max}}$ & 0.1 ms \\
% \hline
Modulation order $M_{\boldsymbol{\omega}}$ & 256 \\
% \hline
% Symbols' number per parameter $N_{\boldsymbol{\omega}}$& 2  \\
% \hline
Sampling rate $q$& 0.01  \\
% \hline
% {The standard deviation of DP noise $\sigma_\mathrm{DP}$ & 0.01  }\\
% % \hline
% Discount factor $\gamma$& 0.97  \\
% % \hline
% Learning rate & 0.01  \\
% % \hline
% Soft update parameter $v$& 0.001  \\
% % \hline
% penalty coefficient $R^s$& 1 \\
\hline
\end{tabular}
\end{table}
Consider $M=20$ clients under the coverage of a BS with a coverage radius of 100 m. The distances between the BS and the clients are uniformly randomly taken from $[10,100]$~m. We consider Rayleigh fading for both uplink and downlink channels. The total bandwidth is $K\cdot B=10$ MHz with $K=10$ by default.
% The other parameters are provided in Table~\ref{}.
% We set $N=20$ clients by default. 
% Each quantized parameter is modulated into $a=2$ 256-QAM symbols, i.e., 
Three network models are considered here:
% \subsubsection{Models} 
% We consider the following three models.
\begin{itemize}
\item \textbf{MLR:} This classification method generalizes logistic regression to multiclass problems. It constructs a linear predictor function to predict the probability of an outcome based on an input observation.
% \item \textbf{Logic regression:} 
\item \textbf{DNN:} This model consists of an input layer, a fully connected hidden layer (with 100 neurons), and an output layer. The rectified linear unit (ReLU) activation function is applied to the hidden layer.
\item \textbf{CNN:} This model contains two convolutional layers with 32 and 64 convolutional filters per layer, and a pooling layer in-between to prevent over-fitting. Following the convolutional layers are two fully connected layers (with 1024 and 512 neurons for FMNIST, and 1600 and 512 neurons for CIFAR10). We use the ReLU in the convolutional and fully connected layers. 
\end{itemize} 
% The default learning rates for the local and personalized models are $\eta_{\mathrm{F},n}^{t}=0.01$ and $\eta_{\mathrm{P},n}^{t}=0.01$, respectively. The clipping threshold is $C=3$, $C=7$, and $C=20$ for MLR, DNN, and CNN, respectively. The default privacy budget and the sampling rate are $\epsilon_\mathrm{Q}=0.01$ and $q=0.01$, respectively. $\sigma_\mathrm{DP}=0.01$. The total bandwidth is $KB=100$ MHz, The maximum transmitting power is $$
The default FL and PL learning rates are $\eta_{\mathrm{F},n}^{t}=0.01$ and $\eta_{\mathrm{P},n}^{t}=0.01$, respectively. 
{ The default PL learning rates are used as the initial PL learning rates when $t=0$, and the default PL and FL learning rates are used for the compared scheduling policy (i.e., Non-Adjustment).}
% {\color{green}In our simulation, the number of quantization bits $R$ and clipping threshold $C$ are set considering the communication overhead and accuracy of network models.}
The clipping threshold is $C=3$, $7$, and $20$ for MLR, DNN, and CNN, respectively. The default privacy budget, 
% {\color{green}the number of quantization bits,} 
the maximum number of transmissions per client, and the weighting coefficient for each client are $\epsilon_\mathrm{Q}=1$, 
% {\color{green}$R=16$,} 
$T_0=20$, and $\lambda_n^t=0.5$. 
By default, $\delta_{\mathrm{Q}}=0.001$ for DNN and MLR, and $\delta_{\mathrm{Q}}=0.005$ for CNN.
% The default $\delta_\mathrm{Q}$ is $0.0007$, $$  
{The corresponding values of $\sigma_{\mathrm{DP}}$ are given in Table~\ref{table_sigma}.} 
The maximum transmission delays are $\tau_{\max}=0.01$, $0.1$, and $0.6$ s for MLR, DNN, and CNN, respectively.
{Given a dataset and an ML model, $L$ and $\mu$ can be obtained by empirically estimating the minimum and maximum of $\frac{\left\Vert \nabla F\left(\boldsymbol{\omega}\right)-\nabla F\left(\boldsymbol{\omega}'\right)\right\Vert }{\left\Vert \boldsymbol{\omega}-\boldsymbol{\omega}'\right\Vert }$~\cite{8664630}. $L$ is the maximum. $\mu$ is the minimum.} 
% \color{green!60!black}
% $L=1.32$, $0.43$, $0.29$, $0.33$, and $\mu=0.27$, $0.13$, $0.05$, $0.09$, for DNN and MLR on the MNIST dataset, and CNN on the FMNIST and CIFAR10 datasets, respectively.
\footnote{For DNN on the MNIST dataset, $L=1.32$ and $\mu=0.27$. 
For MLR on the MNIST dataset, $L=0.43$ and $\mu=0.13$. 
For CNN on the FMNIST dataset, $L=0.29$ and $\mu=0.05$. 
For CNN on the CIFAR10 dataset, $L=0.33$ and $\mu=0.09$. 
% In our simulation, to satisfy the required privacy budget for different network models, 
The other parameters are specified in Table~\ref{parameters}.}

 \begin{table}\small
 \tabcolsep=3pt
    \centering
    % \resizebox{0.45\textwidth}{!}{%
    \caption{The standard deviation of the DP noise, $\sigma_{\mathrm{DP}}$}
    \begin{tabular}{l|c|c|c|c|c|c}
     % \diagbox[width=2em]{Metrics}{Methods} 
     \hline
     & $T_{0}=5$ & $T_{0}=10$& $T_{0}=15$ & $T_{0}=20$& $T_{0}=25$ &$T_{0}=30$  \\
     \hline
       MLR  & {0.001}  & 0.003  & 0.005  & 0.006  & 0.008 & 0.01 \\
       % \hline
       DNN  & 0.004  & 0.008  & 0.012  & 0.016  & 0.02 & 0.024 \\
       % \hline
       CNN  & 0.0025  & 0.0045 & 0.007  & 0.009  & 0.012 & 0.014 \\
       \hline
    \end{tabular} %
    % }
    \label{table_sigma}
\end{table}

% \subsubsection{Datasets}
We consider three widely used public datasets, i.e., MNIST, Fashion-MNIST (FMNIST), and CIFAR10.
% \begin{itemize}
%     \item \textbf{MNIST:} The standard MNIST comprises $60,000$ training and $10,000$ testing examples of $28 \times 28$ grayscale images of handwritten digits between "$0$" and "$9$" \cite{lecun1998gradient}. We use the DNN and MLR to classify this dataset.
%     \item \textbf{FMNIST:} The FMNIST dataset contains $28 \times 28$ grayscale images of $70,000$ fashion products in ten categories (labels), e.g., coats and dresses \cite{xiao2017Fashion}. We use the CNN to classify the FMNIST dataset.
%     \item \textbf{CIFAR10:} This dataset consists of $60,000$ examples of $32 \times 32$ color images in ten classes ($6,000$ per class), including $50,000$ for training and $10,000$ for testing. We use the CNN and MLR to classify the dataset.
% \end{itemize}
Cross-entropy loss is considered for the datasets. 
% for MNIST, FMNIST, and CIFAR10.
% \subsubsection{Benchmark} 
% We additionally explore the performance of other personalized FL frameworks regarding accuracy and fairness. In particular, 

The following benchmarks are considered for WPFL:
\begin{itemize} 
    \item 
    \textbf{pFedMe~\cite{t2020personalized}:} The global FL model is updated in the same way as the typical FL. Learning from the global model, each personalized model is updated based on a regularized loss function using the Moreau envelope. 
    \item
    \textbf{APPLE~\cite{luo2022adapt}:} Each client uploads to the server a core model learned from its personalized model and downloads the other clients' core models in each round. The personalized model is obtained by locally aggregating the core models with learnable weights.
    \item 
    \textbf{FedAMP~\cite{huang2021personalized}:} The server has a set of personalized cloud models. Each client has a local personalized model. In each round, the server updates the personalized cloud models using an attention-inducing function of the uploaded local models and combination weights. Upon receiving the cloud model, each client locally updates its personalized model based on a regularized loss function.
    \item 
    \textbf{FedALA~\cite{zhang2023fedala}:} In every round of FedALA, each client adaptively initializes its local model by aggregating the downloaded global model and the old local model with learned aggregation weights before local training.
    % \item 
    % \textbf{Ditto~\cite{li2021ditto}:} This original Ditto does not consider DP, as described in Section~III-A.
    
\end{itemize}
{For a fair comparison, all these benchmarks are enhanced with the proposed DP mechanism and scheduling policy.

{To assess the proposed quantization-assisted Gaussian mechanism, we consider the following DP implementations {and baselines}:
\begin{itemize}    
    \item \textbf{Gaussian mechanism~\cite{wei2020federated}:} The FL local model is protected by Gaussian noise satisfying $(\epsilon_{\mathrm{Q}},\delta_{\mathrm{Q}})$-DP. { The contribution of quantization to privacy is overlooked.} 
    % {
    % The standard deviation of the noise
    % }
    \item \textbf{Moments accountant (MA)-based DP mechanism~\cite{abadi2016deep}:} Much tighter estimates on the privacy loss can be obtained through the MA technique and bisection. 
    The FL local model is protected by Gaussian noise with a smaller standard deviation {than the Gaussian mechanism~\cite{wei2020federated}.
    % , derived utilizing the MA and bisection methods. 
    This mechanism also overlooks the contribution of quantization to privacy.}
    
    {\item \textbf{DP with dithering quantization~\cite{wang2024p2cefl}:} 
    To guarantee $(\epsilon_{\mathrm{Q}},\delta_{\mathrm{Q}})$-DP, the quantization intervals are determined by sampling a set of gamma random variables for the coordinates of all clients at each slot. Each client's FL local model is quantized by adding uniform noise. With shared random seeds across clients, the server estimates the local models by subtracting the uniform noise. 
    % added to each parameter.
    
    \item \textbf{Perfect Gaussian~\cite{wei2020federated}:} The PFL is executed in an ideal environment with no quantization conducted and no communication error undergone. The FL local model is protected by Gaussian noise that satisfies $(\epsilon_{\mathrm{Q}},\delta_{\mathrm{Q}})$-DP. No quantization noises and transmission errors are considered. Comparing the proposed mechanism with this baseline helps isolate the impact of quantization and imperfect communication on the proposed mechanism.
    
    \item\textbf{WPFL without DP:} The standard WPFL is conducted with no privacy considered, where quantization is conducted for uploading the FL local models and downloading the FL global models, and imperfect communication channels are undergone. Comparing the proposed mechanism with this baseline helps isolate the impact of DP on the proposed mechanism.}    
\end{itemize}}

{To assess the proposed scheduling policy, the following scheduling policies are considered for comparison:}
\begin{itemize}
    \item 
    \textbf{Round-Robin:} The BS selects the clients for the available subchannels in a round-robin manner with no adjustment of the coefficients.
    \item 
    \textbf{Random Selection:} The clients are selected randomly with random subchannel allocation and fixed coefficients within all training rounds.
    \item 
    \textbf{Non-Adjustment:} The BS selects the clients in each round utilizing the KM algorithm with fixed coefficients within all training rounds.
\end{itemize}

\subsubsection{Communication overhead}
\begin{table*}
\centering
\small 
\caption{\small  The average numbers of quantization bits ($B_{\mathrm{q}}$) and overhead bits ($B_{\mathrm{o}}$) per parameter under different DP implementations. Note that a client transmits $\min(16, B_{\rm q}+B_{\rm o})$ bits per parameter since it would be more efficient to transmit all 16 quantization bits per parameter when $B_{\rm q}+B_{\rm o}\geq 16$; e.g., the Gaussian mechanism under MLR.} 
% on the MNIST dataset, DNN on the MNIST dataset, and CNN on the CIFAR10 dataset.}
\begin{tabular}{c|c|cc|cc|cc|cc|cc}
\toprule
\multirow{2}{*}{Model} & \multirow{2}{*}{Dataset} & \multicolumn{2}{c|}{Proposed} & \multicolumn{2}{c|}{MA} & \multicolumn{2}{c|}{Gaussian} & \multicolumn{2}{c|}{Dithering} & \multicolumn{2}{c}{Without DP}\\
\cline{3-12}
& & $B_{\mathrm{q}}$ & $B_{\mathrm{o}}$ & $B_{\mathrm{q}}$ & $B_{\mathrm{o}}$ & $B_{\mathrm{q}}$ & $B_{\mathrm{o}}$ & $B_{\mathrm{q}}$ & $B_{\mathrm{o}}$ & $B_{\mathrm{q}}$ & $B_{\mathrm{o}}$\\
\midrule
MLR & MNIST   & 6.81 & 2.61 & 8.98 & 4.00 & 14.32 & 3.93 & 8.01 & 1.34 & 4.70 & 1.63\\
DNN & MNIST   & 4.55 & 3.04 & 11.13 & 4.00 & 14.82 & 2.10 & 6.91 & 1.20 & 2.54 & 0.72 \\
CNN & FMNIST  & 5.17 & 4.00 & 10.95 & 4.00 & 14.79 & 1.09 & 5.59 & 4.00 & 3.88 & 3.85 \\
CNN & CIFAR10 & 4.93 & 4.00 & 11.34 & 4.00 & 14.84 & 0.89 & 5.20 & 4.00 & 3.55 & 3.92 \\
\bottomrule
\end{tabular}
\label{table_bit}
\end{table*}
{
Under the proposed mechanism, the Gaussian mechanism~\cite{dwork2014algorithmic}, and MA~\cite{abadi2016deep}, we employ a 16-bit quantizer ($R=16$) to ensure fine quantization intervals. Note that the significant bits of most weight parameters remain unused. Only the effectively utilized bits, along with the sign bit, need to be transmitted for each parameter. Moreover, multiple consecutive parameters may have the same number of effective bits. An index list is also sent to specify the count of consecutive parameters sharing the same number of effective bits and the number. 

Under Dithering~\cite{wang2024p2cefl}, the number of quantization bits per parameter can vary across clients. The average numbers of quantization bits are evaluated to be $16.05$, $16.05$, $16.25$, and $16.00$, for MLR on the MNIST dataset, DNN on the MNIST dataset, CNN on the FMNIST dataset, and CNN on the CIFAR10 dataset, respectively. Likewise, only the effective bits and the sign bit are transmitted, along with an index list specifying the count of consecutive parameters sharing the same number of effective bits and the number. 

Table~\ref{table_bit} shows the average numbers of quantization bits ($B_{\mathrm{q}}$) and overhead bits ($B_{\mathrm{o}}$) per parameter under different DP mechanisms, datasets, and models. Clearly, the consideration of DP increases the number of bits to be transmitted, because the DP noise extends the range of the model parameters. Among the schemes considering DP, the proposed mechanism generally requires the smallest number of bits to be transmitted, as the privacy-enhancing capability of quantization is exploited to help reduce the DP noise in the mechanism.}

\subsubsection{Comparison with existing DP mechanisms}
\begin{figure}[t]
\centering  
\subfigure[{Accuracy vs. $T$ (DNN, MNIST)}]
{
\label{privacy_mechanisms_dnn_mnist}
\includegraphics[width=0.23\textwidth]{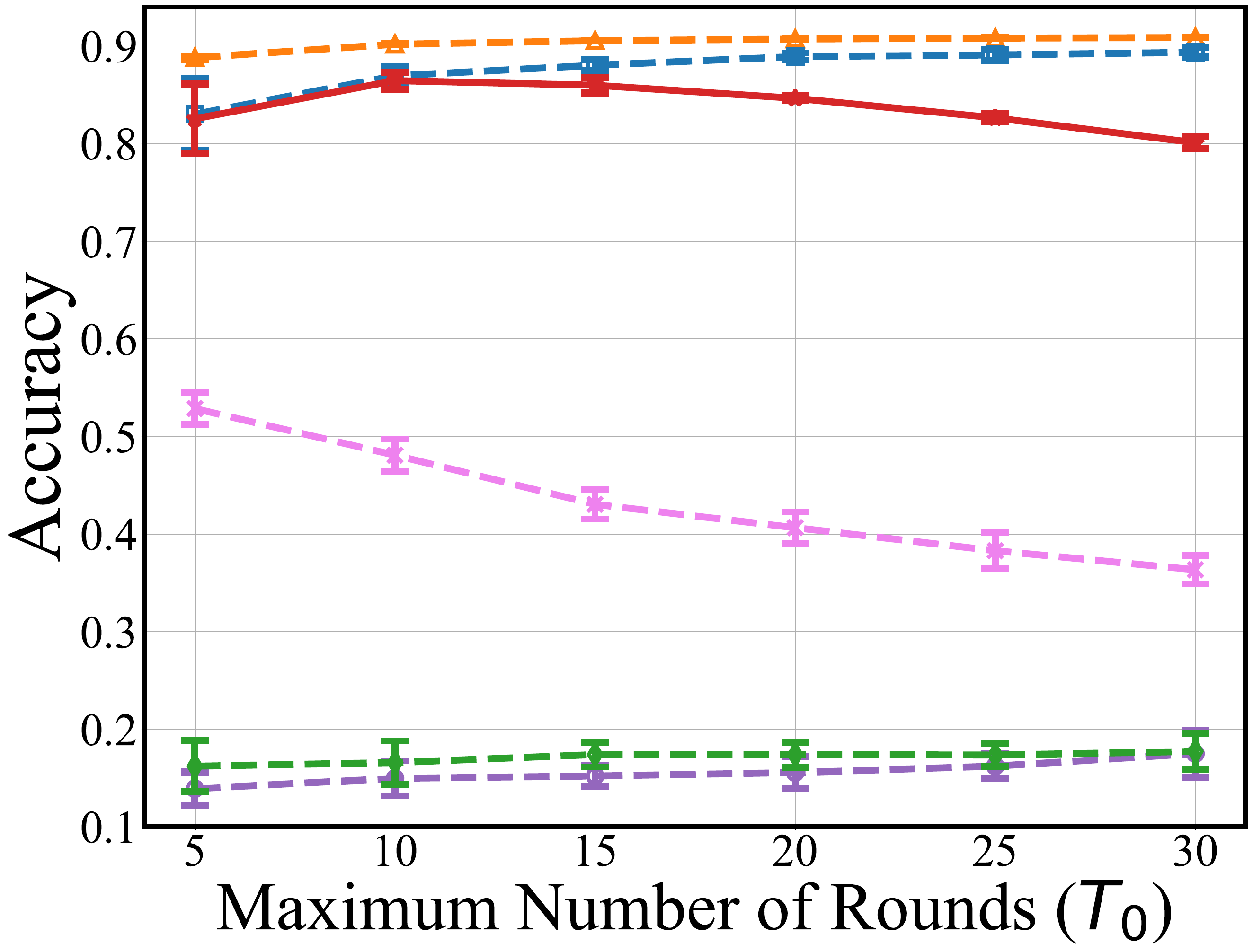}} 
\hspace{-1mm}
\subfigure[{Accuracy vs. $T$ (MLR, MNIST)}]{
\label{ privacy_mechanisms_mlr_mnist}
\includegraphics[width=0.23\textwidth]{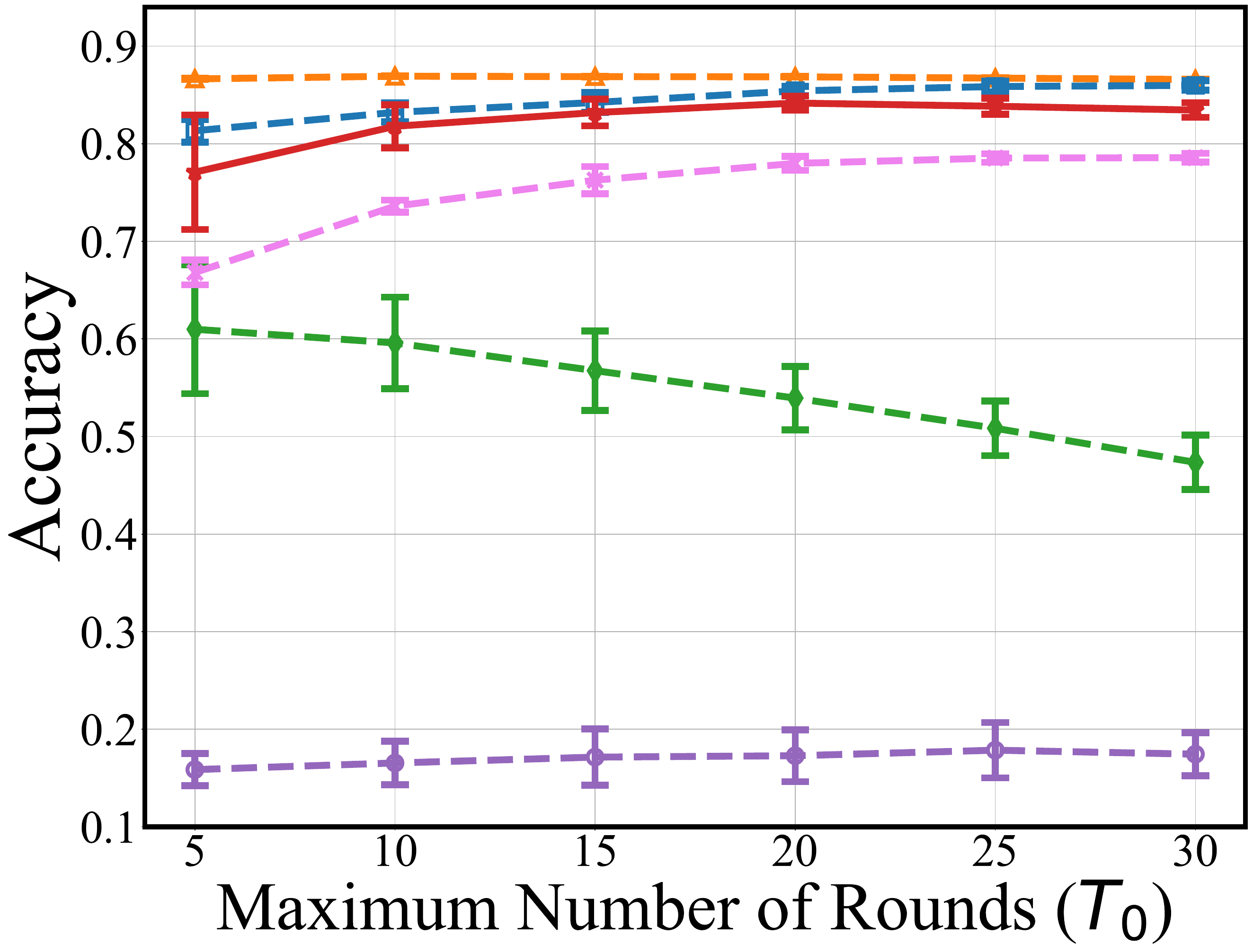}}
\\
\subfigure[Accuracy vs. $T$ (CNN,CIFAR10)]{
\label{privacy_mechanisms_cnn_Cifar10}
\includegraphics[width=0.23\textwidth]{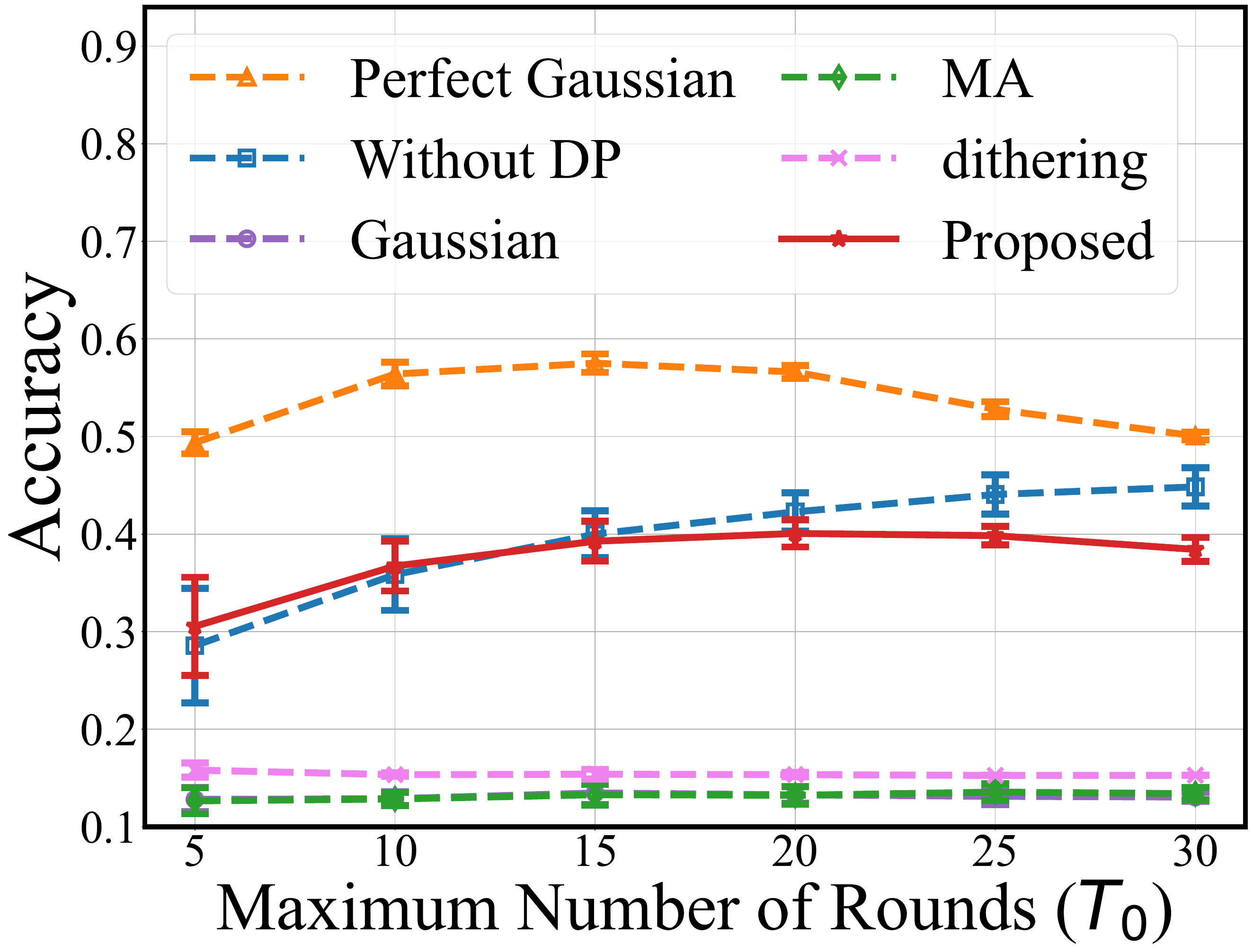}}
\hspace{-1mm}
\subfigure[Accuracy vs. $T$ (CNN, FMNIST)]{
\label{privacy_mechanisms_cnn_fmnist}
\includegraphics[width=0.23\textwidth]{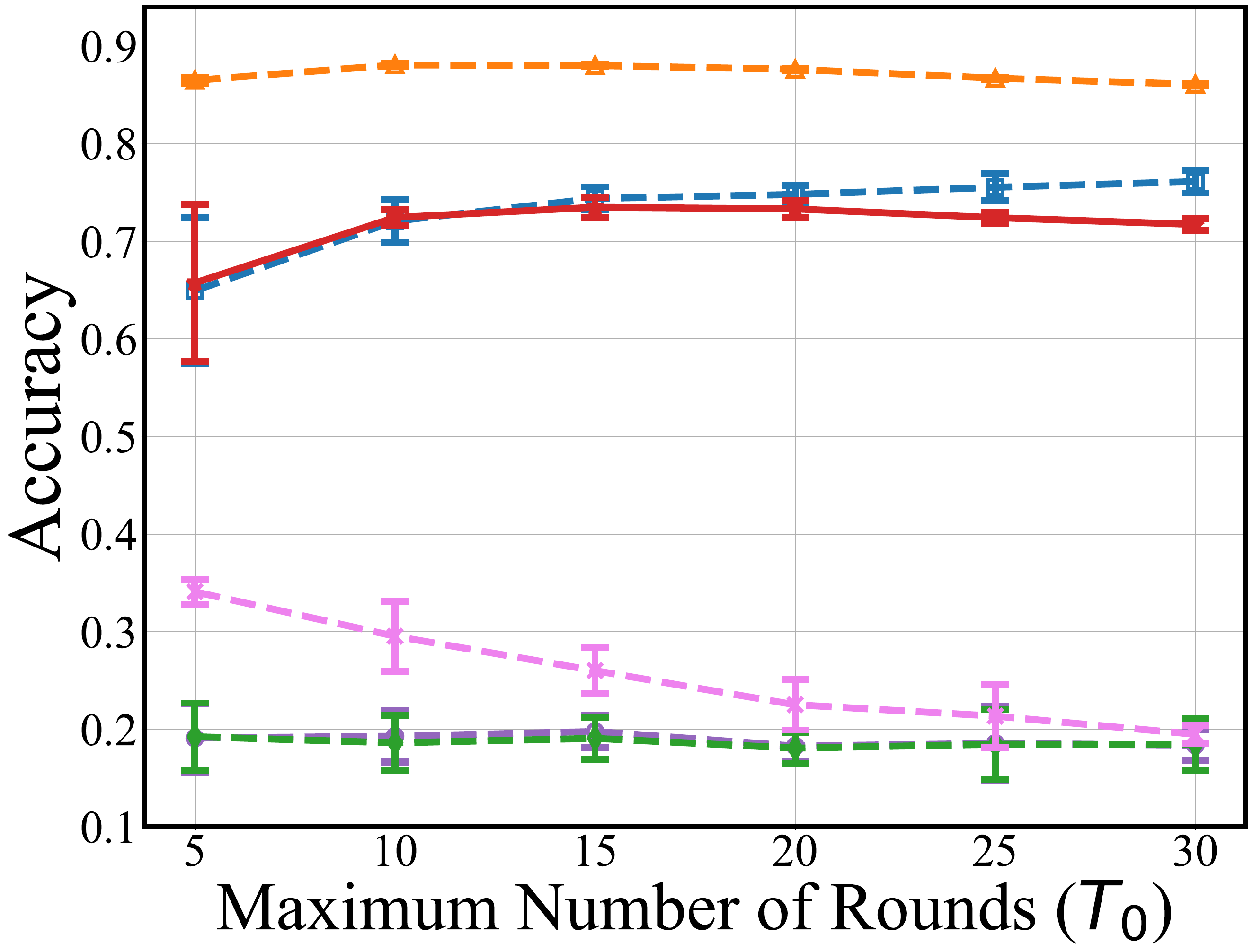}}
\caption{PL model accuracy vs. the maximum number of rounds $T_{0}$ under different allocation methods, where the default setting 
% of $\epsilon_{\mathrm{Q}}$ and $\delta_{\mathrm{Q}}$
is considered: $\epsilon_{\mathrm{Q}}=1$, $\delta_{\mathrm{Q}}=0.001$ for DNN and MLR, and $\delta_{\mathrm{Q}}=0.005$ for CNN.}
\label{privacy_mechanisms_acc}
\end{figure}

% \begin{table*}[]
% % \tabcolsep=1pt
% \centering
% \caption{\small The average numbers of quantization bits ($B_{\mathrm{q}}$) and overhead bits ($B_{\mathrm{o}}$) per parameter under different DP implementations. Note that a client transmits $\min(16, B_{\rm q}+B_{\rm o})$ bits per client since it would be more efficient to transmit all 16 quantization bits per parameter when $B_{\rm q}+B_{\rm o}\geq 16$; see the Gaussian mechanism under MLR on the MNIST dataset, DNN on the MNIST dataset, and CNN on the CIFAR10 dataset.}
% \begin{tabular}{c|c|cc|cc|cc|cc|cc}
% \toprule
% \multirow{2}{*}{Model} & \multirow{2}{*}{Dataset} & \multicolumn{2}{c|}{Proposed} & \multicolumn{2}{c|}{MA} & \multicolumn{2}{c|}{Gaussian} & \multicolumn{2}{c|}{Dithering} & \multicolumn{2}{c}{Without DP}\\
% \cline{3-12}
% & & $B_{\mathrm{q}}$ & $B_{\mathrm{o}}$ & $B_{\mathrm{q}}$ & $B_{\mathrm{o}}$ & $B_{\mathrm{q}}$ & $B_{\mathrm{o}}$ & $B_{\mathrm{q}}$ & $B_{\mathrm{o}}$ & $B_{\mathrm{q}}$ & $B_{\mathrm{o}}$\\
% \midrule
% MLR & MNIST   & 6.81 & 2.61 & 8.98 & 4.00 & 14.32 & 3.93 & 8.01 & 1.34 & 4.70 & 1.63\\
% DNN & MNIST   & 4.55 & 3.04 & 11.13 & 4.00 & 14.82 & 2.10 & 6.91 & 1.20 & 2.54 & 0.72 \\
% CNN & FMNIST  & 5.17 & 4.00 & 10.95 & 4.00 & 14.79 & 1.09 & 5.59 & 4.00 & 3.88 & 3.85 \\
% CNN & CIFAR10 & 4.93 & 4.00 & 11.34 & 4.00 & 14.84 & 0.89 & 5.20 & 4.00 & 3.55 & 3.92 \\
% \bottomrule
% \end{tabular}
% \label{table_bit}
% \end{table*}

{
Fig.~\ref{privacy_mechanisms_acc} evaluates the impact of different DP mechanisms (i.e., Gaussian~\cite{dwork2014algorithmic}, MA~\cite{abadi2016deep}, DP with Dithering~\cite{wang2024p2cefl}, and the proposed quantization-assisted Gaussian mechanism) on the accuracy of WPFL in a noisy wireless environment. 
% {\color{green}where $\epsilon_{\mathrm{Q}}=1$. 
% % We compare the results with the accuracy of models trained without DP perturbation.
% $\delta_{\mathrm{Q}}=0.001$ for the DNN and MLR models on the MNIST dataset, and $\delta_{\mathrm{Q}}=0.005$ for the CNN model on the CIFAR10 and FMNIST datasets.}
The accuracy of WPFL under the proposed DP protection first increases and then decreases as $T_0$ grows. This is because the effect of the DP perturbation and the quantization and transmission errors accumulates, and $\sigma_{\mathrm{DP}}$ increases with $T_{0}$ based on \textbf{Theorem~\ref{privacy_budget}}, causing performance degradation when $T_0$ is large.}

{In the scenarios with quantization and imperfect communication channels, the proposed mechanism achieves at least $5.00\%$ better accuracy than the second-best (i.e., Dithering) 
% mechanisms sustaining quantization noises and transmission errors while overlooking the privacy benefit of quantization (i.e., Gaussian, MA, and Dithering)},  
and is only $9.26\%$ worse than the standard WPFL without DP.
This is because our mechanism utilizes the inherent privacy-preserving ability of quantization, reducing the intensity of the added noise compared to the Gaussian and MA mechanisms. 
Although privacy is guaranteed using uniform noise addition and dithering quantization in Dithering, the quantization intervals can vary over time and differ among clients, depending on gamma random variables sampled in each round. The bit lengths can be large for clients with tiny intervals, causing high transmission errors under Dithering.

The proposed mechanism is 20.76\% worse in test accuracy than Perfect Gaussian. This difference arises because Perfect Gaussian operates in an ideal environment where there is no quantization conducted and no communication errors undergone, and FL and PL model training is influenced solely by Gaussian noise.}  

{
To maintain acceptable performance of PL models in practice, 
a training process can be terminated once the accuracy stops improving or starts to degrade. Additionally, cross-validation can be employed to select $T_0$ by evaluating performance across multiple subsets of the data. Note that the optimal $T_0$ can be different under specific networks, datasets, and DP mechanisms. Fig.~\ref{privacy_mechanisms_acc} provides guidance on adjusting $T_0$ in different experimental setups.}

\subsubsection{Comparison with alternative allocation schemes}

\begin{figure}[t]
\centering  
\subfigure[{Accuracy vs. $T_0$ (DNN, MNIST)}]
{
\label{allocation_acc_mnist_dnn}
\includegraphics[width=0.23\textwidth]{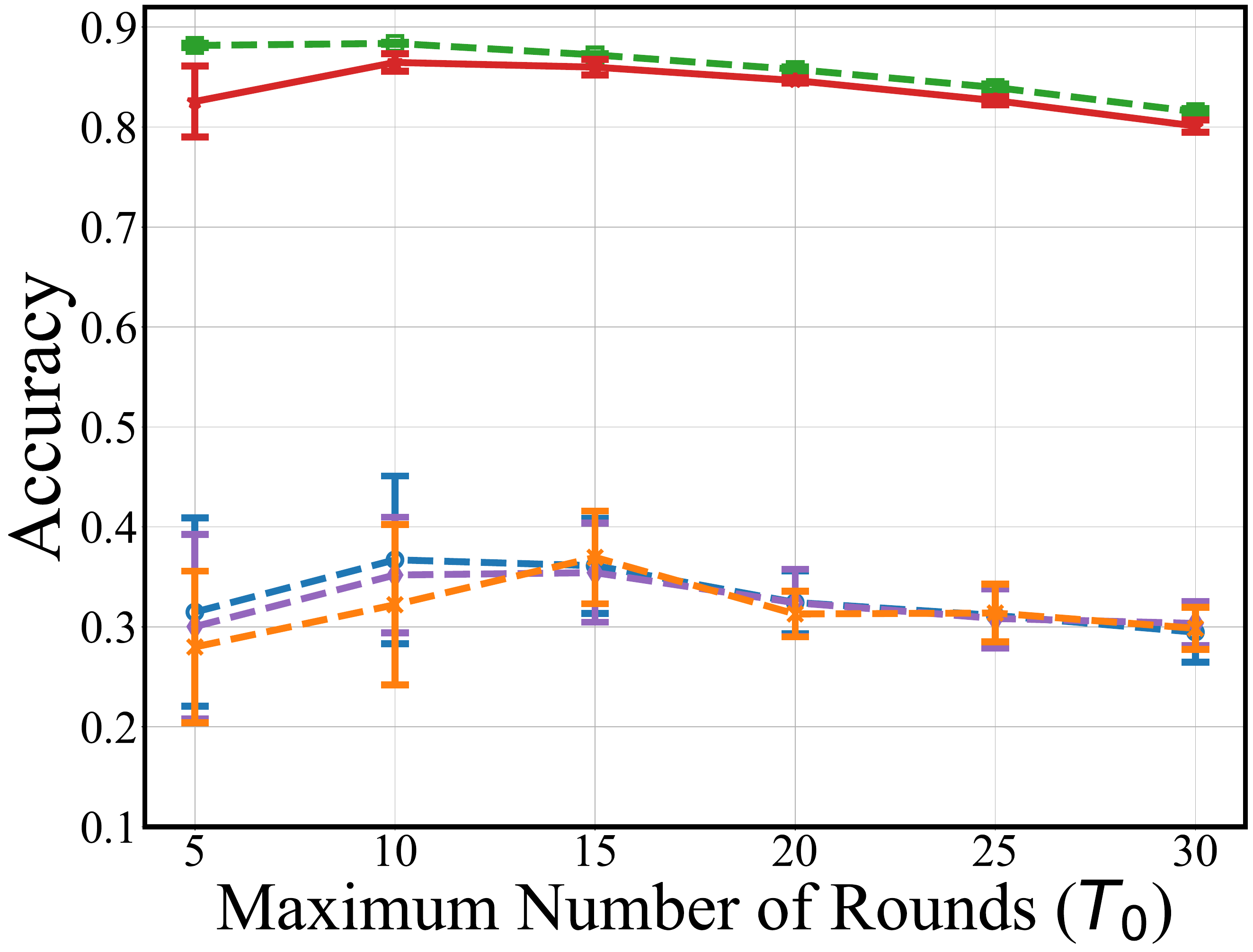}} 
\hspace{-1mm}
\subfigure[{Accuracy vs. $T_0$ (MLR, MNIST)}]{
\label{allocation_acc_mnist_mlr}
\includegraphics[width=0.23\textwidth]{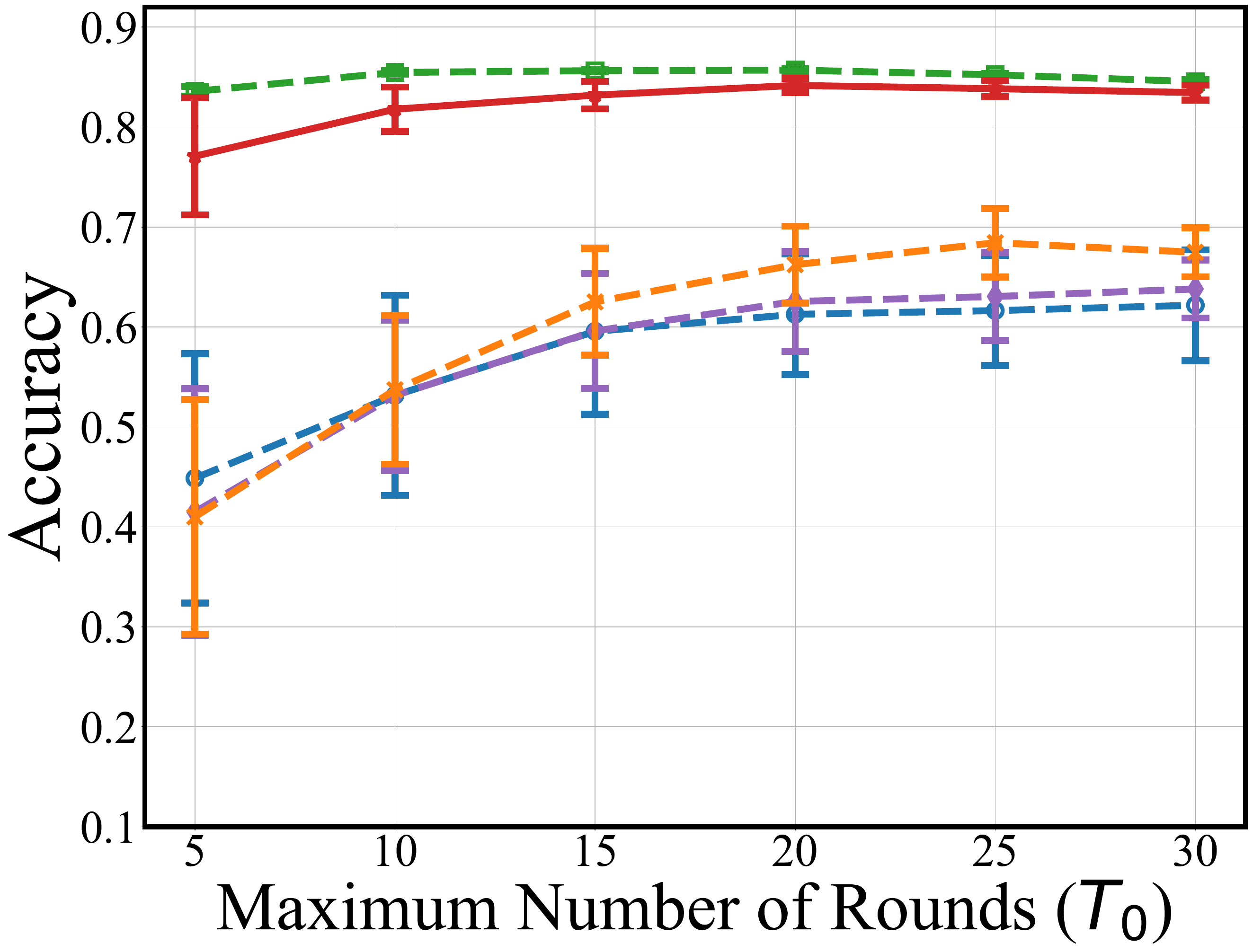}}
\\
\subfigure[Accuracy vs. $T_0$ (CNN,CIFAR10)]{
\label{allocation_acc_Cifar10_cnn}
\includegraphics[width=0.23\textwidth]{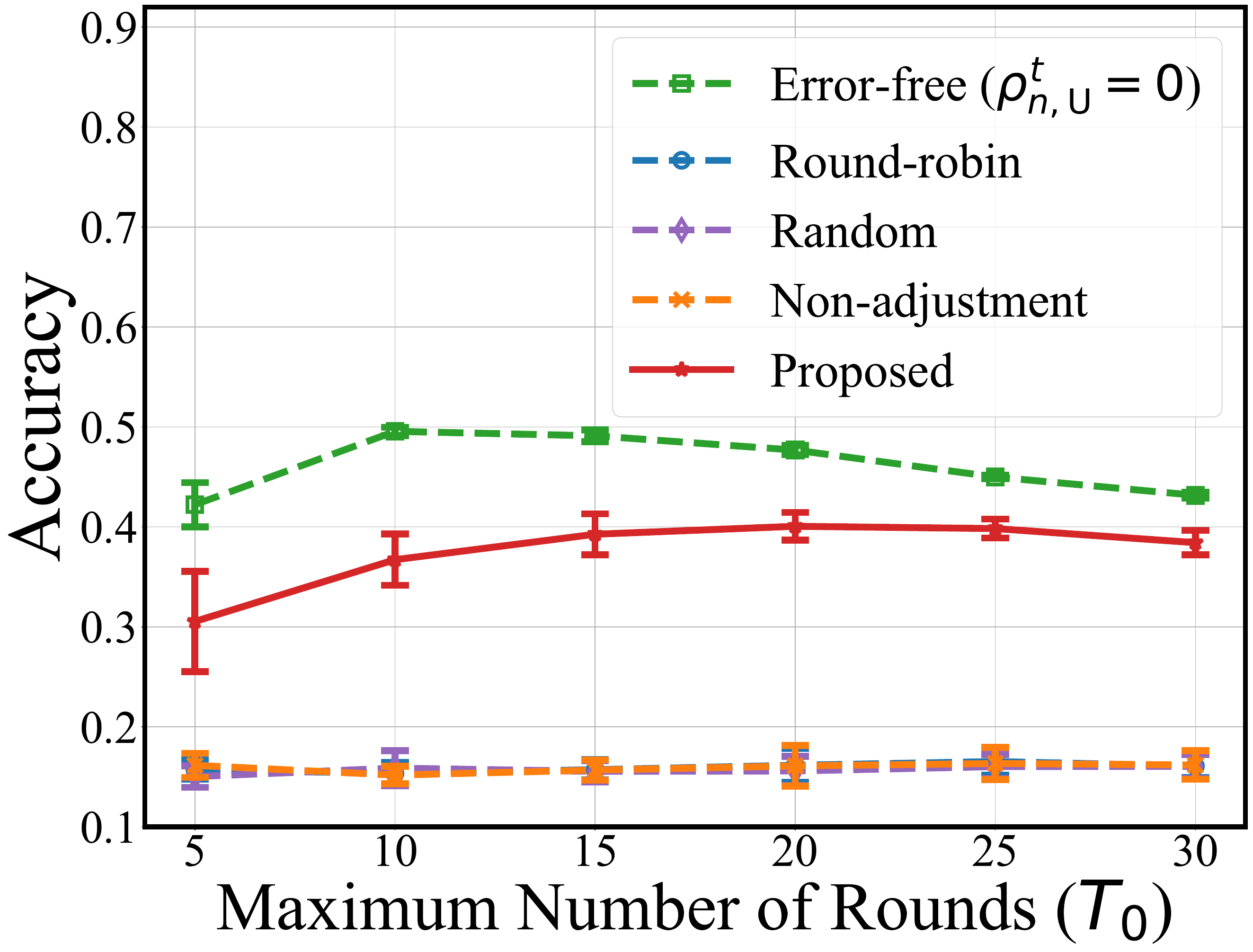}}
\hspace{-1mm}
\subfigure[Accuracy vs. $T_0$ (CNN, FMNIST)]{
\label{allocation_acc_fmnist_cnn}
\includegraphics[width=0.23\textwidth]{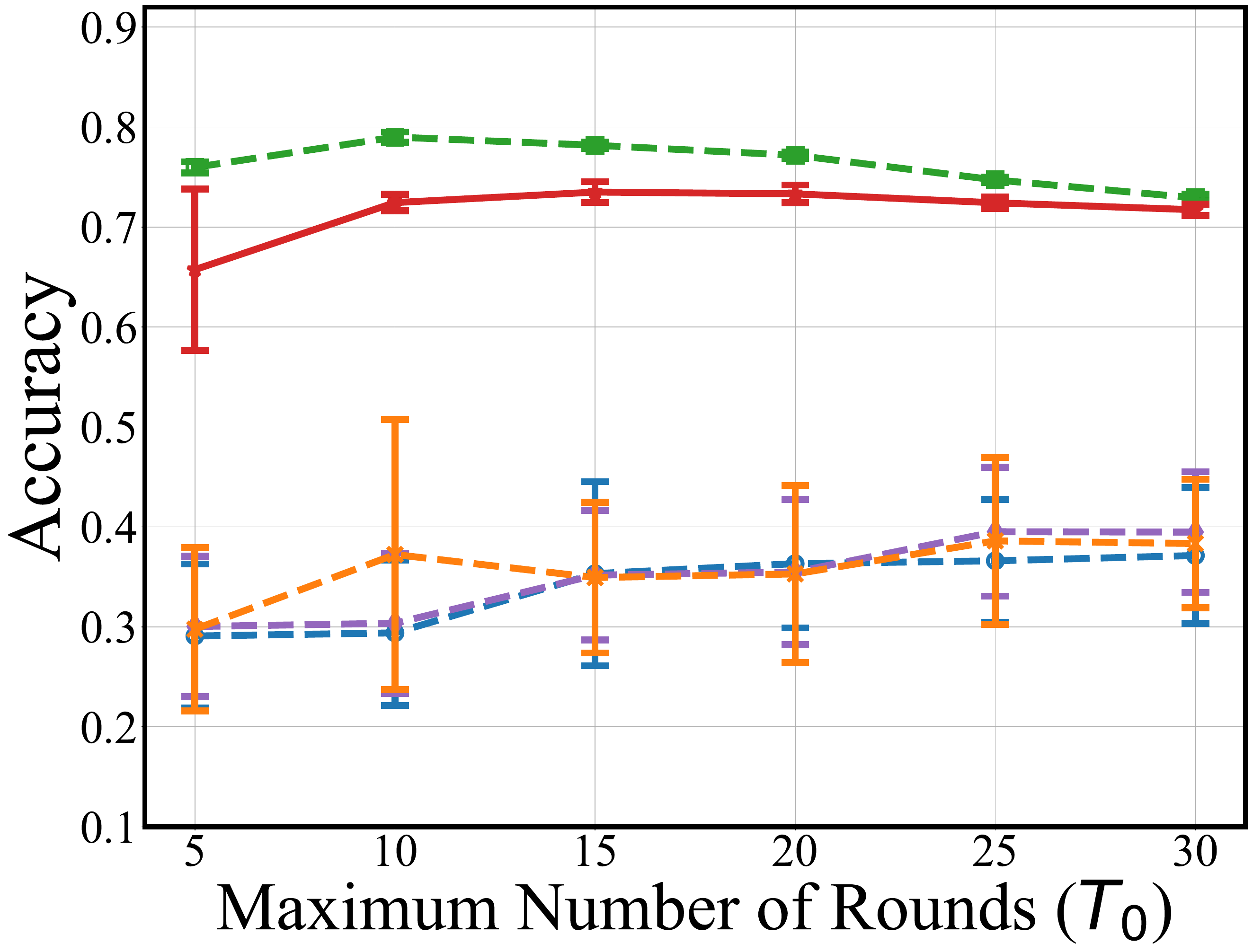}}
\caption{{Accuracy of the PL model concerning the maximum number of rounds $T_{0}$ under different allocation methods. $\epsilon_{\mathrm{Q}}=1$. $\delta_{\mathrm{Q}}=0.001$ for DNN and MLR, and $\delta_{\mathrm{Q}}=0.005$ for CNN.}}
\label{allocation_acc}
\end{figure}

\begin{figure}[t]
\centering  
\subfigure[{Fairness vs. $T_0$ (DNN, MNIST)}]
{
\label{allocation_jain_mnist_dnn}
\includegraphics[width=0.23\textwidth]{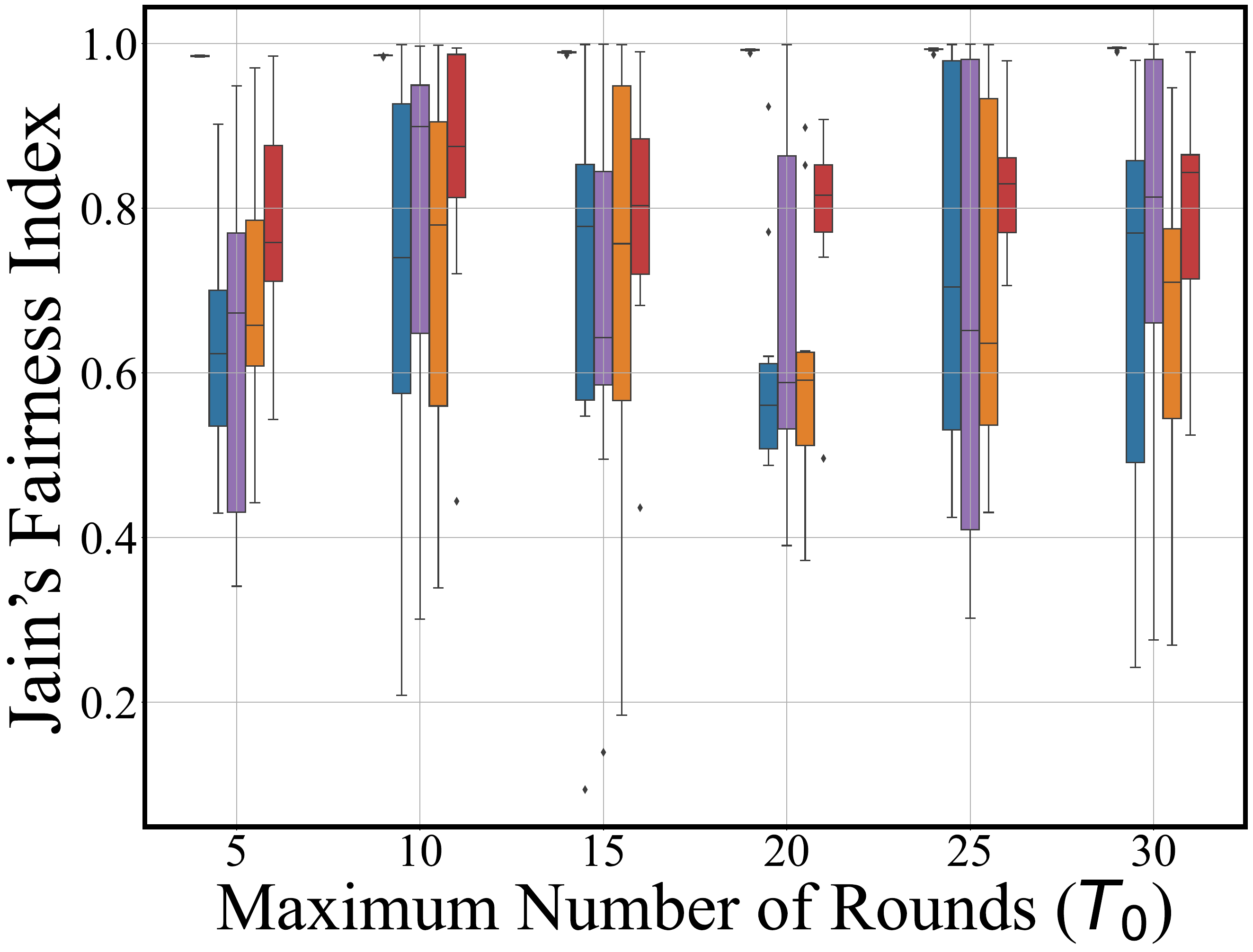}} 
\hspace{-1mm}
\subfigure[Maximum test loss vs. $T_0$ (DNN, MNIST)]
{
\label{allocation_worst_loss_mnist_dnn}
\includegraphics[width=0.23\textwidth]{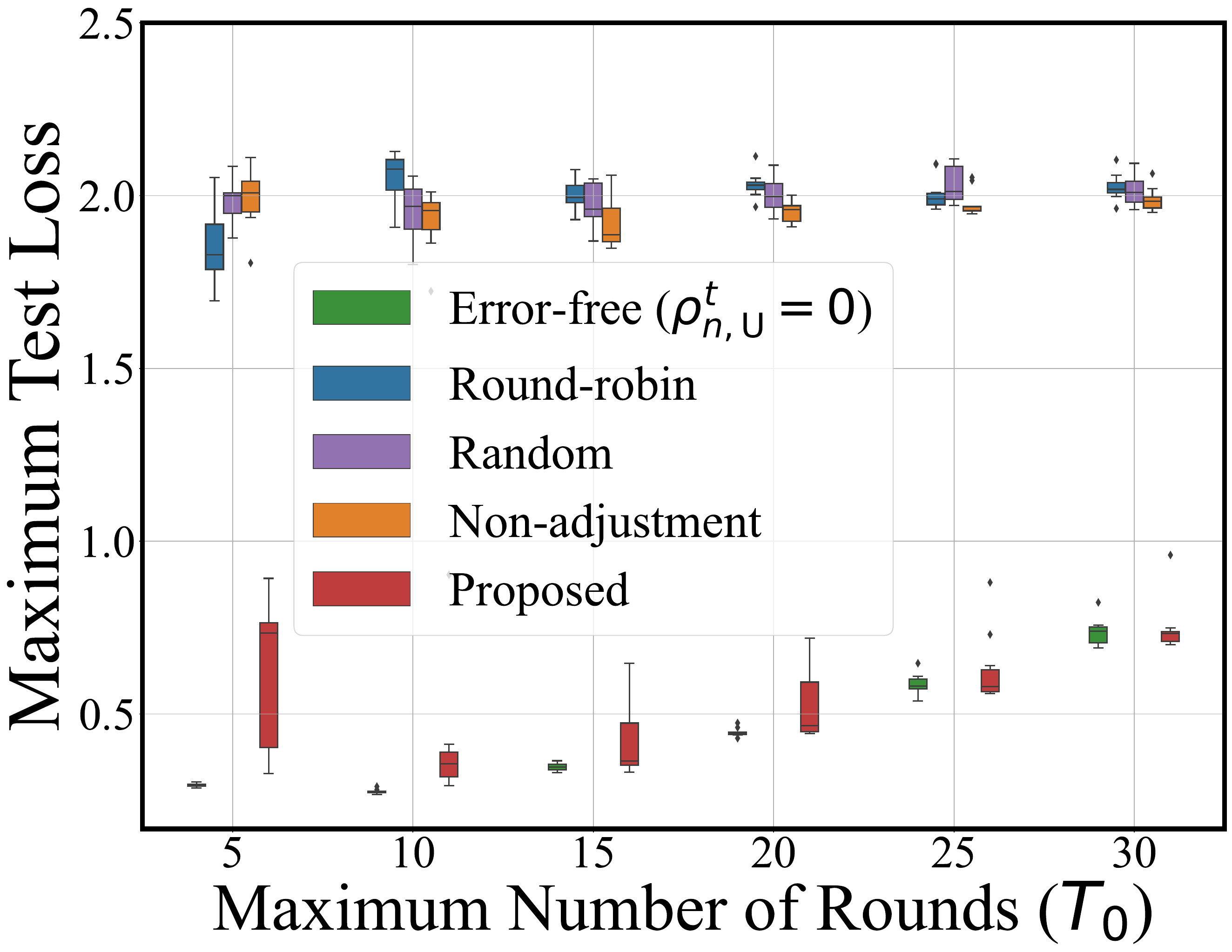}} 
\\
\subfigure[{Fairness vs. $T_0$ (MLR, MNIST)}]{
\label{allocation_jain_mnist_mlr}
\includegraphics[width=0.23\textwidth]{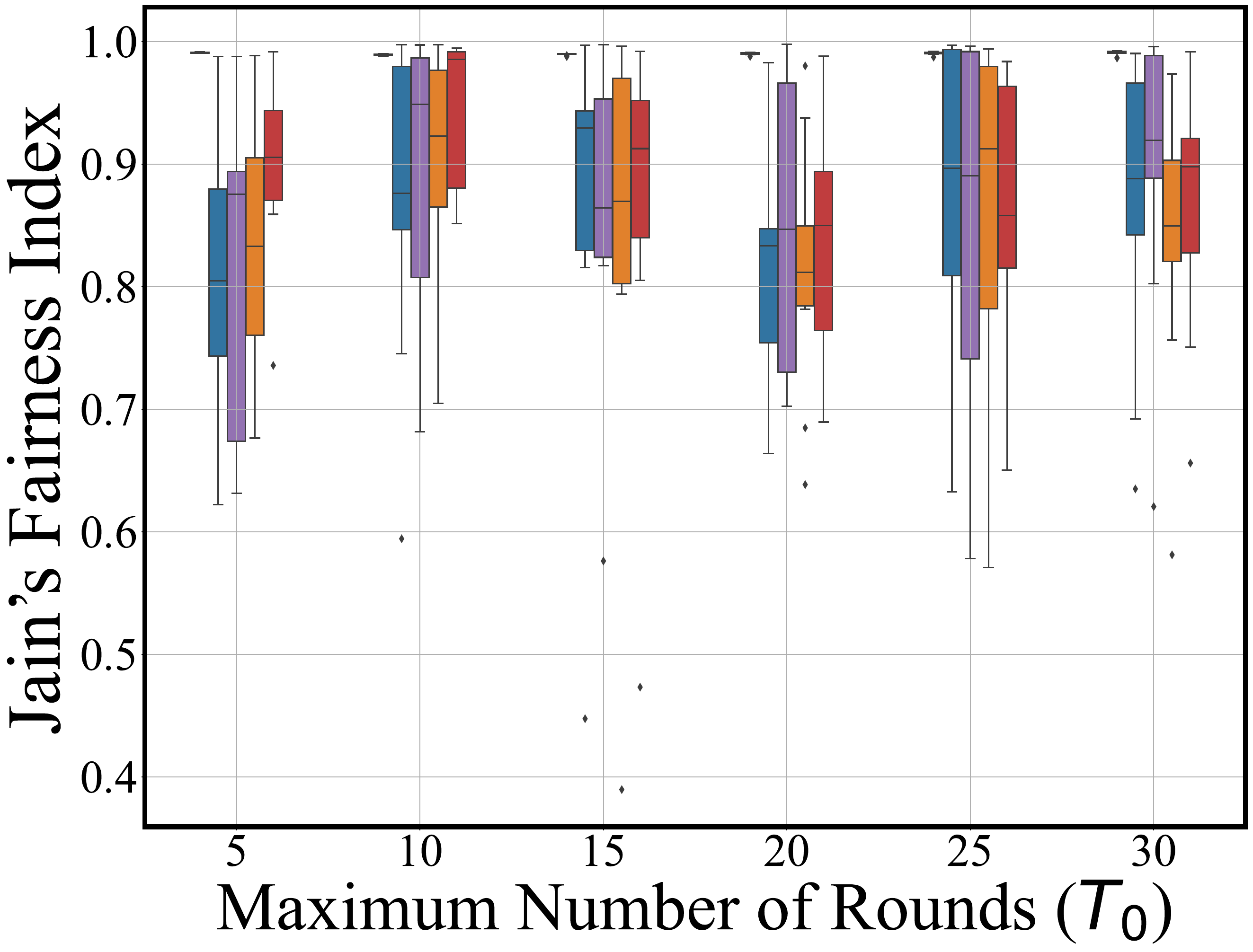}}
\hspace{-1mm}
\subfigure[{Maximum test loss vs. $T_0$ (MLR, MNIST)}]
{
\label{allocation_worst_loss_mnist_mlr}
\includegraphics[width=0.23\textwidth]{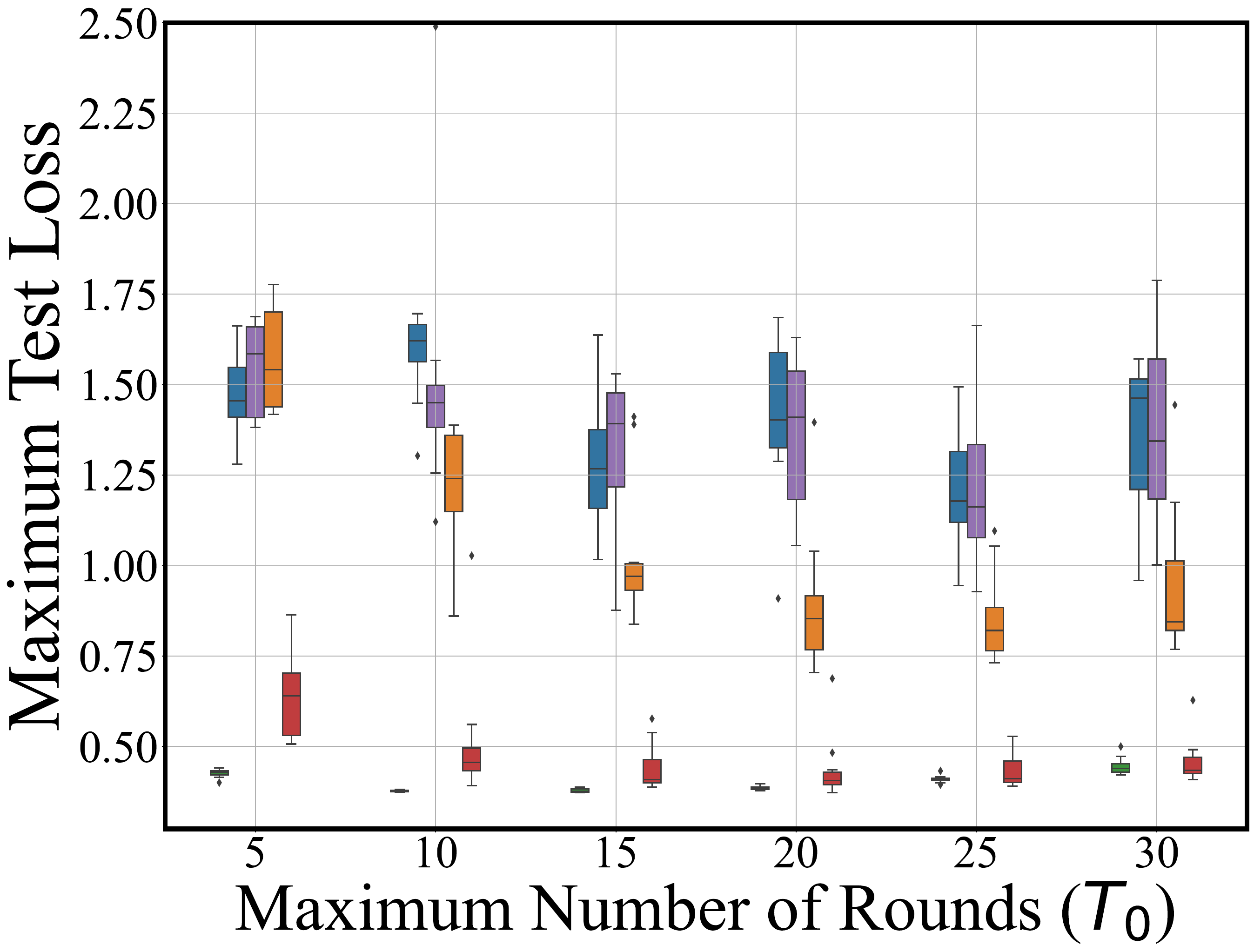}} 
\\
\subfigure[Fairness vs. $T_0$ (CNN,CIFAR10)]{
\label{allocation_jain_Cifar10_cnn}
\includegraphics[width=0.23\textwidth]{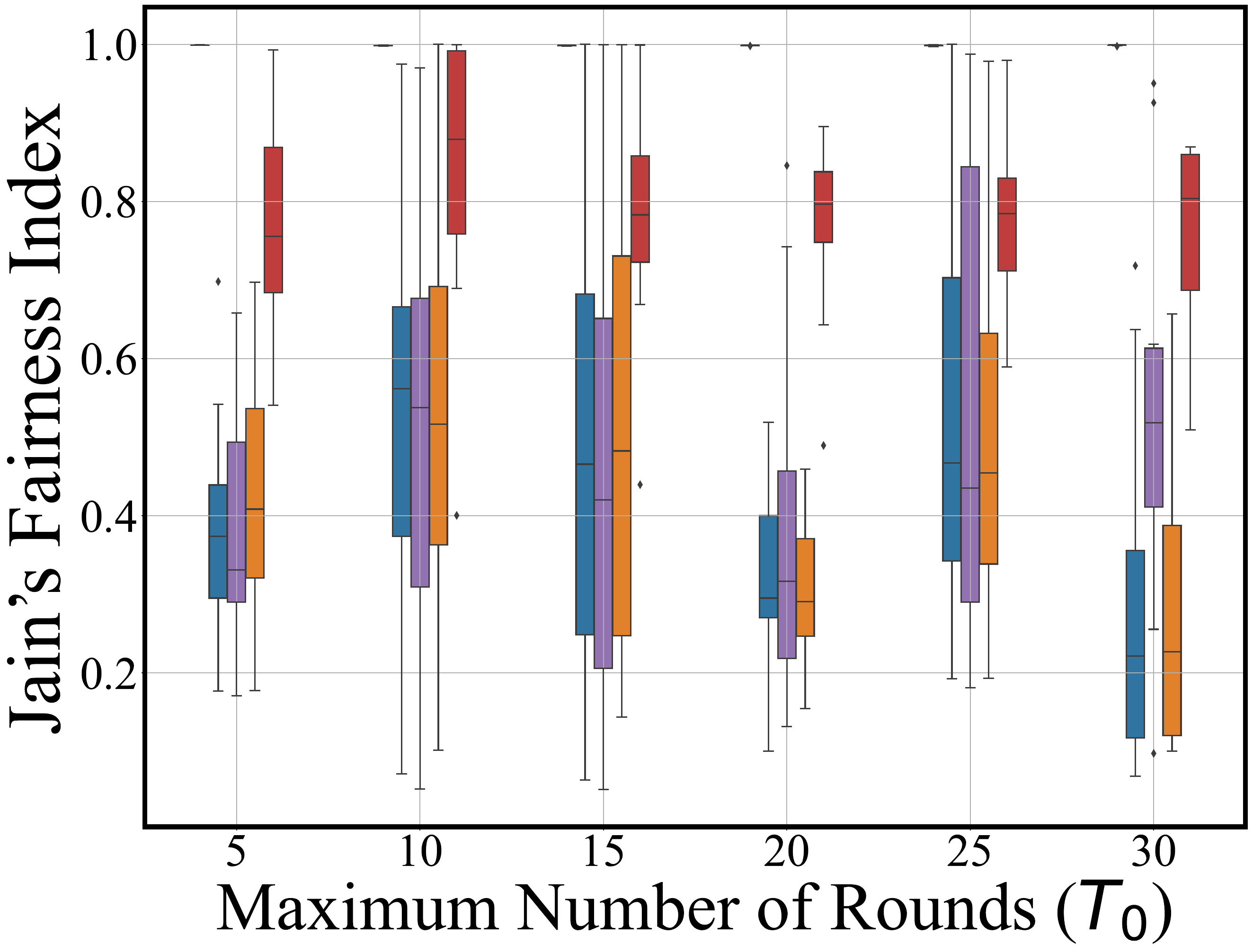}}
\hspace{-1mm}
\subfigure[Maximum test loss vs. $T_0$ (CNN,CIFAR10)]{
\label{allocation_worst_loss_Cifar10_cnn}
\includegraphics[width=0.23\textwidth]{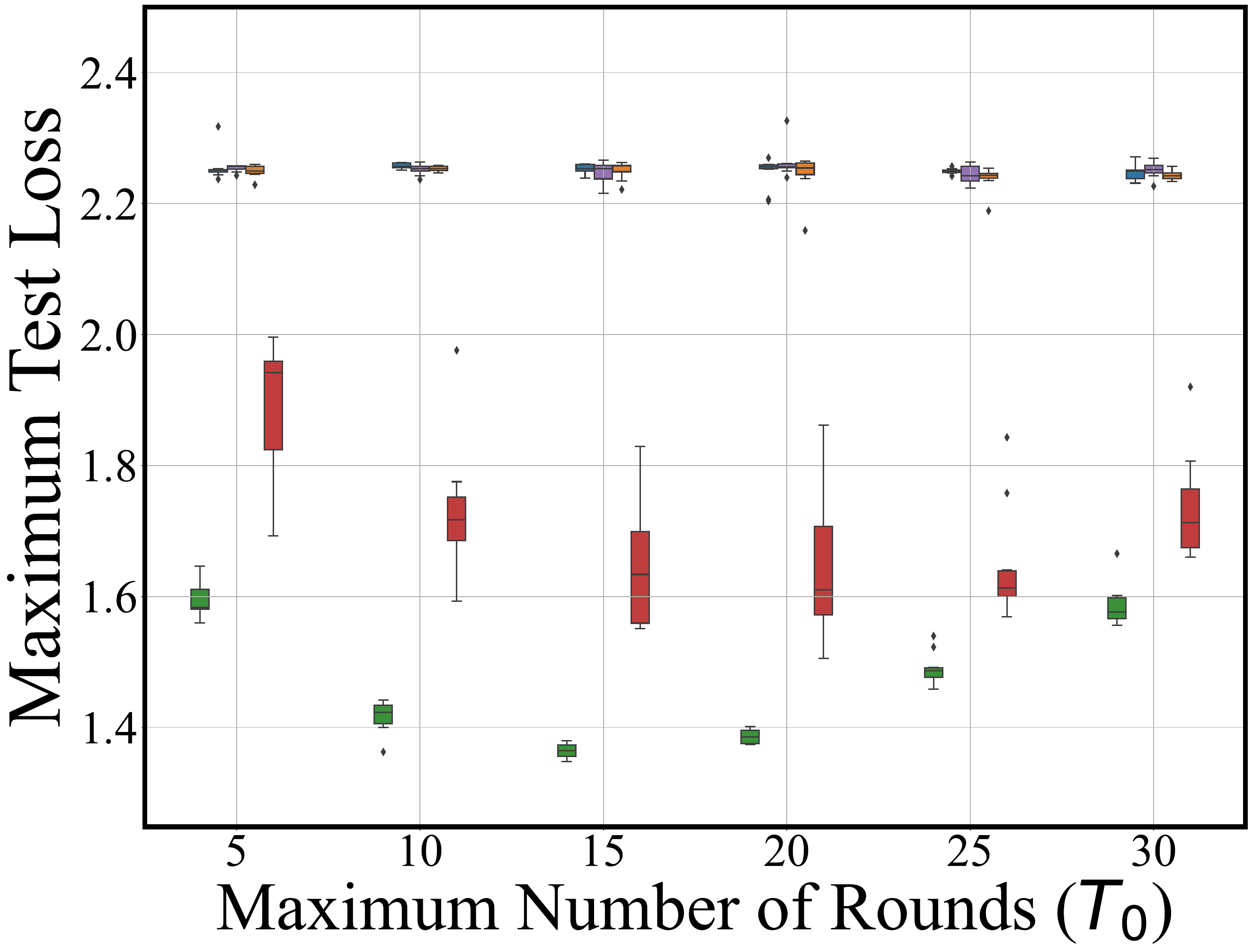}}
\\
\subfigure[Fairness vs. $T_0$ (CNN, FMNIST)]{
\label{allocation_jain_fmnist_cnn}
\includegraphics[width=0.23\textwidth]{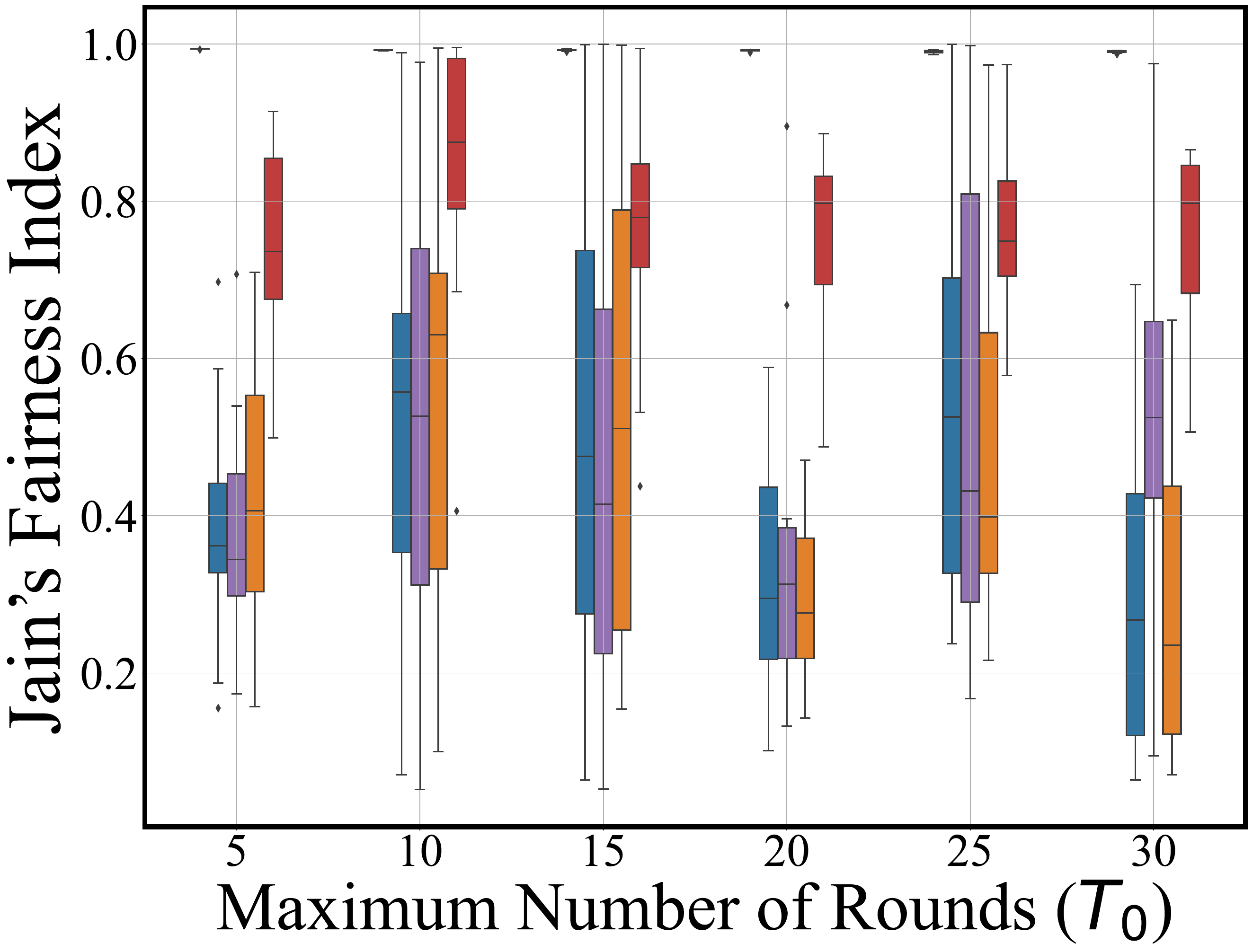}}
\hspace{-1mm}
\subfigure[{Maximum test loss} vs. $T_0$ (CNN,FMNIST)]{
\label{allocation_worst_loss_fmnist_cnn}
\includegraphics[width=0.23\textwidth]{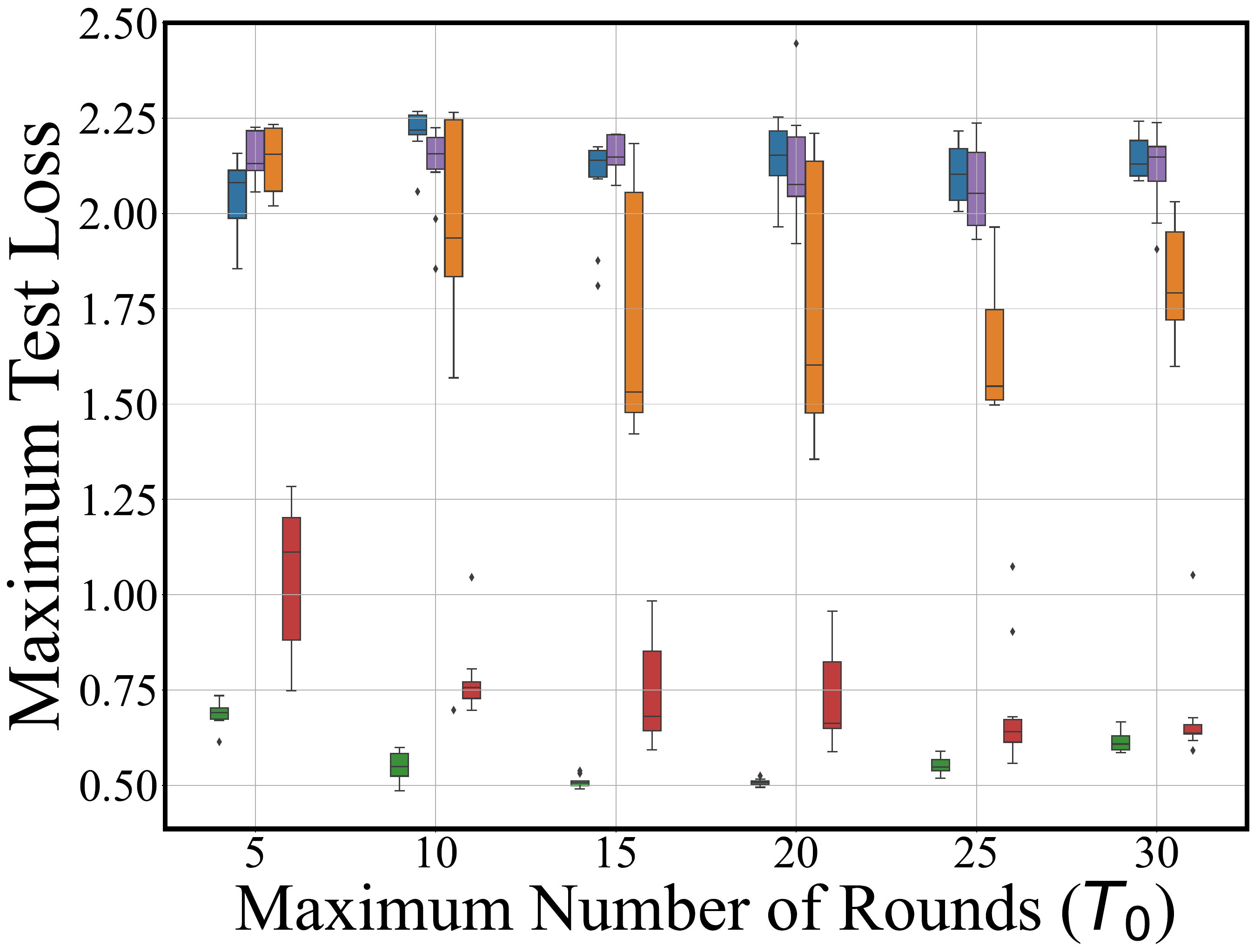}}
\caption{Comparison of fairness and {the maximum test loss of all participating clients} between the benchmarks. $\epsilon_{\mathrm{Q}}=1$. $\delta_{\mathrm{Q}}=0.001$ for DNN and MLR, and $\delta_{\mathrm{Q}}=0.005$ for CNN.}
\label{allocation_fair_loss}
\end{figure}

Figs.~\ref{allocation_acc} and \ref{allocation_fair_loss} plot the accuracy, fairness (i.e., Jain's fairness index), and maximum test loss of all participating clients among the clients of privacy-preserving WPFL. {In addition to imperfect channels, we consider the situation of the proposed allocation and configuration scheme running in an error-free channel (i.e., $\rho_{n,\mathrm{L}}^{t}=0$), which provides the best possible results of the proposed allocation and configuration scheme.}

As shown in Figs.~\ref{allocation_acc}, \ref{allocation_worst_loss_mnist_dnn}, \ref{allocation_worst_loss_mnist_mlr}, \ref{allocation_worst_loss_Cifar10_cnn}, and \ref{allocation_worst_loss_fmnist_cnn}, the proposed scheduling policy outperforms the other benchmarks, in accuracy and maximum test loss of all participating clients.
Particularly, it is better than the second-best (including round-robin, random selection, and non-adjustment) by $87.08\%$ in accuracy and $16.21\%$ in the maximum test loss of participating clients under the CNN model, and $52.26\%$ in accuracy and $15.99\%$ in the maximum test loss under the DNN and MLR models, respectively.
{The proposed scheduling and configuration policy differs marginally from the best possible results achieved under an error-free channel.} This is because our device selection and adaptive coefficient adjustment take into account time-varying errors from DP, quantization, and transmission, and minimize the maximum of the convergence upper bound among the clients, reducing the impact from DP and noisy channels. The accuracy first increases and then decreases as $T_0$ grows, and the maximum test loss of all participating clients first decreases and then increases for WPFL under the proposed policy {running in either noisy or error-free channels}, due to the accumulated impact of the DP noise.

Figs.~\ref{allocation_jain_mnist_dnn}, \ref{allocation_jain_mnist_mlr}, \ref{allocation_jain_Cifar10_cnn}, and \ref{allocation_jain_fmnist_cnn} examine the fairness of WPFL through Jain's fairness index $\mathcal{J}=\frac{(\sum_{n=1}^{N}x_n)^2}{n\sum_{n=1}^{N}x_n^2}$, with $x_n$ being the training loss of client $n$. 
Under the CNN model, our approach is better than the second-best (i.e., Round-Robin) by $38.37\%$ in fairness (Jain's index).
Compared to the benchmarks (i.e., Round-Robin, Random, and Non-Adjustment), the proposed configuration and scheduling policy is substantially fairer, thanks to its consideration of the min-max fairness. 
% The proposed scheme is only worse than that with error-free channels. This is because our policy optimizes client selection, channel allocation, power control, and coefficient adjustment, aiming at accelerating convergence and improving performance in a fair fashion.

\subsubsection{Comparison with state-of-the-art PFL}

\begin{figure}[t]
\centering  
\subfigure[{Accuracy vs. $t$ (DNN, MNIST)}]
{
\label{PFL_acc_mnist_dnn_convergence}
\includegraphics[width=0.23\textwidth]{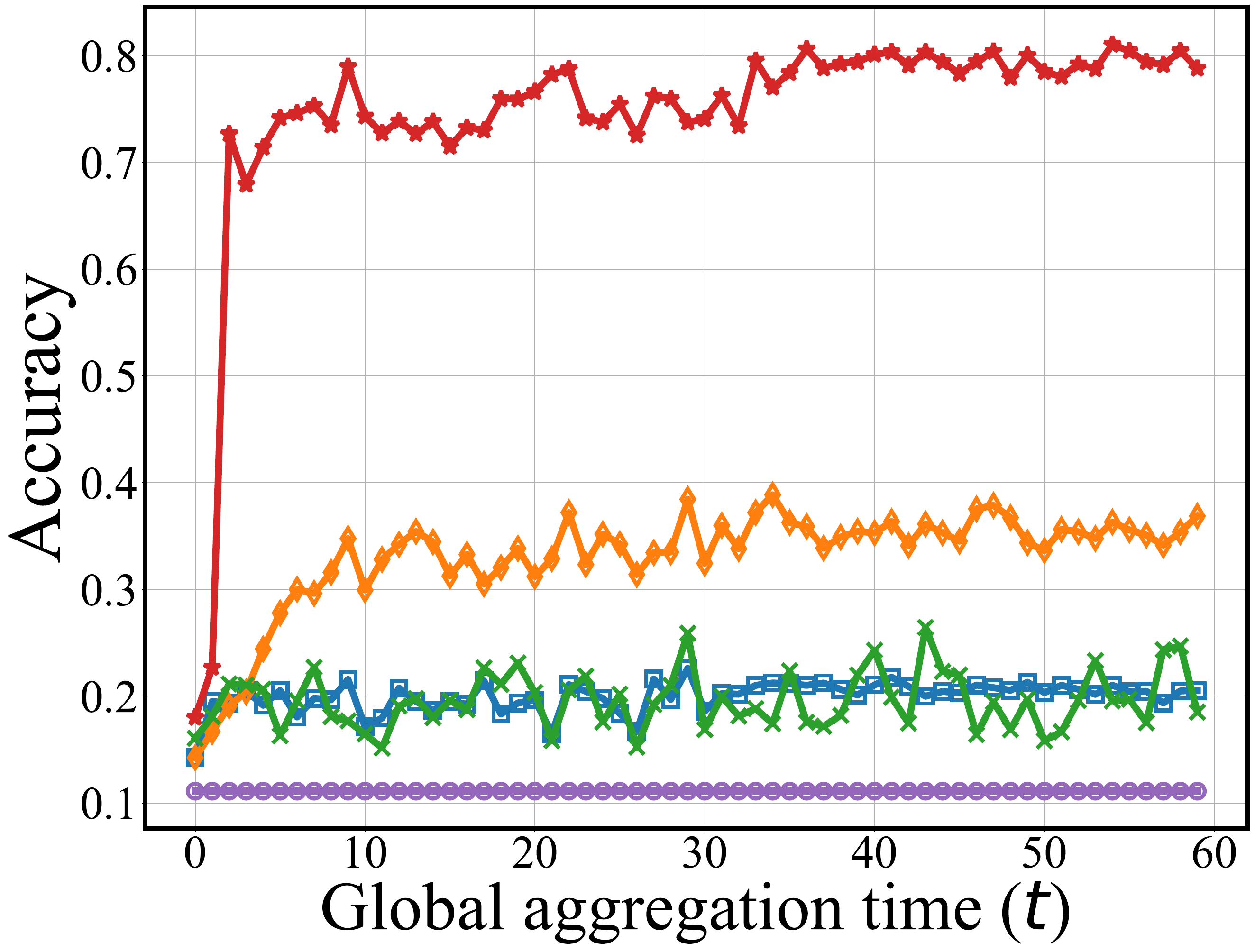}}
\hspace{-1mm}
\subfigure[Accuracy vs. $t$ (MLR, MNIST)]
{
\label{PFL_acc_mnist_mlr_convergence}
\includegraphics[width=0.23\textwidth]{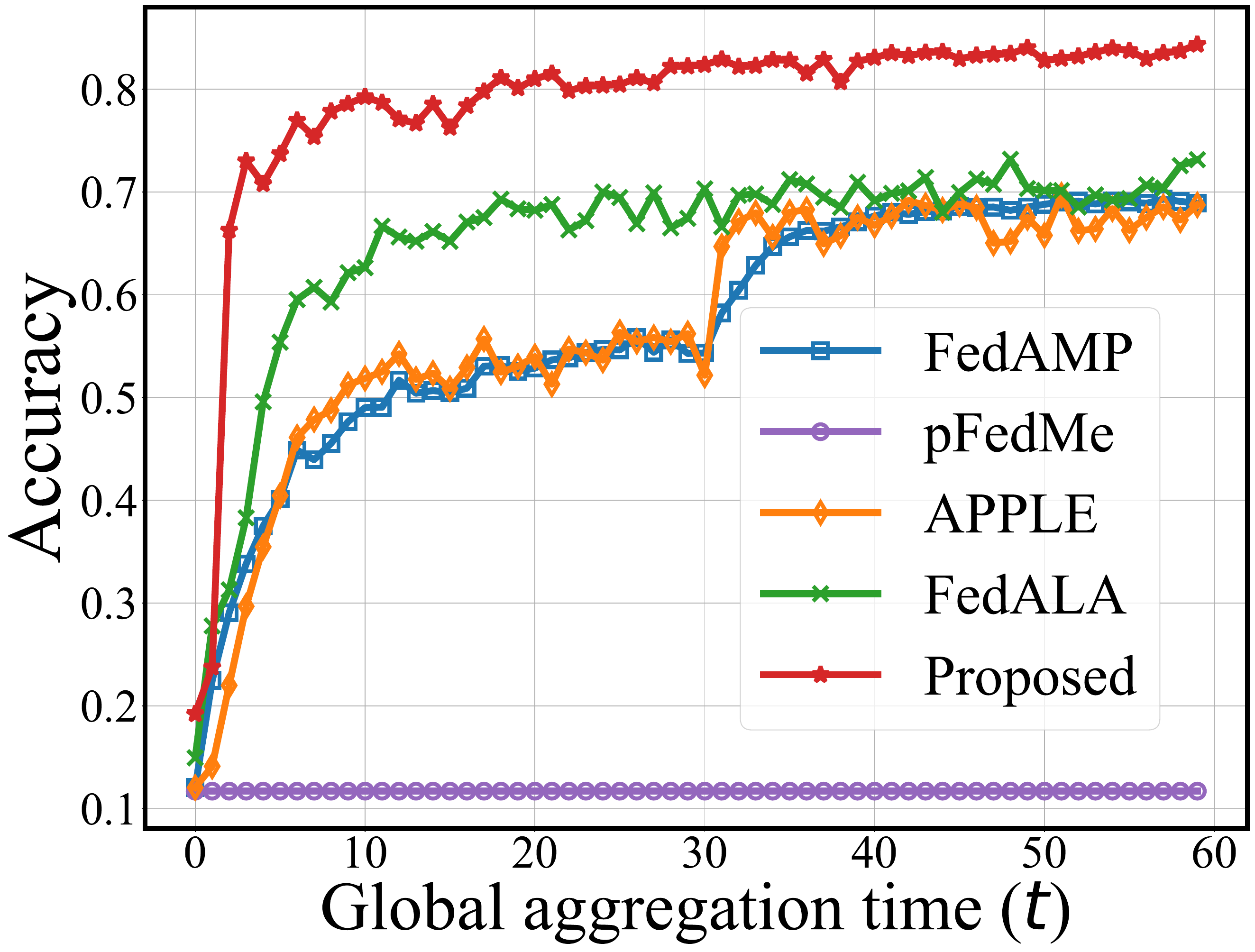}}
\\
\subfigure[Accuracy vs. $t$ (CNN,CIFAR10)]
{
\label{PFL_acc_Cifar10_cnn_convergence}
\includegraphics[width=0.23\textwidth]{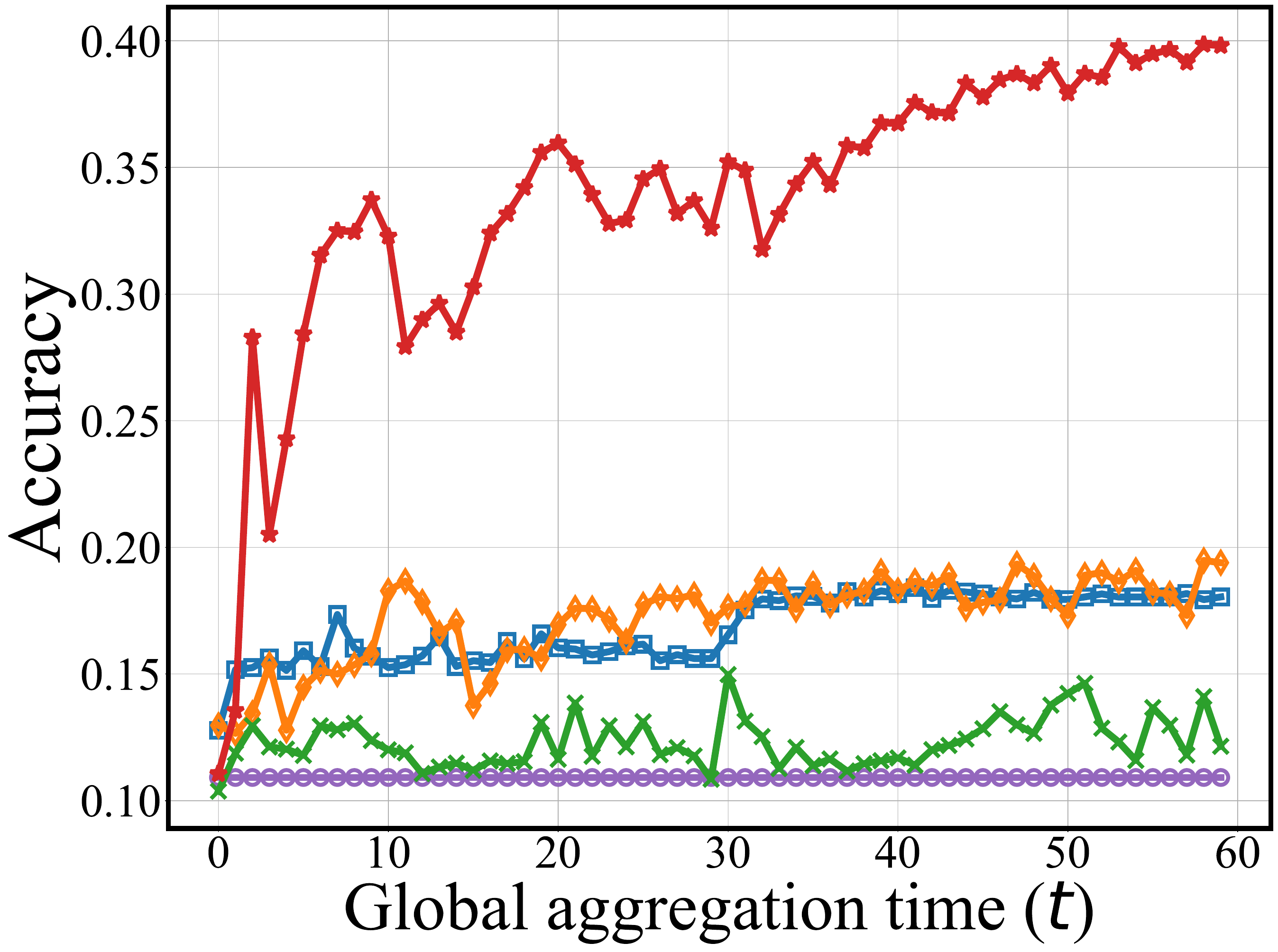}} 
\hspace{-1mm}
\subfigure[Accuracy vs. $t$ (CNN, FMNIST)]
{
\label{PFL_acc_fmnist_cnn_convergence}
\includegraphics[width=0.23\textwidth]{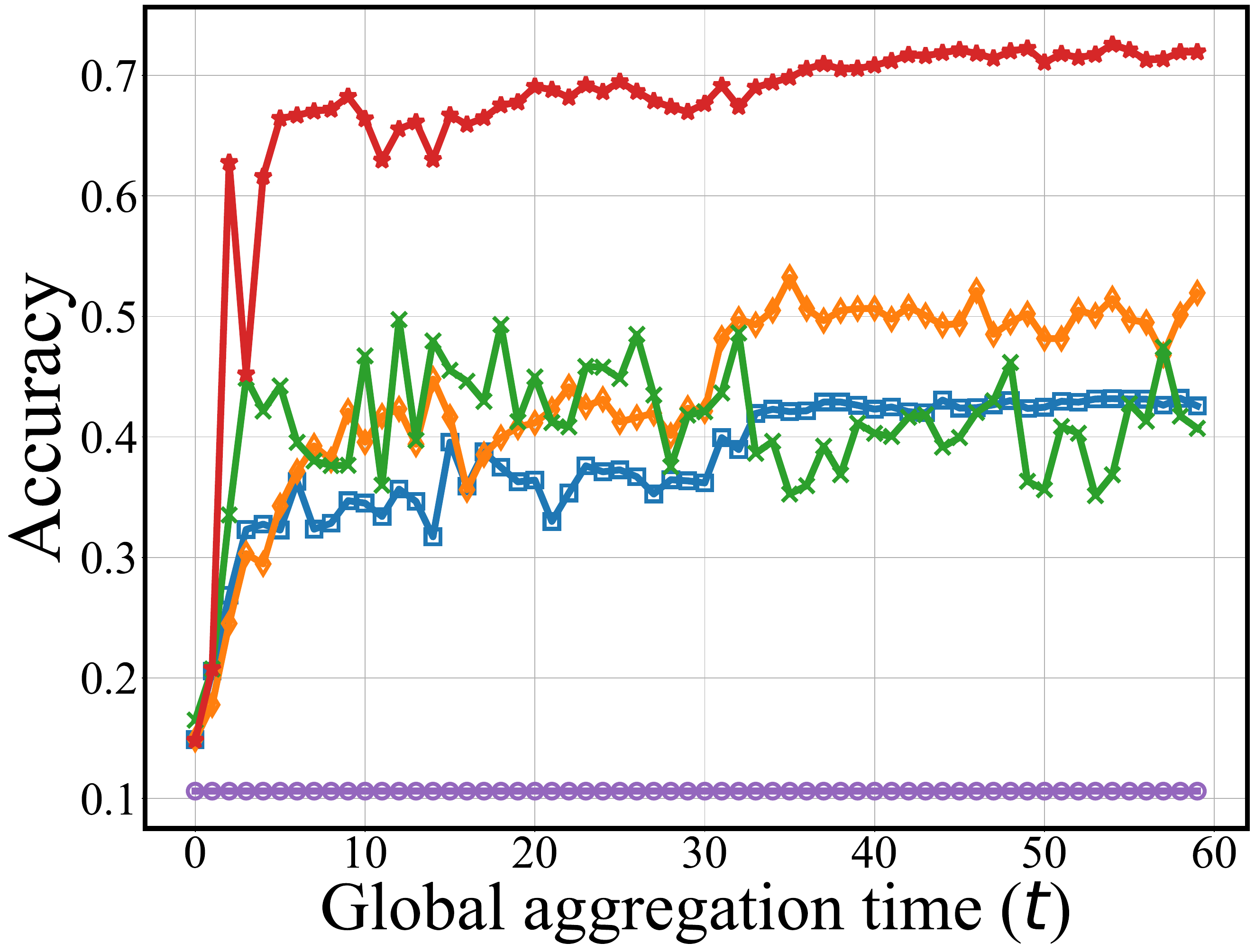}} 
\caption{{Comparison of accuracy vs. aggregation time $t$ between the benchmarks. $\epsilon_{\mathrm{Q}}=1$. $\delta_{\mathrm{Q}}=0.001$ for DNN and MLR, and $\delta_{\mathrm{Q}}=0.005$ for CNN.}}
\label{PFL_acc_convergence}
\end{figure}

\begin{figure}[t]
\centering  
\subfigure[{Accuracy vs. $T_0$ (DNN, MNIST)}]
{
\label{PFL_acc_mnist_dnn}
\includegraphics[width=0.23\textwidth]{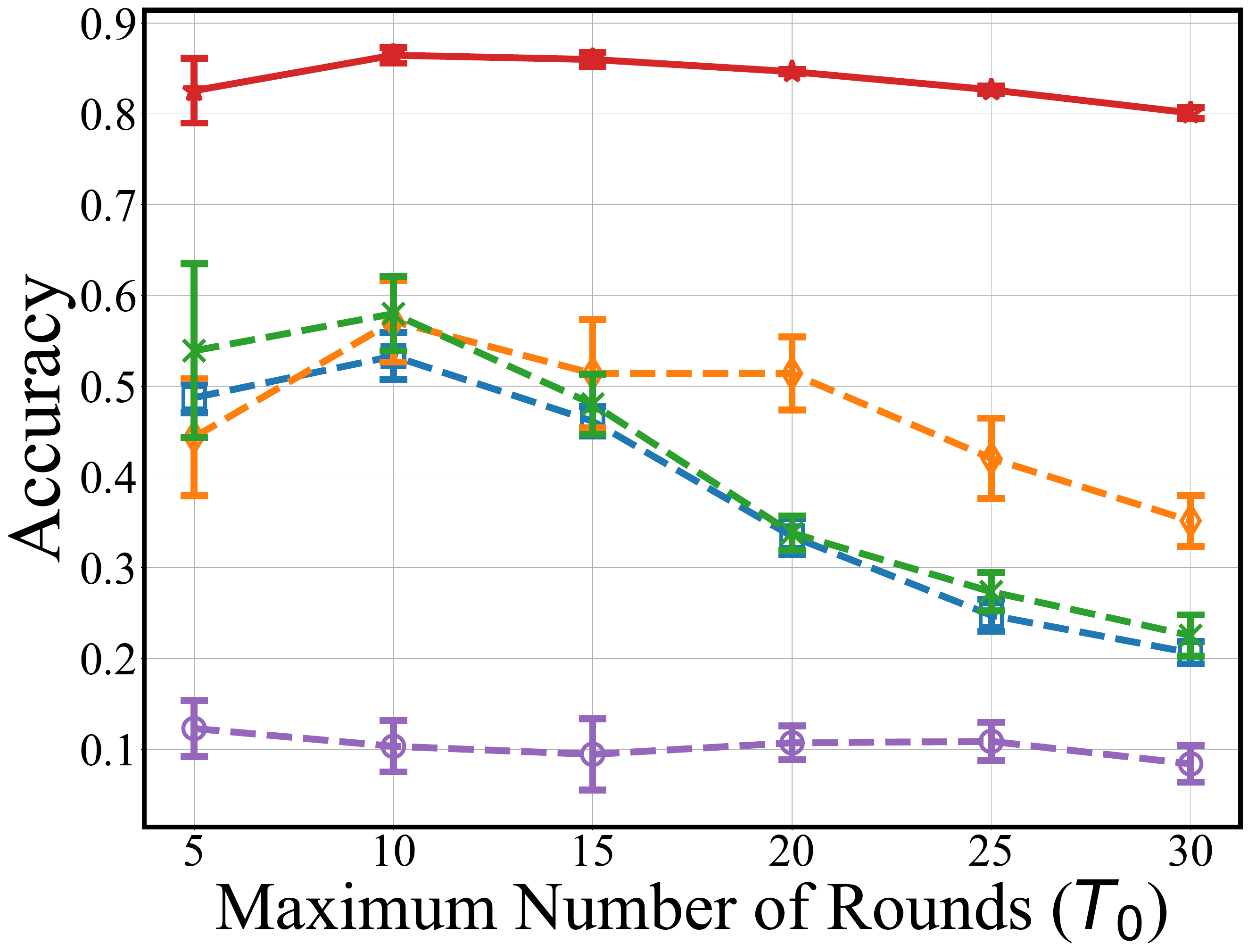}}
\hspace{-1mm}
\subfigure[Accuracy vs. $T_0$ (MLR, MNIST)]
{
\label{PFL_acc_mnist_mlr}
\includegraphics[width=0.23\textwidth]{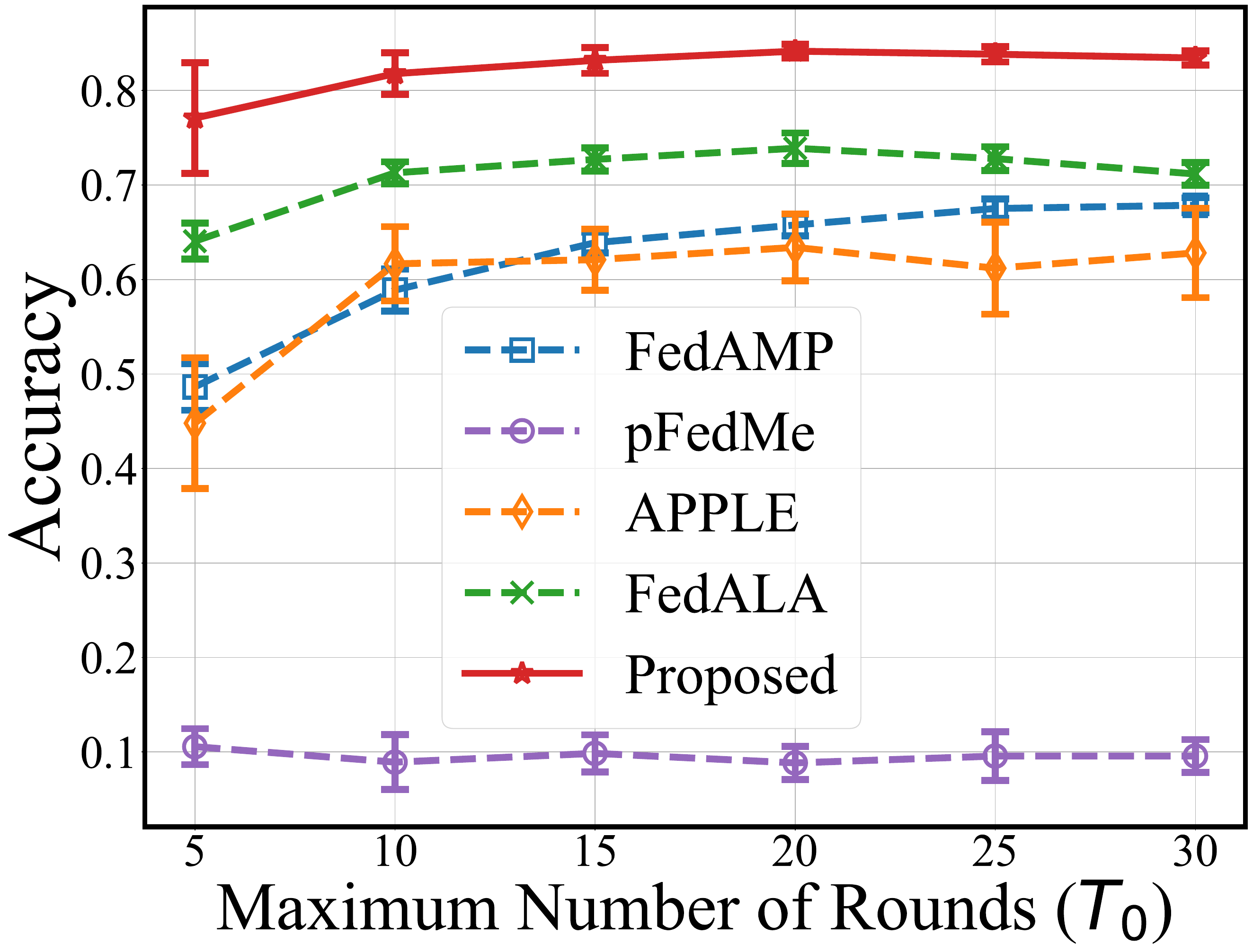}}
\\
\subfigure[Accuracy vs. $T_0$ (CNN,CIFAR10)]
{
\label{PFL_acc_Cifar10_cnn}
\includegraphics[width=0.23\textwidth]{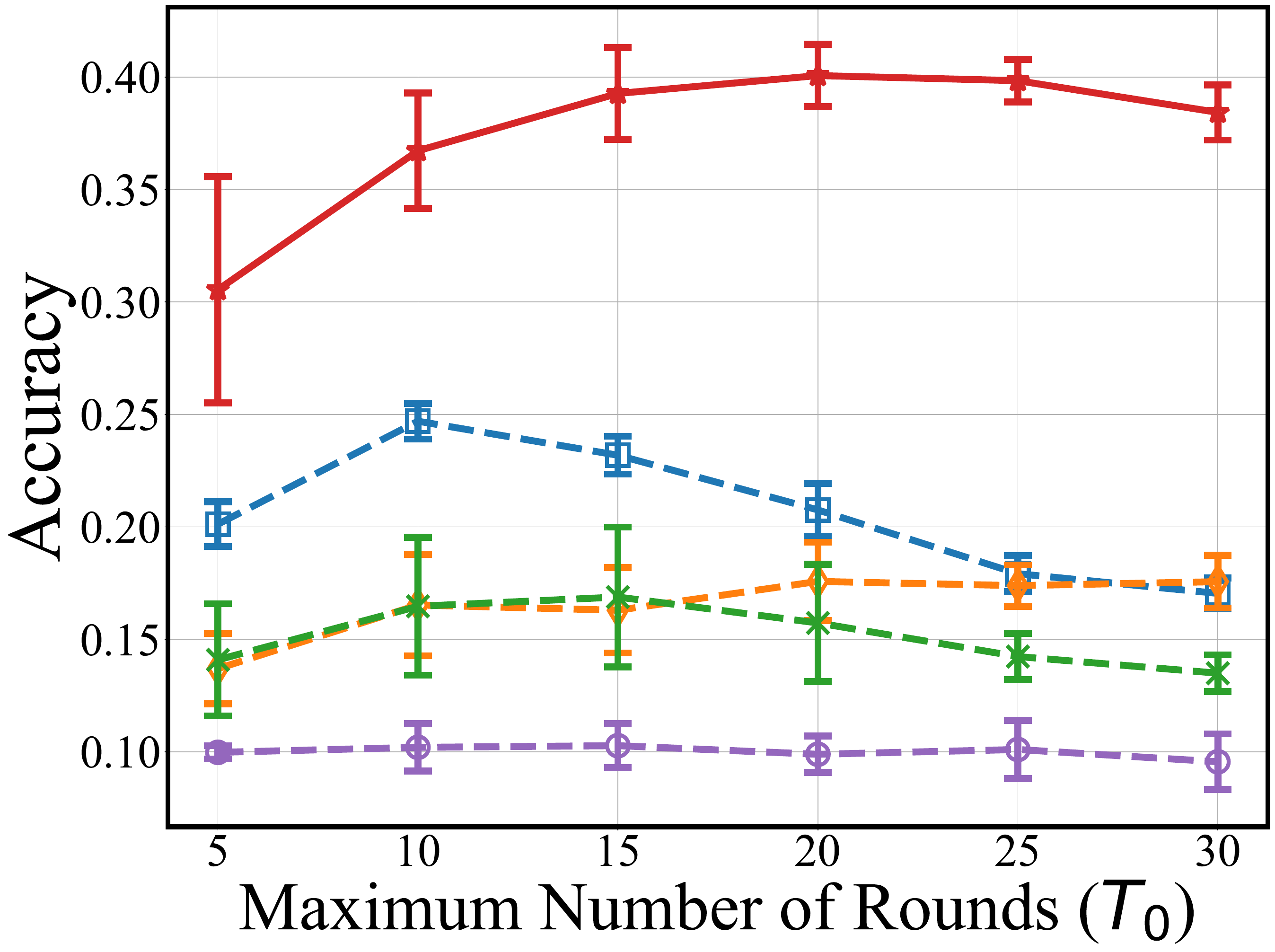}} 
\hspace{-1mm}
\subfigure[Accuracy vs. $T_0$ (CNN, FMNIST)]
{
\label{PFL_acc_fmnist_cnn}
\includegraphics[width=0.23\textwidth]{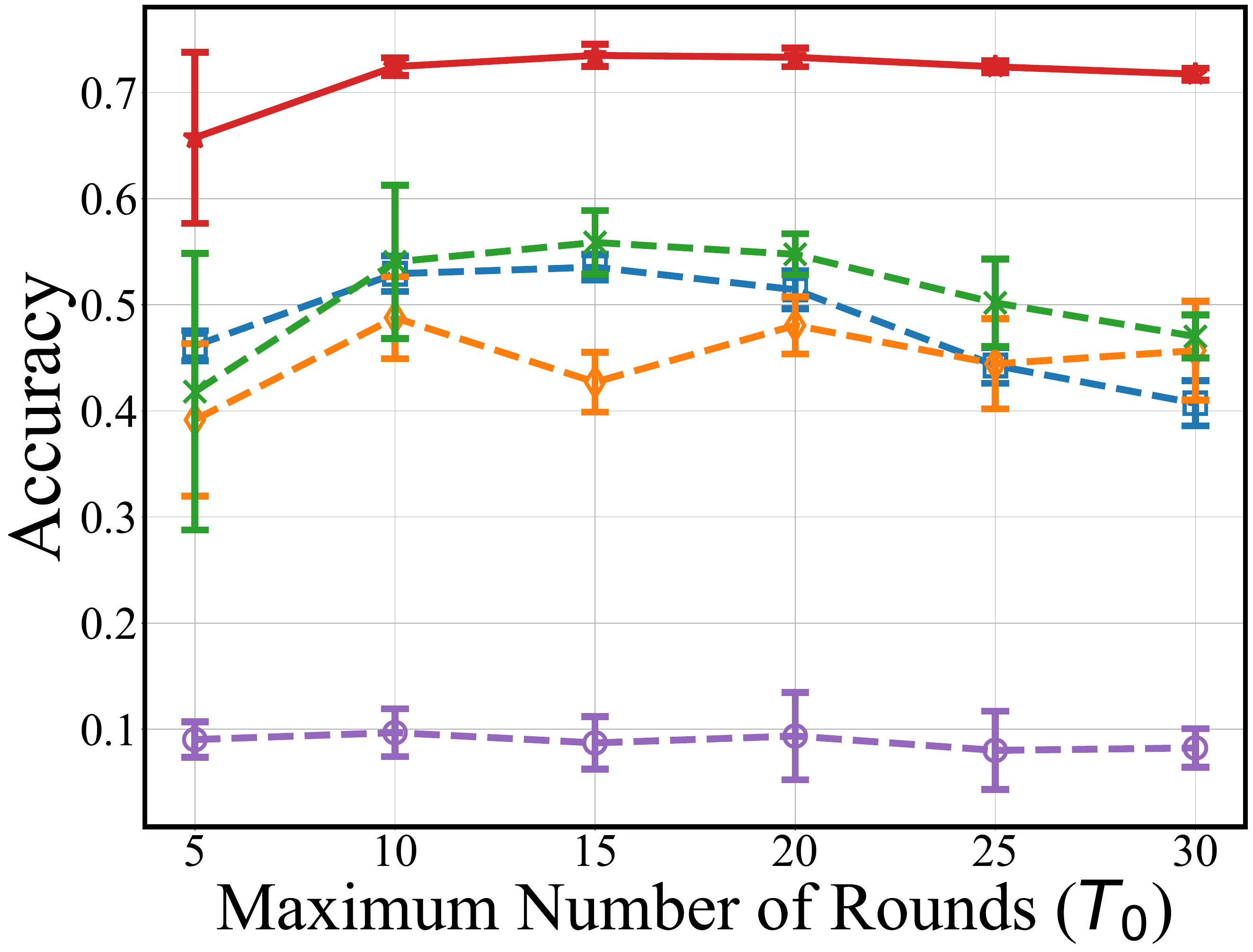}} 
\caption{{Comparison of testing accuracy between the benchmarks. $\epsilon_{\mathrm{Q}}=1$. $\delta_{\mathrm{Q}}=0.001$ for DNN and MLR, and $\delta_{\mathrm{Q}}=0.005$ for CNN.}}
\label{PFL_acc}
\end{figure}

\begin{figure}[t]
\centering  
\subfigure[Fairness vs. $T_0$ (DNN, MNIST)]
{
\label{PFL_fairness_mnist_dnn}
\includegraphics[width=0.23\textwidth]{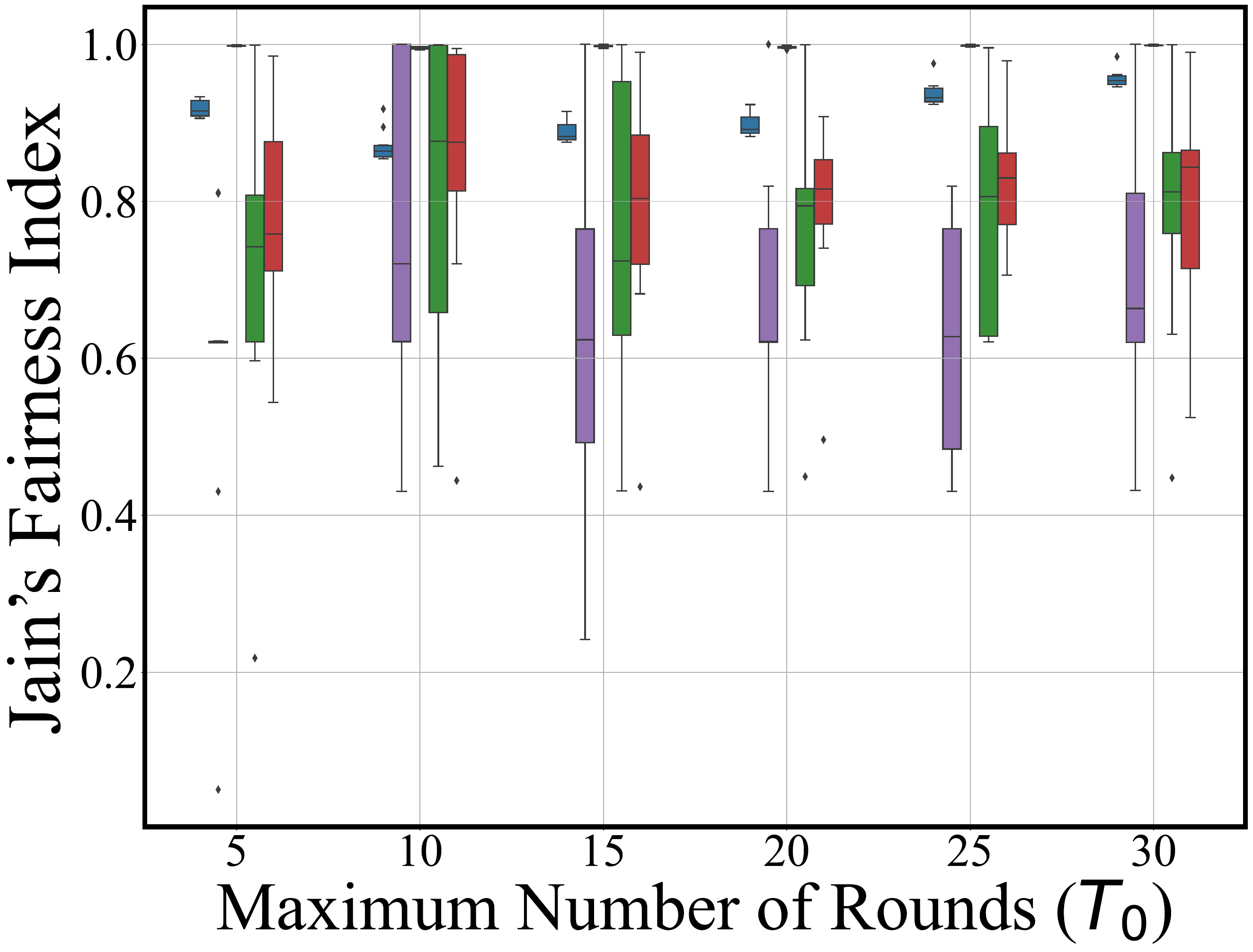}} 
\hspace{-1mm}
\subfigure[{Maximum test loss vs. $T_0$ (DNN, MNIST)}]
{
\label{PFL_worst_loss_mnist_dnn}
\includegraphics[width=0.23\textwidth]{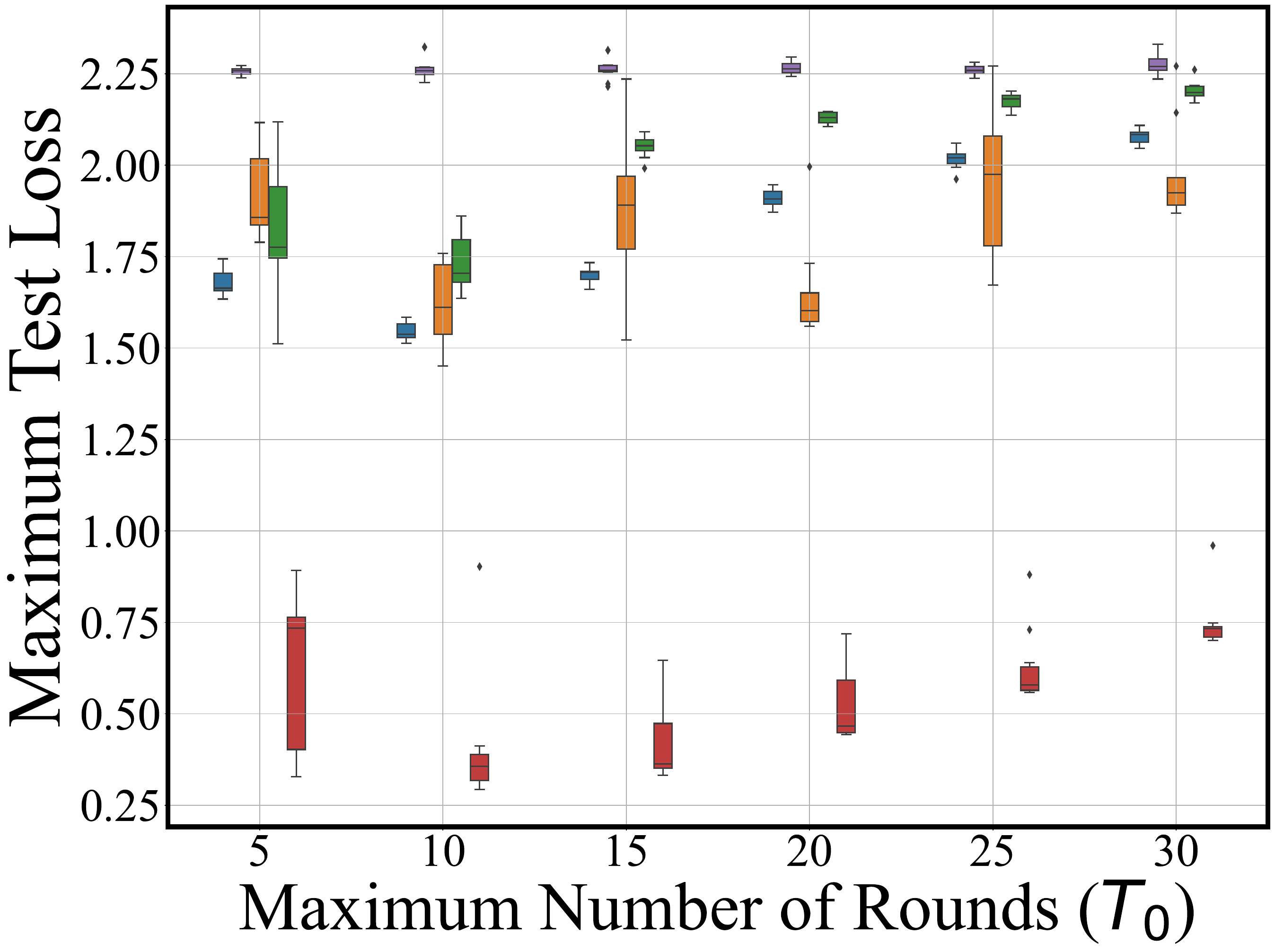}} 
\\
\subfigure[{Fairness vs. $T_0$ (MLR, MNIST)}]{
\label{PFL_fairness_mnist_mlr}
\includegraphics[width=0.23\textwidth]{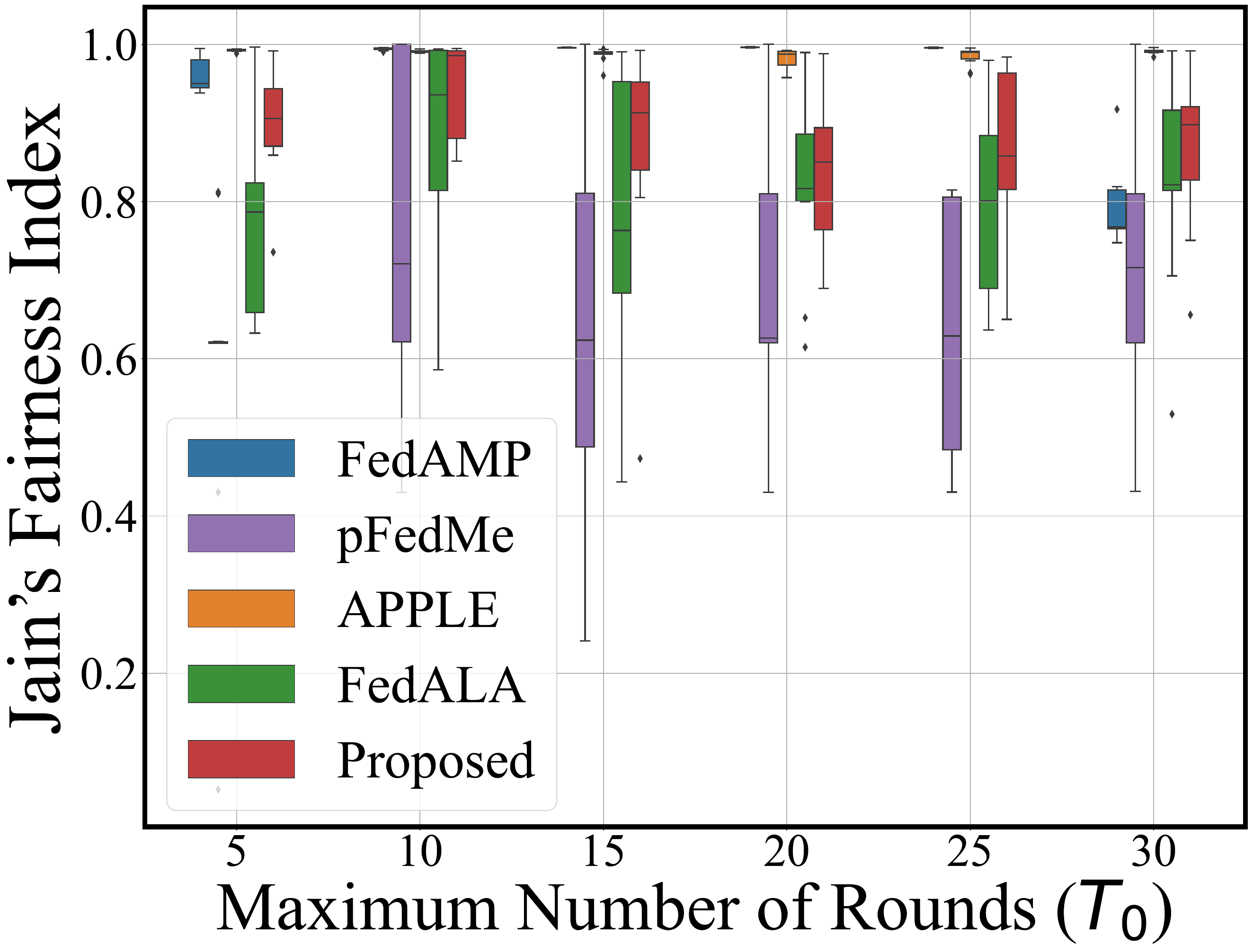}}
\hspace{-1mm}
\subfigure[{Maximum test loss vs. $T_0$ (MLR, MNIST)}]
{
\label{PFL_worst_loss_mnist_mlr}
\includegraphics[width=0.23\textwidth]{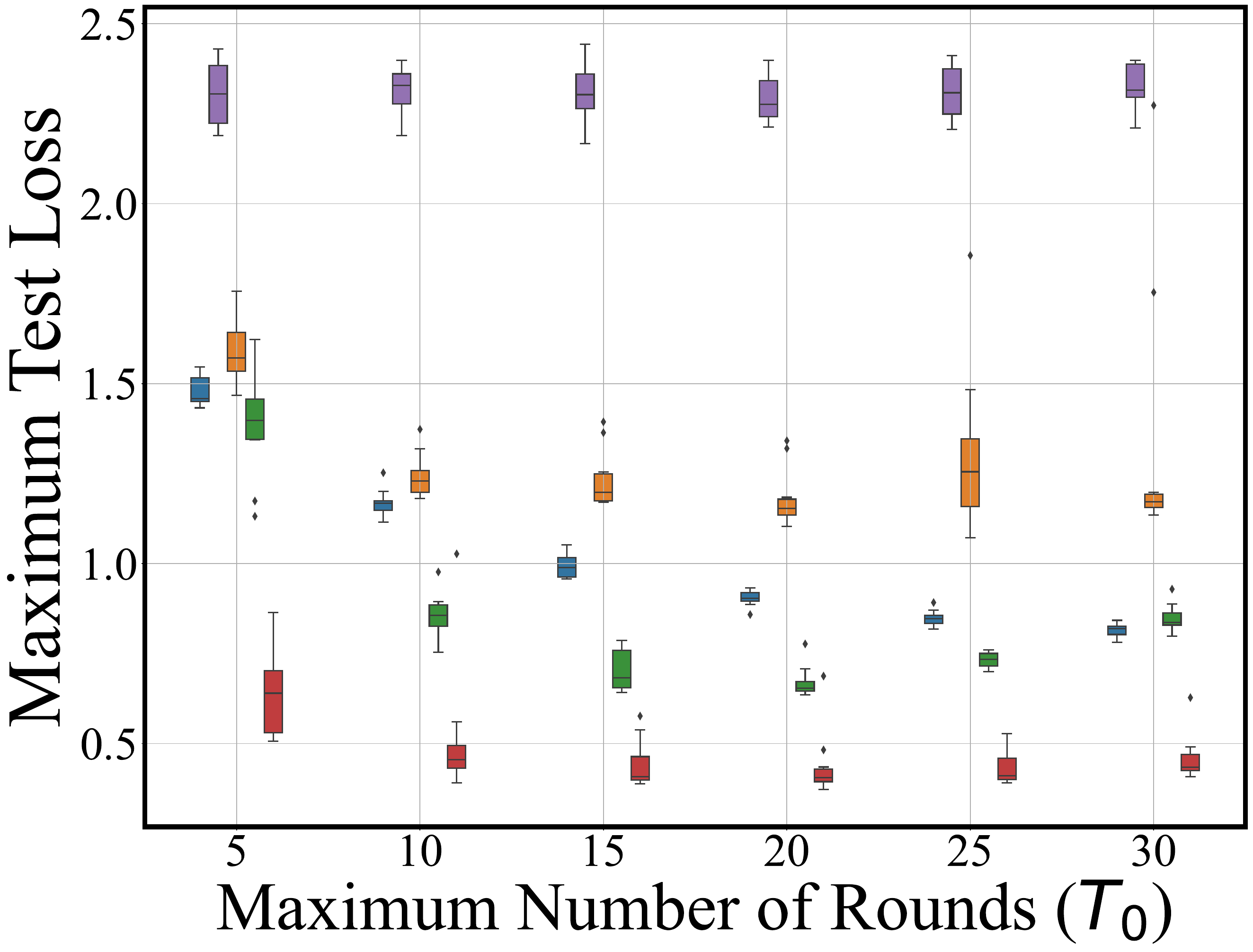}} 
\\
\subfigure[Fairness vs. $T_0$ (CNN,CIFAR10)]{
\label{PFL_fairness_Cifar10_cnn}
\includegraphics[width=0.23\textwidth]{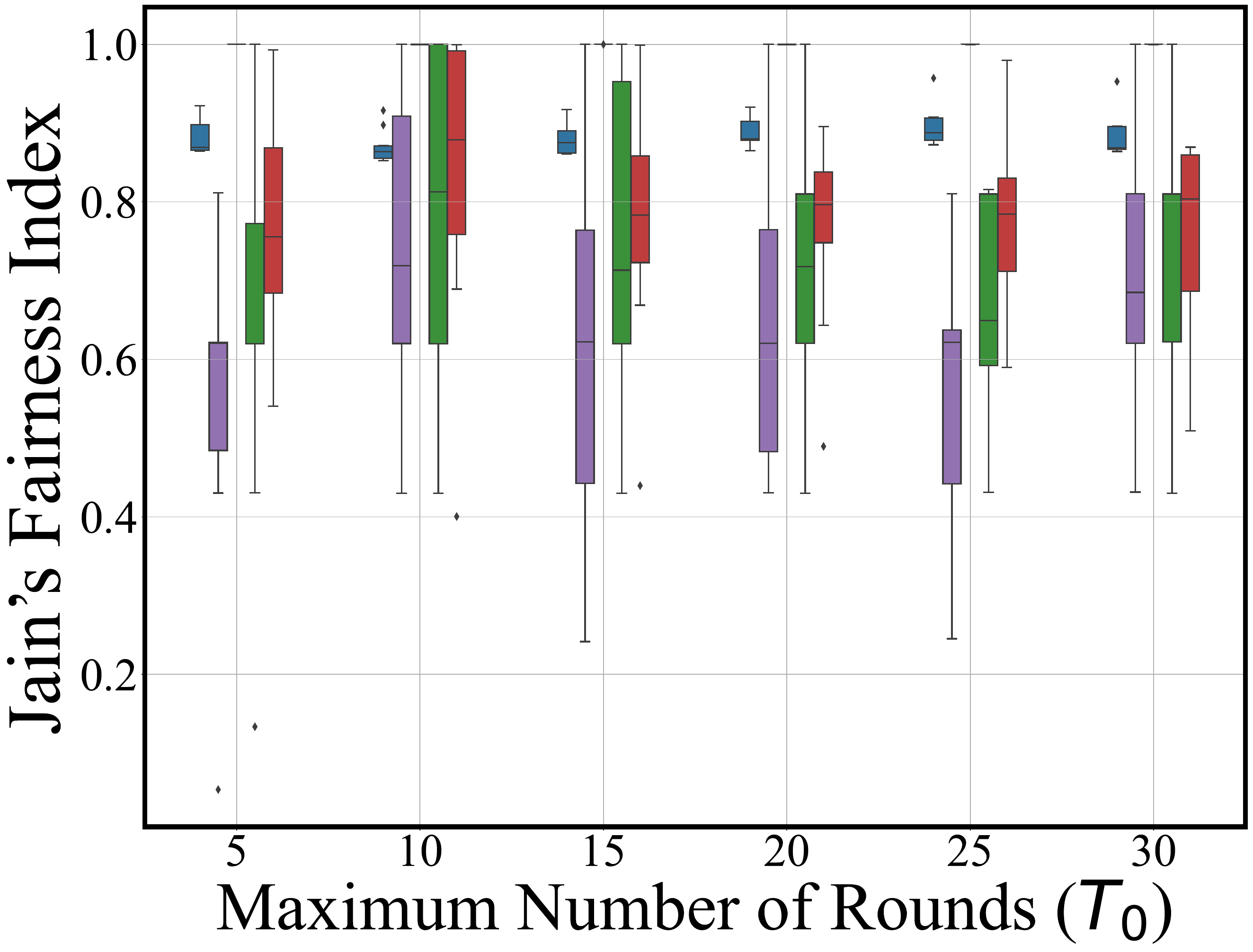}}
\hspace{-1mm}
\subfigure[Maximum test loss vs. $T_0$ (CNN,CIFAR10)]{
\label{PFL_worst_loss_Cifar10_cnn}
\includegraphics[width=0.23\textwidth]{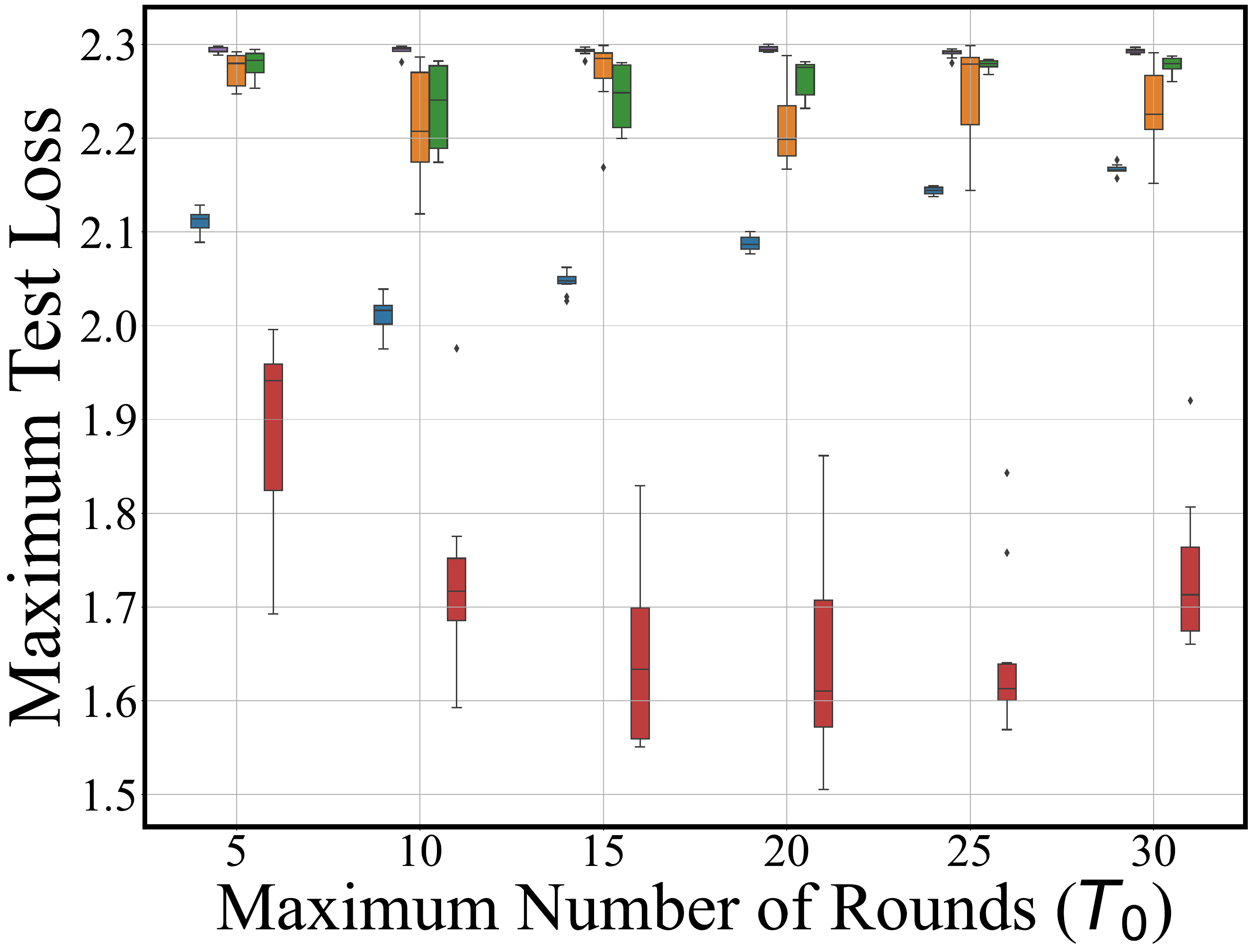}}
\\
\subfigure[Fairness vs. $T_0$ (CNN, FMNIST)]{
\label{PFL_fairness_fmnist_cnn}
\includegraphics[width=0.23\textwidth]{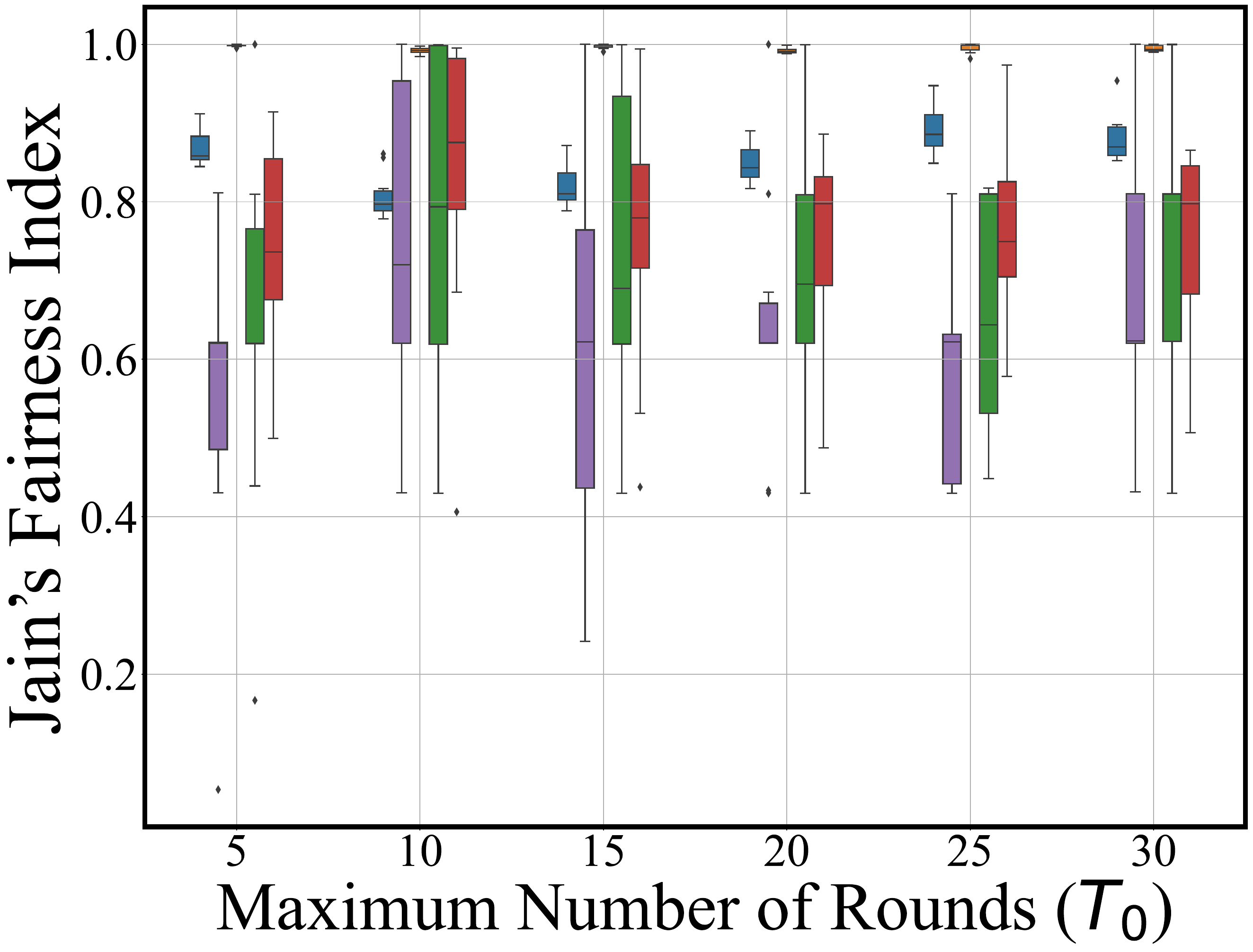}}
\hspace{-1mm}
\subfigure[{Maximum test loss} vs. $T_0$ (CNN,FMNIST)]{
\label{PFL_worst_loss_fmnist_cnn}
\includegraphics[width=0.23\textwidth]{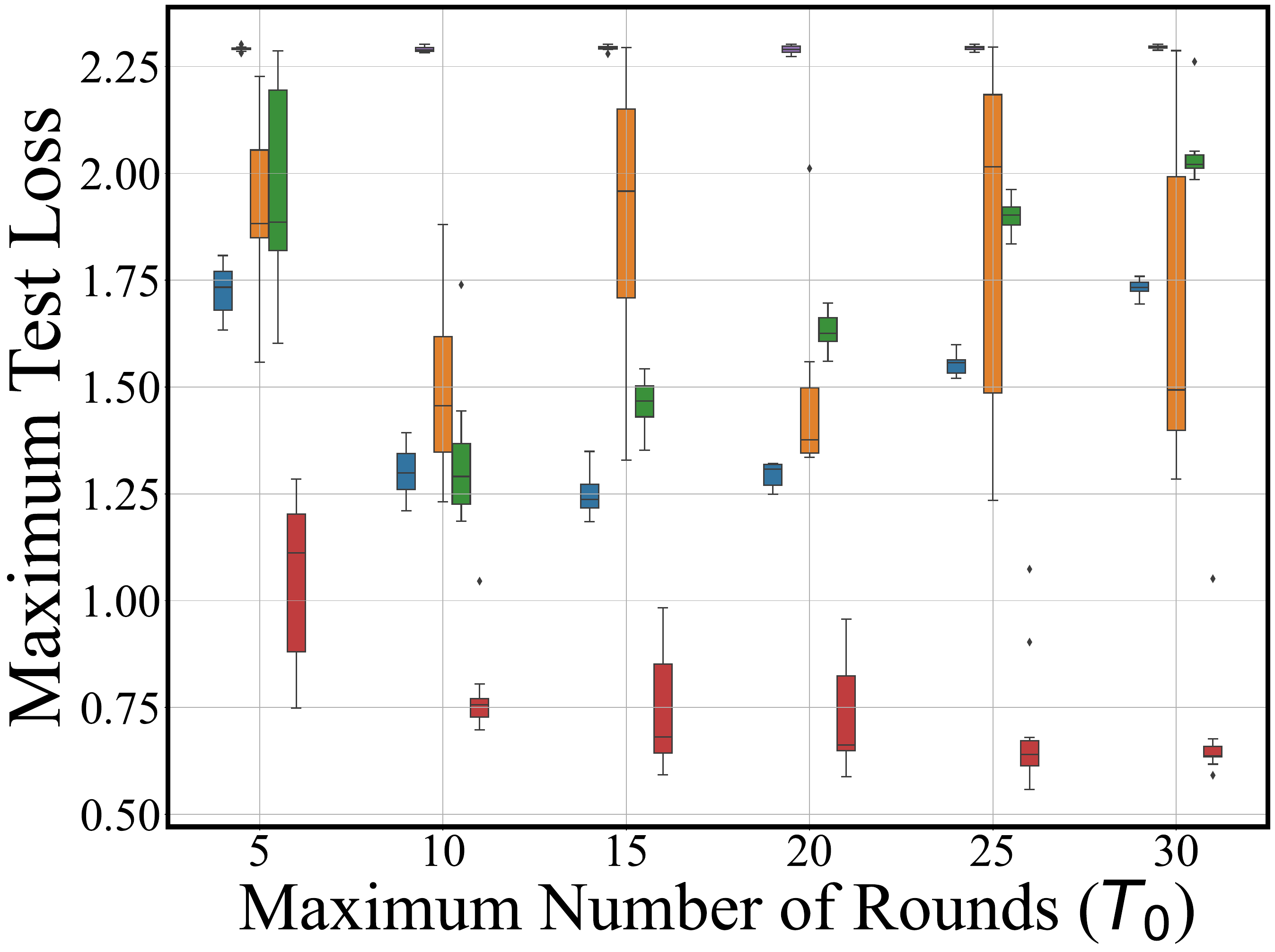}}
\caption{Comparison of fairness and {maximum test loss of all participating clients} between the benchmarks. $\epsilon_{\mathrm{Q}}=1$. $\delta_{\mathrm{Q}}=0.001$ for DNN and MLR, and $\delta_{\mathrm{Q}}=0.005$ for CNN.}
\label{PFL_fairness_worst_loss}
\end{figure}

Fig.~\ref{PFL_acc_convergence} plots the training accuracy against the increasing number $t$ of global aggregations under $T_0=30$, under the proposed WPFL and the benchmarks (i.e., FedAMP, pFedMe, APPLE, and FedALA). The proposed WPFL achieves at least $11.30\%$ better accuracy than the benchmarks, due to the self-adaptive configuration and scheduling policy proposed under imperfect and noisy channels.
Figs.~\ref{PFL_acc} and \ref{PFL_fairness_worst_loss} compare the accuracy, fairness (i.e., Jain's fairness index), and the maximum test loss of all participating clients between the proposed privacy-preserving WFPL framework and the PFL benchmarks (i.e., pFedMe, APPLE, FedAMP, and FedALA), under different $T_0$ values, datasets, and models (i.e., DNN and MLR on the MNIST dataset, and CNN on the FMNIST and CIFAR10 datasets). All schemes are protected through our DP mechanism with the consistent default $\epsilon_\mathrm{Q}$ and $\delta_{\mathrm{Q}}$ values. {The PFL benchmarks do not utilize the proposed configuration policy}, due to their different training procedures and loss functions. 

% As shown in Figs.~\ref{PFL_acc}, \ref{PFL_worst_loss_mnist_dnn}, \ref{PFL_worst_loss_mnist_mlr}, \ref{PFL_worst_loss_Cifar10_cnn}, and \ref{PFL_worst_loss_fmnist_cnn}, our proposed WPFL outperforms the benchmarks in accuracy and maximum test loss of all participating clients.
% Particularly, it is better than the second-best (i.e., FedAMP) by $10.43\%$ in accuracy and $8.16\%$ in maximum test loss.
% Generally, the accuracy and maximum test loss of all participating clients improve first, then degrade as $T_0$ grows, especially for the DNN models on the MNIST dataset and the CNN model on the FMNIST and CIFAR10 datasets. The performance of FedAMP, APPLE, and FedALA deteriorates faster than that of the proposed WPFL. This is because the PFL benchmarks do not adaptively adjust their training process and settings in response to the time-varying errors, resulting in a greater cumulative impact of DP noise, quantization, and transmission errors.

As shown in Figs.~\ref{PFL_acc}, \ref{PFL_worst_loss_mnist_dnn}, \ref{PFL_worst_loss_mnist_mlr}, \ref{PFL_worst_loss_Cifar10_cnn}, and \ref{PFL_worst_loss_fmnist_cnn}, our proposed WPFL outperforms the benchmarks in accuracy and maximum test loss of all participating clients.
Particularly, it is better than the second-best (i.e., FedAMP) by $10.43\%$ in accuracy and $8.16\%$ in maximum test loss.
Generally, the accuracy and maximum test loss of all participating clients improve first, then degrade as $T_0$ grows, especially for the DNN models on the MNIST dataset and the CNN model on the FMNIST and CIFAR10 datasets. 
{Although the trade-off between personalization and generalization is considered for all benchmarks, the performance of pFedMe is always worse than the other benchmarks. This is because the PL model is uploaded directly for global model aggregation in pFedMe, and the trade-off weighting coefficient is fixed across the whole training. }

{To capture the similarity among the PL models, in each round of FedAMP and APPLE, the PL cloud models for clients are obtained by aggregating the PL local models based on weights updated using attention-inducing function and stochastic gradient descent (SGD), respectively. 
Both FedAMP and APPLE require multiple models to upload and download, which increases resource contention, leading to increased transmission errors.

By downloading an FL global model instead of a group of models, FedALA reduces communication overhead. In each round, the PL model of each client is initialized by aggregating the FL global model and its local model with element-wise learnable weights to capture the desired information. However, updates based on the previous round’s PL models may be ineffective due to time-varying bit errors in transmissions.}
The performances of FedAMP, APPLE, and FedALA deteriorate faster than that of the proposed WPFL, because these PFL benchmarks do not adaptively adjust their training process and settings in response to the time-varying errors, and they undergo a greater cumulative impact of DP noise, quantization, and transmission errors.

Figs.~\ref{PFL_fairness_mnist_dnn}, \ref{PFL_fairness_mnist_mlr}, \ref{PFL_fairness_Cifar10_cnn}, and \ref{PFL_fairness_fmnist_cnn} gauge the fairness of the proposed WPFL and its PFL benchmarks. The proposed WPFL outperforms pFedMe and FedALA in fairness, due to our optimization for performance fairness. Although the fairness of the proposed WPFL is worse than that of FedAMP and APPLE, the proposed WPFL can achieve dramatically better accuracy and maximum test loss of all participating clients while maintaining relatively satisfying fairness. By contrast, FedAMP and APPLE offer much poor accuracy and maximum test loss, rendering the achieved fairness less meaningful.

% \begin{figure}[t]
% \centering  
% \subfigure[Fairness vs. $T_0$ (MLR, MNIST)]
% {
% \label{PFL_acc_mnist_mlr}
% \includegraphics[width=0.23\textwidth]{pic/PFL/test/PFL_acc_mnist_mlr_3.pdf}}
% \\
% \subfigure[Accuracy vs. $T_0$ (MLR, MNIST)]
% {
% \label{PFL_jain_mnist_mlr}
% \includegraphics[width=0.23\textwidth]{pic/PFL/test/PFL_jain_mnist_mlr_3.pdf}} 
% \hspace{-1mm}
% \subfigure[Maximum test loss vs. $T_0$ (MLR, MNIST)]
% {
% \label{PFL_worst_loss_mnist_mlr}
% \includegraphics[width=0.23\textwidth]{pic/PFL/test/PFL_worst_loss_mnist_mlr_3.pdf}} 
% \caption{\textcolor{blue}{Comparison of the accuracy, fairness. and maximum test loss of all participating clients between the benchmarks under different $T_0$. $\epsilon_{\mathrm{Q}}=1$.}}
% \end{figure}

\section{Conclusion}
In this paper, we proposed a new WPFL framework, where quantization errors were exploited in coupling with a Gaussian DP mechanism to enhance the privacy of WPFL and min-max fairness was enforced to balance its convergence and fairness. 
% Specifically, we analyzed the convergence upper bounds of PL models under the impact of quantization errors, Gaussian DP noises, and imperfect communication channels. By minimizing the maximum of the bounds, we designed the optimal transmission scheduling strategy that yields min-max fairness for WPFL. 
Experiments validated our analysis and demonstrated that, under the CNN model, our approach substantially outperforms its alternative scheduling strategies (including round-robin, random selection, and non-adjustment) by $87.08\%$ in accuracy, $16.21\%$ in the maximum test loss of participating clients, and $38.37\%$ in fairness. 
With the quantization-assisted Gaussian mechanism, WPFL is $16.10\%$ better in accuracy than using only the Gaussian mechanisms (e.g., MA), 
% and only $9.26\%$ lower than that of WPFL with no DP perturbation, 
validating the idea of exploiting quantization errors for privacy enhancement.
Moreover, our approach dramatically surpasses the wireless deployment of the state-of-the-art PFL (e.g., FedAMP)
% pFedMe and FedALA) 
by $10.43\%$ in accuracy and $8.16\%$ in maximum test loss. 
% Despite being $11.94\%$ and $...\%$ lower than the wireless deployment of FedAMP and APPLE, respectively in fairness, our approach is dramatically better by $26.71\%$ and $50.84\%$ in accuracy and the maximum test loss, respectively. 
% This demonstrates the performance enhancement achieved under our proposed WPFL.

% This work deepens the understanding of how quantization contributes to privacy, as well as the impact of quantization-assisted DP and imperfect communication channels on the convergence of WPFL. It also improves the fairness of PL models, paving the way for WPFL systems in real-world applications, where accuracy, fairness, and privacy are all critical. In the future, we will focus on the privacy benefits of quantization and communication errors, and conduct theoretical analyses on the convergence of decentralized WPFL models under imperfect wireless environments and DP mechanisms.
\bibliography{DittoDP}

% Generated by IEEEtran.bst, version: 1.14 (2015/08/26)
\begin{thebibliography}{10}
\providecommand{\url}[1]{#1}
\csname url@samestyle\endcsname
\providecommand{\newblock}{\relax}
\providecommand{\bibinfo}[2]{#2}
\providecommand{\BIBentrySTDinterwordspacing}{\spaceskip=0pt\relax}
\providecommand{\BIBentryALTinterwordstretchfactor}{4}
\providecommand{\BIBentryALTinterwordspacing}{\spaceskip=\fontdimen2\font plus
\BIBentryALTinterwordstretchfactor\fontdimen3\font minus
  \fontdimen4\font\relax}
\providecommand{\BIBforeignlanguage}[2]{{%
\expandafter\ifx\csname l@#1\endcsname\relax
\typeout{** WARNING: IEEEtran.bst: No hyphenation pattern has been}%
\typeout{** loaded for the language `#1'. Using the pattern for}%
\typeout{** the default language instead.}%
\else
\language=\csname l@#1\endcsname
\fi
#2}}
\providecommand{\BIBdecl}{\relax}
\BIBdecl

\bibitem{li2021ditto}
T.~Li, S.~Hu, A.~Beirami \emph{et~al.}, ``Ditto: Fair and robust federated
  learning through personalization,'' in \emph{Proc. 38th Int. Conf. Mach.
  Learn.}, vol. 139, 2021, pp. 6357--6368.

\bibitem{sami2023over}
H.~U. Sami and B.~G{\"u}ler, ``Over-the-air clustered federated learning,''
  \emph{IEEE Trans. Wirel. Commun.}, vol.~23, no.~7, pp. 7877--7893, 2023.

\bibitem{mestoukirdi2023user}
M.~Mestoukirdi, M.~Zecchin, D.~Gesbert, and Q.~Li, ``User-centric federated
  learning: Trading off wireless resources for personalization,'' \emph{IEEE
  trans. mach. learn. commun. netw.}, vol.~1, pp. 346--359, 2023.

\bibitem{zhao2023ensemble}
Z.~Zhao, J.~Wang, W.~Hong, T.~Q. Quek, Z.~Ding, and M.~Peng, ``Ensemble
  federated learning with non-iid data in wireless networks,'' \emph{IEEE
  Trans. Wirel. Commun.}, vol.~23, no.~4, pp. 3557--3571, 2024.

\bibitem{you2023hierarchical}
C.~You, K.~Guo, H.~H. Yang, and T.~Q. Quek, ``Hierarchical personalized
  federated learning over massive mobile edge computing networks,'' \emph{IEEE
  Trans. Wirel. Commun.}, vol.~22, no.~11, pp. 8141--8157, 2023.

\bibitem{li2019fedmd}
D.~Li and J.~Wang, ``{FedMD: Heterogenous federated learning via model
  distillation},'' in \emph{Proc. NeurIPS Workshop}, 2019, pp. 1--8.

\bibitem{fallah2020personalized}
A.~Fallah, A.~Mokhtari, and A.~Ozdaglar, ``Personalized federated learning with
  theoretical guarantees: A model-agnostic meta-learning approach,''
  \emph{Proc. Adv. Neural Inf. Process. Syst. (NeurIPS)}, vol.~33, pp.
  3557--3568, 2020.

\bibitem{wei2023personalized}
K.~Wei, J.~Li, C.~Ma \emph{et~al.}, ``Personalized federated learning with
  differential privacy and convergence guarantee,'' \emph{IEEE Trans. Inf.
  Forensics Security}, vol.~18, pp. 4488--4503, 2023.

\bibitem{you2022semi}
C.~You, D.~Feng, K.~Guo \emph{et~al.}, ``Semi-synchronous personalized
  federated learning over mobile edge networks,'' \emph{IEEE Trans. Wirel.
  Commun.}, vol.~22, no.~4, pp. 2262--2277, 2022.

\bibitem{t2020personalized}
C.~T~Dinh, N.~Tran, and J.~Nguyen, ``Personalized federated learning with
  moreau envelopes,'' \emph{Proc. Adv. Neural Inf. Process. Syst. (NeurIPS)},
  vol.~33, pp. 21\,394--21\,405, 2020.

\bibitem{li2020federated}
T.~Li, A.~K. Sahu, M.~Zaheer \emph{et~al.}, ``Federated optimization in
  heterogeneous networks,'' \emph{Proc. 3rd Conf. Mach. Learn. Syst. (MLSys)},
  vol.~2, pp. 429--450, 2020.

\bibitem{huang2021personalized}
Y.~Huang, L.~Chu, Z.~Zhou \emph{et~al.}, ``{Personalized cross-silo federated
  learning on non-IID data},'' in \emph{Proc. AAAI Conf. Artif. Intell.},
  vol.~35, no.~9, 2021, pp. 7865--7873.

\bibitem{luo2022adapt}
J.~Luo and S.~Wu, ``Adapt to adaptation: Learning personalization for
  cross-silo federated learning,'' in \emph{Proc. 31th Int. Joint Conf. Artif.
  Intell. (IJCAI)}, vol. 2022, 2022, pp. 2166--2173.

\bibitem{zhang2023fedala}
J.~Zhang, Y.~Hua, H.~Wang \emph{et~al.}, ``{FedALA: Adaptive local aggregation
  for personalized federated learning},'' in \emph{Proc. AAAI Conf. Artif.
  Intell.}, vol.~37, no.~9, 2023, pp. 11\,237--11\,244.

\bibitem{zhang2023federated}
H.~Zhang, M.~Tao, Y.~Shi \emph{et~al.}, ``Federated multi-task learning with
  non-stationary and heterogeneous data in wireless networks,'' \emph{IEEE
  Trans. Wirel. Commun.}, vol.~23, no.~4, pp. 2653--2667, 2023.

\bibitem{cui2024overview}
Q.~Cui, X.~You, N.~Wei \emph{et~al.}, ``{Overview of AI and Communication for
  6G Network: Fundamentals, Challenges, and Future Research Opportunities},''
  \emph{Sci China Inf Sci}, vol.~68, no.~7, p. 171301, 2025.

\bibitem{chen2020joint}
M.~Chen, Z.~Yang, W.~Saad \emph{et~al.}, ``A joint learning and communications
  framework for federated learning over wireless networks,'' \emph{IEEE Trans.
  Wirel. Commun.}, vol.~20, no.~1, pp. 269--283, 2020.

\bibitem{9798757}
W.~Ni, Y.~Liu, Z.~Yang \emph{et~al.}, ``Integrating over-the-air federated
  learning and non-orthogonal multiple access: What role can ris play?''
  \emph{IEEE Trans. Wirel. Commun.}, vol.~21, no.~12, pp. 10\,083--10\,099,
  2022.

\bibitem{9815289}
W.~Ni, Y.~Liu, Y.~C. Eldar \emph{et~al.}, ``Star-ris integrated nonorthogonal
  multiple access and over-the-air federated learning: Framework, analysis, and
  optimization,'' \emph{IEEE Internet Things J.}, vol.~9, no.~18, pp.
  17\,136--17\,156, 2022.

\bibitem{10064038}
W.~Ni, J.~Zheng, and H.~Tian, ``Semi-federated learning for collaborative
  intelligence in massive iot networks,'' \emph{IEEE Internet Things J.},
  vol.~10, no.~13, pp. 11\,942--11\,943, 2023.

\bibitem{abadi2016deep}
M.~Abadi, A.~Chu, I.~Goodfellow \emph{et~al.}, ``Deep learning with
  differential privacy,'' in \emph{n Proc. ACM SIGSAC Conf. Comput. Commun.
  Secur. (CCS)}, 2016, pp. 308--318.

\bibitem{wei2020federated}
K.~Wei, J.~Li, M.~Ding \emph{et~al.}, ``Federated learning with differential
  privacy: Algorithms and performance analysis,'' \emph{IEEE Trans. Inf.
  Forensics Secur.}, vol.~15, pp. 3454--3469, 2020.

\bibitem{zhao2020local}
Y.~Zhao, J.~Zhao, M.~Yang \emph{et~al.}, ``Local differential privacy-based
  federated learning for internet of things,'' \emph{IEEE Internet Things J.},
  vol.~8, no.~11, pp. 8836--8853, 2020.

\bibitem{truex2020ldp}
S.~Truex, L.~Liu, K.-H. Chow \emph{et~al.}, ``{LDP-Fed: Federated learning with
  local differential privacy},'' in \emph{Proc. 3rd ACM Int. Workshop Edge
  Syst. Anal. Netw.}, 2020, pp. 61--66.

\bibitem{yuan2023amplitude}
X.~Yuan, W.~Ni, M.~Ding \emph{et~al.}, ``Amplitude-varying perturbation for
  balancing privacy and utility in federated learning,'' \emph{IEEE Trans. Inf.
  Forensics Security}, vol.~18, pp. 1884--1897, 2023.

\bibitem{liu2024differentially}
H.~Liu, J.~Yan, and Y.-J.~A. Zhang, ``{Differentially private over-the-air
  federated learning over MIMO fading channels},'' \emph{IEEE Trans. Wirel.
  Commun.}, vol.~41, no.~11, pp. 3533--3547, 2024.

\bibitem{chen2022feddual}
Q.~Chen, Z.~Wang, H.~Wang \emph{et~al.}, ``{FedDual: Pair-wise gossip helps
  federated learning in large decentralized networks},'' \emph{IEEE Trans. Inf.
  Forensics Security}, vol.~18, pp. 335--350, 2022.

\bibitem{lang2023joint}
N.~Lang, E.~Sofer, T.~Shaked \emph{et~al.}, ``Joint privacy enhancement and
  quantization in federated learning,'' \emph{IEEE Trans. Signal Process.},
  vol.~71, pp. 295--310, 2023.

\bibitem{lyu2024secure}
X.~Lyu, X.~Hou, C.~Ren, X.~Ge, P.~Yang, Q.~Cui, and X.~Tao, ``Secure and
  efficient federated learning with provable performance guarantees via
  stochastic quantization,'' \emph{IEEE Trans. Inf. Forensics Secur.}, vol.~19,
  pp. 4070--4085, 2024.

\bibitem{wang2024p2cefl}
G.~Wang, Q.~Qi, R.~Han \emph{et~al.}, ``P2cefl: Privacy-preserving and
  communication efficient federated learning with sparse gradient and dithering
  quantization,'' \emph{IEEE Trans. Mob. Comput.}, 2024.

\bibitem{tan2022towards}
A.~Z. Tan, H.~Yu, L.~Cui \emph{et~al.}, ``Towards personalized federated
  learning,'' \emph{IEEE Trans. Neural Netw. Learn. Syst.}, vol.~34, no.~12,
  pp. 9587--9603, 2022.

\bibitem{liu2022privacy}
K.~Liu, S.~Hu, S.~Z. Wu \emph{et~al.}, ``On privacy and personalization in
  cross-silo federated learning,'' \emph{Proc. Adv. Neural Inf. Process. Syst.
  (NeurIPS)}, vol.~35, pp. 5925--5940, 2022.

\bibitem{okegbile2023differentially}
S.~D. Okegbile, J.~Cai, H.~Zheng \emph{et~al.}, ``Differentially private
  federated multi-task learning framework for enhancing human-to-virtual
  connectivity in human digital twin,'' \emph{IEEE J. Sel. Areas Commun.},
  2023.

\bibitem{hu2020personalized}
R.~Hu, Y.~Guo, H.~Li \emph{et~al.}, ``Personalized federated learning with
  differential privacy,'' \emph{IEEE Internet Things J.}, vol.~7, no.~10, pp.
  9530--9539, 2020.

\bibitem{li2020fair}
T.~Li, M.~Sanjabi, A.~Beirami \emph{et~al.}, ``Fair resource allocation in
  federated learning,'' in \emph{Proc. Int. Conf. Learn. Represent.}, 2020, pp.
  1--13.

\bibitem{hu2022federated}
Z.~Hu, K.~Shaloudegi, G.~Zhang \emph{et~al.}, ``Federated learning meets
  multi-objective optimization,'' \emph{IEEE Trans. Netw. Sci. Eng.}, vol.~9,
  no.~4, pp. 2039--2051, 2022.

\bibitem{dwork2014algorithmic}
C.~Dwork and A.~Roth, ``The algorithmic foundations of differential privacy,''
  \emph{Found. Trends Theor. Comput. Sci.}, vol.~9, no. 3--4, pp. 211--407,
  2014.

\bibitem{cho2002general}
K.~Cho and D.~Yoon, ``{On the general BER expression of one-and two-dimensional
  amplitude modulations},'' \emph{IEEE Trans. Commun.}, vol.~50, no.~7, pp.
  1074--1080, 2002.

\bibitem{nasr2019comprehensive}
M.~Nasr, R.~Shokri, and A.~Houmansadr, ``Comprehensive privacy analysis of deep
  learning: Passive and active white-box inference attacks against centralized
  and federated learning,'' in \emph{Proc. IEEE Symp. Secur. Privacy (SP)},
  2019, pp. 739--753.

\bibitem{fredrikson2015model}
M.~Fredrikson, S.~Jha, and T.~Ristenpart, ``Model inversion attacks that
  exploit confidence information and basic countermeasures,'' in \emph{Proc.
  22nd ACM SIGSAC Conf. Comput. Commun. Secur.}, 2015, pp. 1322--1333.

\bibitem{karimi2016linear}
H.~Karimi, J.~Nutini, and M.~Schmidt, ``Linear convergence of gradient and
  proximal-gradient methods under the polyak-{\l}ojasiewicz condition,'' in
  \emph{Proc. Joint Eur. Conf. Mach. Learn. Knowl. Discovery Databases}.\hskip
  1em plus 0.5em minus 0.4em\relax Springer, 2016, pp. 795--811.

\bibitem{o2006metric}
M.~O'Searcoid, \emph{Metric spaces}.\hskip 1em plus 0.5em minus 0.4em\relax
  Berlin, Germany: Springer, 2006.

\bibitem{kuhn1955hungarian}
H.~W. Kuhn, ``{The Hungarian method for the assignment problem},'' \emph{Naval
  Res. Logist. Quart.}, vol.~2, no. 1-2, pp. 83--97, 1955.

\bibitem{jungnickel2005graphs}
D.~Jungnickel and D.~Jungnickel, \emph{Graphs, networks and algorithms}.\hskip
  1em plus 0.5em minus 0.4em\relax Berlin, Germany: Springer, 2005, vol.~3.

\bibitem{sun2021joint}
C.~Sun, W.~Ni, and X.~Wang, ``Joint computation offloading and trajectory
  planning for uav-assisted edge computing,'' \emph{IEEE Trans. Wirel.
  Commun.}, vol.~20, no.~8, pp. 5343--5358, 2021.

\bibitem{8664630}
S.~Wang, T.~Tuor, T.~Salonidis \emph{et~al.}, ``Adaptive federated learning in
  resource constrained edge computing systems,'' \emph{IEEE J. Sel. Areas
  Commun.}, vol.~37, no.~6, pp. 1205--1221, June 2019.

\end{thebibliography}

\twocolumn[\newpage]
\appendix
\subsection{{Proof of \textbf{Theorem \ref{privacy_budget}}}}
\label{privacy_budget_proof}
% \begin{spacing}{1}
    We first focus on the DP mechanism in one communication round. According to \cite{abadi2016deep},
    the privacy bound for multivariate noise is converted into a one-dimensional form by assuming that $\mathcal{M}_\mathrm{Q}(\mathcal{X}_{1}')$ and $\mathcal{M}_\mathrm{Q}(\mathcal{X}_{1})$ are identical except for the first element, without loss of generality. For conciseness, $u_{n,1}(\bullet)$ is written as $u(\bullet)$.  
    Under the quantization-assisted Gaussian mechanism $\mathcal{M}_\mathrm{Q}$, the Max Divergence and the $\delta$-Approximate Max Divergence are given by
    
    \vspace{-\baselineskip}
    % \end{spacing}
    % \begin{singlespace}
    {\begin{subequations} \footnotesize
    \label{divergence}
    \begin{align}
        \label{m_d}
        D_{\infty}\left[\mathcal{M}\left({\mathcal{X}_{1}'}\right)||\mathcal{M}\left({\mathcal{X}_{1}}\right)\right]&=\max_{\chi\in\mathcal{Q_{\mathrm{L}}}}\ln\frac{p}{p_{1}};
        \\
        D_{\infty}^{\delta}\left[\mathcal{M}\left({\mathcal{X}_{1}'}\right)||\mathcal{M}\left({\mathcal{X}_{1}}\right)\right]&=\max_{\chi\in\mathcal{Q_{\mathrm{L}}}}\ln\frac{p-\delta}{p_{1}},
        \label{m_d_delta}
            \end{align}
    \end{subequations} }%
    % \end{singlespace}
    % \begin{spacing}{1}
    where $\mathcal{Q_{\mathrm{L}}}=\{\chi_{1},\ldots,\chi_{|\mathcal{Q_{\mathrm{L}}}|}\}$ collects the quantization levels; for conciseness, $p$ and $p_{1}$ are defined as 
    % \begin{singlespace}

    \vspace{-\baselineskip}
    \begin{subequations} \footnotesize
    \begin{align}
        \label{p}
        p&=(1\!\!-\!\!q)\Pr[\mathcal{M}_{\mathrm{Q}}(\mathcal{X}_{1})\!=\!\chi]\!+\!q\Pr\left[\mathcal{M}_{\mathrm{Q}}(\mathcal{X}_{0})\!=\!\chi\right];
        \\
        \label{p1}
        p_{1}&=\Pr\left[\mathcal{M}_{\mathrm{Q}}(\mathcal{X}_{1})=\chi\right],
        \end{align}
    \end{subequations} 
    % \end{singlespace}
    where $\mathcal{X}_{0}=x_n\cup \mathcal{X}_{1}$ is the adjacent dataset of $\mathcal{X}_{1}$, $\forall d_n\in \mathcal{X}$, $x_n \notin \mathcal{X}_{1}$; $\chi\in\mathcal{Q_{\mathrm{L}}}$ is a quantization level. 
    According to (\ref{DP_mechanism}), when $\chi\in\mathcal{Q_{\mathrm{L}}}\setminus\{\chi_{1},\chi_{|\mathcal{Q_{\mathrm{L}}}|}\}$, the probabilities in (\ref{p}) and (\ref{p1}) are given by

    \vspace{-\baselineskip}
    % \end{spacing}
{\begin{subequations}\footnotesize
    \label{P}
        \begin{align}
            \label{Pr1}
            &\Pr[\mathcal{M}_{\mathrm{Q}}(\mathcal{X}_{1})\!=\!\chi]\!=\!Q(\frac{\chi\!\!-\!\!E_{\mathrm{L}}^{\mathrm{max}}\!\!-\!\!u(\mathcal{X}_{1})}{\sigma_{\mathrm{DP}}})\!\!-\!\!Q(\frac{\chi\!\!+\!\!E_{\mathrm{L}}^{\mathrm{max}}\!\!-\!\!u(\mathcal{X}_{1})}{\sigma_{\mathrm{DP}}});
            \\
            &\Pr[\mathcal{M}_{\mathrm{Q}}(\mathcal{X}_{0})\!=\!\chi]\!=\!Q(\frac{\chi\!\!-\!\!E_{\mathrm{L}}^{\mathrm{max}}\!\!-\!\!u(\mathcal{X}_{0})}{\sigma_{\mathrm{DP}}})\!\!-\!\!Q(\frac{\chi\!\!+\!\!E_{\mathrm{L}}^{\mathrm{max}}\!\!-\!\!u(\mathcal{X}_{0})}{\sigma_{\mathrm{DP}}}).
            \label{Pr2}
        \end{align}
    \end{subequations} }%
    When $\chi\in\{\chi_{1},\chi_{|\mathcal{Q_{\mathrm{L}}}|}\}$, we have
    {
\begin{subequations}\footnotesize
    \label{P'}
        \begin{align}
            \label{Pr1'}
            &\Pr[\mathcal{M}_{\mathrm{Q}}(\mathcal{X}_{1})\!=\!\chi]\!=Q(\frac{\chi\!\!-\!\!E_{\mathrm{L}}^{\mathrm{max}}\!\!-\!\!u(\mathcal{X}_{1})}{\sigma_{\mathrm{DP}}});
            \\
            &\Pr[\mathcal{M}_{\mathrm{Q}}(\mathcal{X}_{0})\!=\!\chi]\!=Q(\frac{\chi\!\!-\!\!E_{\mathrm{L}}^{\mathrm{max}}\!\!-\!\!u(\mathcal{X}_{0})}{\sigma_{\mathrm{DP}}}).
            \label{Pr2'}
        \end{align}
    \end{subequations} }
        % \end{spacing}
        
    Since each parameter of the model is bounded by the clipping threshold $C$, i.e., $u(\bullet)\leq C$, when $\chi\in\mathcal{Q_{\mathrm{L}}}\setminus\{\chi_{1},\chi_{|\mathcal{Q_{\mathrm{L}}}|}\}$,
   
    %according to the probability distribution characteristics of a zero-mean Gaussian random variable, the following inequality is obtained
    % \begin{equation}
    %     \label{ineq_ln}
    %     \Pr[\mathcal{M}_{\mathrm{Q}}(\bullet)=0]\leq\Pr[\mathcal{M}_{\mathrm{Q}}(\bullet)=Q]\leq \Pr[\mathcal{M}_{\mathrm{Q}}(\bullet)=\Lambda].
    % \end{equation}
\begin{subequations}\footnotesize
    \label{ineq_ln}
        \begin{align}
        \Pr[\mathcal{M}_{\mathrm{Q}}(\bullet)\!\!=\!\!\chi]&\!\!\geq \!\!Q(\frac{2C\!\!+\!\!3\sigma_{\mathrm{DP}}\!\!-\!\!E_{\mathrm{L}}^{\mathrm{max}}}{\sigma_{\mathrm{DP}}})\!\!-\!\!Q(\frac{2C\!\!+\!\!3\sigma_{\mathrm{DP}}\!\!+\!\!E_{\mathrm{L}}^{\mathrm{max}}}{\sigma_{\mathrm{DP}}}),
            \\
           \Pr[\mathcal{M}_{\mathrm{Q}}(\bullet)\!\!=\!\!\chi]&\!\!\leq\!\! Q(\frac{-E_{\mathrm{L}}^{\mathrm{max}}}{\sigma_{\mathrm{DP}}})\!\!-\!\!Q(\frac{E_{\mathrm{L}}^{\mathrm{max}}}{\sigma_{\mathrm{DP}}});
        \end{align}
    \end{subequations}%
    when $\chi\in\{\chi_{1},\chi_{|\mathcal{Q_{\mathrm{L}}}|}\}$
    %        {\footnotesize \begin{subequations}
    % \label{ineq_ln'}
    %     \begin{align}
    %     \Pr[\mathcal{M}_{\mathrm{Q}}(\bullet)=\chi]&\geq Q(\frac{2C\!\!+\!\!3\sigma_{\mathrm{DP}}\!\!-\!\!E_{\mathrm{L}}^{\mathrm{max}}}{\sigma_{\mathrm{DP}}}),
    %         \\
    %        \Pr[\mathcal{M}_{\mathrm{Q}}(\bullet)=\chi]&\leq Q(\frac{3\sigma_{\mathrm{DP}}-E_{\mathrm{L}}^{\mathrm{max}}}{\sigma_{\mathrm{DP}}}),
    %     \end{align}
    % \end{subequations}}%
    {\footnotesize 
     % \begin{subequations}
        \begin{align}
            \label{ineq_ln'}
            Q(\frac{2C\!\!+\!\!3\sigma_{\mathrm{DP}}\!\!-\!\!E_{\mathrm{L}}^{\mathrm{max}}}{\sigma_{\mathrm{DP}}})\leq \Pr[\mathcal{M}_{\mathrm{Q}}(\bullet)=\chi] \leq Q(\frac{3\sigma_{\mathrm{DP}}\!\!-\!\!E_{\mathrm{L}}^{\mathrm{max}}}{\sigma_{\mathrm{DP}}}) .
        \end{align}
    % \end{subequations}
    }
        % \end{spacing}

    By substituting (\ref{ineq_ln}) and (\ref{ineq_ln'}) into (\ref{m_d}) and then plugging the results into (\ref{divergence}), it follows that

    \vspace{-\baselineskip}
    \begin{subequations} \footnotesize
    \label{budget_one_round}
        \begin{align}
        D_{\infty}\left[\mathcal{M}\left(\mathcal{X}_{1}'\right)||\mathcal{M}\left({\mathcal{X}_{1}}\right)\right]\!&\leq\! \max\{\ln\frac{\psi}{\psi_{1}},\ln\frac{\psi'}{\psi_{1}'}\},
            \\
           D_{\infty}^{\delta}\left[\mathcal{M}\left({\mathcal{X}_{1}'}\right)||\mathcal{M}\left({\mathcal{X}_{1}}\right)\right]\!&\leq \!\max\{\ln\frac{\psi\!\!-\!\!\delta}{\psi_{1}} ,\ln\frac{\psi'\!\!-\!\!\delta}{\psi_{1}'} \}\!\!=\!\!\epsilon_\mathrm{Q},\nonumber
           \\
           \Rightarrow \delta_\mathrm{Q}\!\!&=\!\!\max\{\psi\!\!-\!\!\psi_{1}e^{\epsilon_\mathrm{Q}},\psi'\!\!-\!\!\psi_{1}'e^{\epsilon_\mathrm{Q}}\}\!.
        \end{align}
    \end{subequations}
    By applying Composition Theorem \cite[Thm 3.16]{dwork2014algorithmic} to (\ref{budget_one_round}), \textbf{Theorem \ref{privacy_budget}} readily follows.
    % \end{spacing}

\subsection{Proof of \textbf{Lemma \ref{Lemma0_0}}}
\label{Lemma0_0_proof}
% By substituting (\ref{received_local_model}), (\ref{noisy_local_model}), and (\ref{error_local_model}) into (\ref{aggregated_glb}), we have 
By substituting (\ref{received_local_model}) and (\ref{noisy_local_model}) into (\ref{aggregated_glb}), we have 
{
\begin{subequations}\footnotesize
    \begin{align}
        \label{receive_glb_difference1}
        % \begin{split}
            \mathbb{E}&\left[\parallel\tilde{\boldsymbol{\omega}}_{\mathrm{L}}^{t}-\boldsymbol{\omega}^{\ast}\parallel^{2}\right] 
        %     {=\mathbb{E}\bigg[\parallel\frac{1}{|\mathcal{N}_{t}|}\underset{n\in\mathcal{N}_{t}}{\sum}\left(\mathbf{s}_{n}^{t}\circ\left(\boldsymbol{u}_{n}^{t}+\boldsymbol{\zeta}_{n,\mathrm{L}}^{t}\right) \right.}\\
        %     &{\left.+\left(\mathbf{1}_{|\boldsymbol{\omega}|}-\mathbf{s}_{n}^{t}\right)\circ\left(\boldsymbol{u}_{n}^{t}+\mathbf{z}_{n}^{t}+\mathbf{E}_{n,\mathrm{L}}^{t}\right)\right)-\boldsymbol{\omega}^{\ast}\parallel^{2}\bigg] }
        % \end{split}
        % \\
        % \label{receive_glb_difference2}
        % \begin{split}
            % &
            =\mathbb{E}\Big[\parallel\frac{1}{|\mathcal{N}_{t}|}{\sum}_{n\in\mathcal{N}_{t}}\boldsymbol{u}_{n}^{t}+\Lambda_0-\boldsymbol{\omega}^{\ast}\parallel^{2}\Big]
        % \end{split}
        \\
        \label{receive_glb_difference3}
        % \begin{split}
            &\leq\frac{1}{|\mathcal{N}_{t}|}\!{\sum}_{n\in\mathcal{N}_{t}}\!\mathbb{E}\Big[\!\!\parallel\!\!\left(\hat{\boldsymbol{\omega}}_{n,\mathrm{G}}^{t}\!\!-\!\!\eta_{\mathrm{F},n}^{t}\nabla F_{n}(\hat{\boldsymbol{\omega}}_{n,\mathrm{G}}^{t})\right)\!\!-\!\!\boldsymbol{\omega}^{\ast}  \!\!+\!\!\Lambda_0\!\!\parallel^{2}\!\!\Big]
        % \end{split}
        % \\
        % \label{receive_glb_difference4}
        % \begin{split}
        %     &=\frac{1}{|\mathcal{N}_{t}|}\underset{n\in\mathcal{N}_{t}}{\sum}\mathbb{E}\bigg[\parallel\hat{\boldsymbol{\omega}}_{n,\mathrm{G}}^{t}-\boldsymbol{\omega}^{\ast}-\eta_{\mathrm{F},n}^{t}\nabla F(\hat{\boldsymbol{\omega}}_{n,\mathrm{G}}^{t}) \\
        %     &+\left(\mathbf{s}_{n}^{t}\circ\boldsymbol{\zeta}_{n,\mathrm{L}}^{t}+\left(1-\mathbf{s}_{n}^{t}\right)\circ\left(\mathbf{z}_{n}^{t}+\mathbf{E}_{n,\mathrm{L}}^{t}\right)\right)\parallel^{2}\bigg]
        % \end{split}
        \\
        \label{receive_glb_difference5}
        \begin{split}
            &=\frac{1}{|\mathcal{N}_{t}|}{\sum}_{n\in\mathcal{N}_{t}}\Big(2\eta_{\mathrm{F},n}^{t}\mathbb{E}\big\langle \nabla F(\hat{\boldsymbol{\omega}}_{n,\mathrm{G}}^{t}),\boldsymbol{\omega}^{\ast}-\hat{\boldsymbol{\omega}}_{n,\mathrm{G}}^{t}\big\rangle
            \\
&+\!\!2\mathbb{E}\big\langle \!\Lambda_0, 
           \!-\!\eta_{\mathrm{F},n}^{t}\nabla F(\hat{\boldsymbol{\omega}}_{n,\mathrm{G}}^{t})\big\rangle+\!2\mathbb{E}\big\langle\!\Lambda_0,\hat{\boldsymbol{\omega}}_{n,\mathrm{G}}^{t}\!\!-\!\!\boldsymbol{\omega}^{\ast}\!\big\rangle \! \Big)
            \\
&\mathbb{E}\!\Big[\!\parallel\!\Lambda_0\!\parallel^{2}\!\Big]\!\!+\!\!\mathbb{E}\left[\parallel\!-\!\eta_{\mathrm{F},n}^{t}\nabla F(\hat{\boldsymbol{\omega}}_{n,\mathrm{G}}^{t})\parallel^{2}\right]+\mathbb{E}\left[\!\parallel\!\hat{\boldsymbol{\omega}}_{n,\mathrm{G}}^{t}-\boldsymbol{\omega}^{\ast}\!\parallel^{2}\!\right]
        \end{split}
        \\
        \label{receive_glb_difference7}
        \begin{split}
               &\leq\frac{1}{|\mathcal{N}_{t}|}{\sum}_{n\in\mathcal{N}_{t}}\Big(\big(1+\frac{1}{\phi_{1}}+\frac{1}{\phi_{2}}\big)\mathbb{E}\big[\parallel\Lambda_0\parallel^{2}\big] \!+\!\left(\!1\!+\!\phi_{1}\!\right)\cdot\\
            &\mathbb{E}\left[\!\parallel\!-\!\eta_{\mathrm{F},n}^{t}\!\nabla F(\hat{\boldsymbol{\omega}}_{n,\mathrm{G}}^{t})\!\!\parallel^{2}\!\right] \!\!+\!\!\big(\!1\!\!+\!\!\phi_{2}\!\!-\!\!\mu\eta_{\mathrm{F},n}^{t}\!\big)\mathbb{E}\big[\parallel\!\!\hat{\boldsymbol{\omega}}_{n,\mathrm{G}}^{t}\!\!-\!\!\boldsymbol{\omega}^{\ast}\!\!\parallel^{2}\!\big]\!\! \Big)
        \end{split}
        \\
        \label{receive_glb_difference8}
        \begin{split}
               &\leq\frac{1}{|\mathcal{N}_{t}|}{\sum}_{n\in\mathcal{N}_{t}}\Big(\big(1\!\!+\!\!\frac{1}{\phi_{1}}\!\!+\!\!\frac{1}{\phi_{2}}\big)\mathbb{E}\big[\!\parallel\!\Lambda_0\!\parallel^{2}\!\big]+\!\!\big(\!1\!\!+\!\!\phi_{2} \\
            % &+\left(1+\phi_{1}\right)\mathbb{E}\left[\parallel-\eta_{\mathrm{F},n}^{t}\nabla F(\hat{\boldsymbol{\omega}}_{n,\mathrm{G}}^{t})\parallel^{2}\right] \\
            &\!\!+\!\!\left(1\!\!+\!\!\phi_{1}\right)L^{2}\!\big(\eta_{\mathrm{F},n}^{t}\big)^{2}\!\!-\!\!\mu\eta_{\mathrm{F},n}^{t}\big)\mathbb{E}\big[\!\parallel\!\hat{\boldsymbol{\omega}}_{n,\mathrm{G}}^{t}\!\!-\!\!\boldsymbol{\omega}^{\ast}\parallel^{2}\!\big] \!\Big),
        \end{split}
    \end{align}
\end{subequations}}%
{where $\Lambda_0=\mathbf{s}_{n}^{t}\circ\boldsymbol{\zeta}_{n,\mathrm{L}}^{t}+(\mathbf{1}_{|\boldsymbol{\omega}|}-\mathbf{s}_{n}^{t})\circ(\mathbf{z}_{n}^{t}+\mathbf{E}_{n,\mathrm{L}}^{t})$. (\ref{receive_glb_difference3}) is obtained by substituting (\ref{update_lc_model}) into (\ref{receive_glb_difference1})} and then utilizing the Cauchy–Schwarz inequality; 
{(\ref{receive_glb_difference7}) is obtained by considering the $\mu$-strong convexity of $F(\cdot)$ and $F\left(\boldsymbol{\omega}^{\ast}\right)-F\left(\hat{\boldsymbol{\omega}}_{n,\mathrm{G}}^{t}\right)\leq 0$ for the first term of \eqref{receive_glb_difference3}, exploiting $2\left\langle a,b\right\rangle \leq v a^{2}+\frac{1}{v}b^{2}$, with $a=\Lambda_0$, $b=-\eta_{\mathrm{F},n}^{t}\nabla F(\hat{\boldsymbol{\omega}}_{n,\mathrm{G}}^{t})$, and $v=\phi_{1}>0$ for the second term, and exploiting $2\left\langle a,b\right\rangle \leq v a^{2}+\frac{1}{v}b^{2}$, with $a=\Lambda_0$, $b=\hat{\boldsymbol{\omega}}_{n,\mathrm{G}}^{t}-\boldsymbol{\omega}^{\ast}$, and $v=\phi_{2}>0$ for the third term; }
% (\ref{receive_glb_difference7}) follows from $F\left(\boldsymbol{\omega}^{\ast}\right)-F\left(\hat{\boldsymbol{\omega}}_{n,\mathrm{G}}^{t}\right)\leq 0$; 
(\ref{receive_glb_difference8}) is obtained under the $L$-smoothness of $F(\cdot)$  %(\ref{receive_glb_difference7}) is obtained
and
% Based on Cauchy–Schwarz inequality, the upper bound of $\Lambda_0$ is obtained, as given by
% \begin{subequations}

\vspace{-\baselineskip}
\begin{small}
    \begin{align} 
        \label{Omega0_1}
        % \begin{split}
        \mathbb{E}\left[\parallel\Lambda_0 \parallel^{2}\right]
            \!\!&=\!\! \mathbb{E}\!\left[\parallel\!\mathbf{s}_{n}^{t}\!\circ\!\boldsymbol{\zeta}_{n,\mathrm{L}}^{t}\!\parallel^{2}\right] \!\!+\!\!\mathbb{E}\left[\parallel\!\left(\mathbf{1}_{|\boldsymbol{\omega}|}\!\!-\!\!\mathbf{s}_{n}^{t}\right)\!\circ\!\left(\mathbf{z}_{n}^{t}\!\!+\!\!\mathbf{E}_{n,\mathrm{L}}^{t}\right)\!\parallel^{2}\right] \nonumber \\
            &+2\mathbb{E}\left\langle \mathbf{s}_{n}^{t}\circ\boldsymbol{\zeta}_{n,\mathrm{L}}^{t},\left(1-\mathbf{s}_{n}^{t}\right)\circ\left(\mathbf{z}_{n}^{t}+\mathbf{E}_{n,\mathrm{L}}^{t}\right)\right\rangle  
        % \end{split}
    \end{align}
% \end{subequations}
\end{small}

We further have
{
\begin{subequations} \small
    \begin{align}
        \label{varOmega1}
        \mathbb{E}&\left[\parallel\mathbf{s}_{n}^{t}\circ\boldsymbol{\zeta}_{n,\mathrm{L}}^{t}\parallel^{2}\right]=\mathbb{E}\Big[{\sum}_{i=1}^{|\boldsymbol{\omega}|}\left(s_{n,i}^{t}\zeta_{n,i,\mathrm{L}}^{t}\right)^{2}\Big] \\
        % \label{varOmega2}
        %  & {\leq\mathbb{E}\left[\sum_{i=1}^{|\boldsymbol{\omega}|}\left(s_{n,i}^{t}\right)^{2}\left(\left|u_{n,i}^{t+1}\right|+C+3\sigma_{\mathrm{DP}}\right)^{2}\right]} \\
        \label{varOmega3}
         & \leq {\sum}_{i=1}^{|\boldsymbol{\omega}|}\mathbb{E}\left[\left(s_{n,i}^{t}\right)^{2}\right]\mathbb{E}\left[\left(\left|u_{n,i}^{t+1}\right|+C+3\sigma_{\mathrm{DP}}\right)^{2}\right] \\
         \label{varOmega4}
         & {\leq2\rho_{n,\mathrm{L}}^{t}\mathbb{E}\Big[{\sum}_{i=1}^{|\boldsymbol{\omega}|}\big(\left|u_{n,i}^{t+1}\right|^{2}+(C+3\sigma_{\mathrm{DP}})^{2}\big)\Big] } \\
        %  \label{varOmega5}
        % &=2\rho_{n,\mathrm{L}}^{t}\mathbb{E}\left[\parallel\boldsymbol{u}_{n}^{t+1}\parallel^{2}\right]+2|\boldsymbol{\omega}|\rho_{n,\mathrm{L}}^{t}(C+3\sigma_{\mathrm{DP}})^{2} \\
        \label{varOmega6}
        &\leq 2\rho_{n,\mathrm{L}}^{t}C^2+2|\boldsymbol{\omega}|\rho_{n,\mathrm{L}}^{t}(C+3\sigma_{\mathrm{DP}})^{2} \,;
    \end{align}
\end{subequations}}%
% where (\ref{varOmega2}) follows from $\zeta_{n,i,\mathrm{L}}^{t} \leq \left|u_{n,i}^{t}\right|+C+3\sigma_{\mathrm{DP}}$; 
where (\ref{varOmega3}) is obtained by substituting $\zeta_{n,i,\mathrm{L}}^{t} \leq |u_{n,i}^{t}|+C+3\sigma_{\mathrm{DP}}$ into \eqref{varOmega1}, followed by exploiting $\mathbb{E}\big[\sum_{i=1}^{|\boldsymbol{\omega}|}a\big]=\sum_{i=1}^{|\boldsymbol{\omega}|}\mathbb{E}\left[a\right]$ and $\mathbb{E}\left[bc\right]=\mathbb{E}\left[b\right]\mathbb{E}\left[c\right]$ with $a=(s_{n,i}^{t})^{2}(|u_{n,i}^{t}|+C+3\sigma_{\mathrm{DP}})^{2}$, $b=(s_{n,i}^{t})^{2}$, $c=(|u_{n,i}^{t}|+C+3\sigma_{\mathrm{DP}})^{2}$ 
% ($b$ and $c$ are independent random variables); 
and (\ref{varOmega4}) is due to $\mathbb{E}[(s_{n,i}^{t})^{2}]=\rho_{n,\mathrm{L}}^{t}$ and $(a+b)^2\leq 2(a^2+b^2)$; (\ref{varOmega6}) follows from $\parallel\boldsymbol{u}_{n}^{t}\parallel^{2} \leq C^2$. 
{
\begin{subequations} \small
    \begin{align}
        %\label{noisy_quan1}
        \mathbb{E}&\left[\parallel(\mathbf{1}_{|\boldsymbol{\omega}|}-\mathbf{s}_{n}^{t})\circ(\mathbf{z}_{n}^{t}+\mathbf{E}_{n,\mathrm{L}}^{t})\parallel^{2}\right] \nonumber\\
        % \label{noisy_quan1}
        % &=\sum_{i=1}^{|\boldsymbol{\omega}|}\mathbb{E}\left[(1-{s}_{n,i}^{t})^2\right]\mathbb{E}\left[ \left(z_{n,i}^{t}+E_{n,i,\mathrm{L}}^{t}\right)^{2}\right]  \\
        % \label{noisy_quan2}
        % &{=\left(1-\rho_{n,\mathrm{L}}^{t}\right)\sum_{i=1}^{|\boldsymbol{\omega}|}\mathbb{E}\left[ \left(z_{n,i}^{t}+E_{n,i,\mathrm{L}}^{t}\right)^{2}\right] } \\
        % \label{noisy_quan3}
        % &{=\left(1\!\!-\!\!\rho_{n,\mathrm{L}}^{t}\right)\sum_{i=1}^{|\boldsymbol{\omega}|}\mathbb{E}\left[ \left(z_{n,i}^{t}\right)^{2}\!\!+\!\!\left(E_{n,i,\mathrm{L}}^{t}\right)^{2}\!\!+\!\!2z_{n,i}^{t}E_{n,i,\mathrm{L}}^{t}\right] } \\
        \label{noisy_quan4}
        &=\left(1-\rho_{n,\mathrm{L}}^{t}\right){\sum}_{i=1}^{|\boldsymbol{\omega}|}\big(\mathbb{E}\big[ \left(z_{n,i}^{t}\right)^{2}\big] +\mathbb{E}\big[ \left(E_{n,i,\mathrm{L}}^{t}\right)^{2}\big] \big) \\
        % \label{noisy_quan5}
        % &{\leq\left(1-\rho_{n,\mathrm{L}}^{t}\right)\sum_{i=1}^{|\boldsymbol{\omega}|}\left(\sigma_{\mathrm{DP}}^{2}+\left(E_{\mathrm{L}}^{\mathrm{max}}\right)^{2}\right) }\\
        \label{noisy_quan6}
        &\leq|\boldsymbol{\omega}|\left(1-\rho_{n,\mathrm{L}}^{t}\right)\big(\sigma_{\mathrm{DP}}^{2}+(E_{\mathrm{L}}^{\mathrm{max}})^{2}\big) \,,
    \end{align}
\end{subequations}}%
{where (\ref{noisy_quan6}) is obtained by substituting $\mathbb{E}\left[(1-s_{n,i}^{t})^2\right]=1-\rho_{n,\mathrm{L}}^{t}$ and $|{E}_{n,i,\mathrm{L}}^{t}| \leq E_{\mathrm{L}}^{\mathrm{max}}$ into (\ref{noisy_quan4}). }
\begin{small}
    \begin{align} 
        \label{L_s}
        \mathbb{E}&\left\langle \mathbf{s}_{n}^{t}\circ\boldsymbol{\zeta}_{n,\mathrm{L}}^{t},\left(1-\mathbf{s}_{n}^{t}\right)\circ\left(\mathbf{z}_{n}^{t}+\mathbf{E}_{n,\mathrm{L}}^{t}\right)\right\rangle\nonumber  \\
        &=\mathbb{E}\Big[\!{\sum}_{i=1}^{|\boldsymbol{\omega}|}\left(s_{n,i,\mathrm{L}}^{t}\zeta_{n,i,\mathrm{L}}^{t}\left(\!1\!\!-\!\!s_{n,i,\mathrm{L}}^{t}\!\right)\!\left(z_{n,i}^{t}\!\!+\!\!E_{n,i,\mathrm{L}}^{t}\right)\!\right)\!\Big] \!=\!0 \,,
    \end{align}%
\end{small}
% \begin{spacing}{1}
{which is due to the fact that $s_{n,i,\mathrm{L}}^{t}=0$ or $s_{n,i,\mathrm{L}}^{t}=1$.}
% \end{spacing}
% \begin{equation}
% \label{L_smooth}
%     \mathbb{E}\left[\parallel-\eta_{\mathrm{F},n}^{t}\nabla F(\hat{\boldsymbol{\omega}}_{n,\mathrm{G}}^{t})\parallel^{2}\right]\leq L^{2}\left(\eta_{\mathrm{F},n}^{t}\right)^{2}\mathbb{E}\left[\parallel\hat{\boldsymbol{\omega}}_{n,\mathrm{G}}^{t}-\boldsymbol{\omega}^{\ast}\parallel^{2}\right] \,,
% \end{equation}
% which is obtained under the $L$-smoothness of $F(\cdot)$.

{
By plugging (\ref{varOmega6}), (\ref{noisy_quan6}), and (\ref{L_s}) into (\ref{Omega0_1}), the upper bound of $\frac{1}{|\mathcal{N}_{t}|}{\sum}_{n\in\mathcal{N}_{t}}\mathbb{E}\left[\parallel\Lambda_0 \parallel^{2}\right]$ can be rewritten as

\vspace{-\baselineskip}
{
\begin{subequations}\footnotesize
  \begin{align}
    \label{Upper_Omega_0_1}
    \frac{1}{|\mathcal{N}_{t}|}&\!{\sum}_{n\in\mathcal{N}_{t}}\!\mathbb{E}\left[\parallel\!\Lambda_0 \!\parallel^{2}\right]\! \leq \!\frac{1}{|\mathcal{N}_{t}|}{\sum}_{n\in\mathcal{N}_{t}}\Big(2\rho_{n,\mathrm{L}}^{t}C^{2}\!\!\nonumber \\
    &+\!\!2|\boldsymbol{\omega}|\rho_{n,\mathrm{L}}^{t}(C\!\!+\!\!3\sigma_{\mathrm{DP}})^{2} \!\!+\!\!|\boldsymbol{\omega}|(1\!\!-\!\!\rho_{n,\mathrm{L}}^{t})\big(\sigma_{\mathrm{DP}}^{2}\!\!+\!\!(E_{\mathrm{L}}^{\mathrm{max}})^{2}\big)\Big) \\
    \label{Upper_Omega_0_2}
    &=\Theta_{\mathrm{L}}^{t}+|\boldsymbol{\omega}|\big(\sigma_{\mathrm{DP}}^{2}+\left(E_{\mathrm{L}}^{\mathrm{max}}\right)^{2}\big)
\end{align}  
\end{subequations}}

By plugging (\ref{Upper_Omega_0_2}) into (\ref{receive_glb_difference8}),  \textbf{Lemma \ref{Lemma0_0}} follows.}

% {
% By plugging (\ref{varOmega6}), (\ref{noisy_quan6}), and (\ref{L_s}) into (\ref{Omega0_1}), 
% % the upper bound of $\frac{1}{|\mathcal{N}_{t}|}\underset{n\in\mathcal{N}_{t}}{\sum}\mathbb{E}\left[\parallel\Lambda_0 \parallel^{2}\right]$ can be rewritten as
% % {\footnotesize
% % \begin{subequations}
% %   \begin{align}
% %     \label{Upper_Omega_0_1}
% %     \frac{1}{|\mathcal{N}_{t}|}\underset{n\in\mathcal{N}_{t}}{\sum}\mathbb{E}&\left[\parallel\Lambda_0 \parallel^{2}\right] \leq \frac{1}{|\mathcal{N}_{t}|}\underset{n\in\mathcal{N}_{t}}{\sum}\Big(2\rho_{n,\mathrm{L}}^{t}C^{2}+2|\boldsymbol{\omega}|\rho_{n,\mathrm{L}}^{t}(C+3\sigma_{\mathrm{DP}})^{2} \nonumber \\
% %     &+|\boldsymbol{\omega}|\left(1-\rho_{n,\mathrm{L}}^{t}\right)\left(\sigma_{\mathrm{DP}}^{2}+\left(E_{\mathrm{L}}^{\mathrm{max}}\right)^{2}\right)\Big) \\
% %     \label{Upper_Omega_0_2}
% %     &=\Theta_{\mathrm{L}}^{t}+|\boldsymbol{\omega}|\left(\sigma_{\mathrm{DP}}^{2}+\left(E_{\mathrm{L}}^{\mathrm{max}}\right)^{2}\right)
% % \end{align}  
% % \end{subequations}}
% then plugging the result of (\ref{Omega0_1}) into (\ref{receive_glb_difference8}),  \textbf{Lemma \ref{Lemma0_0}} follows.}

\subsection{Proof of \textbf{Theorem \ref{theorem1}}}
\label{theorem1_proof}
% By plugging (\ref{noisy_glb}) and (\ref{error_glb_model}) into (\ref{error_glb}), we have
By plugging (\ref{noisy_glb}) into (\ref{error_glb}), we have
\begin{subequations}\footnotesize
    \begin{align}
        \label{d_glm1}
        % \begin{split}
            \mathbb{E}&\big[\parallel\hat{\boldsymbol{\omega}}_{n,\mathrm{G}}^{t+1}-\boldsymbol{\omega}^{\ast}\parallel^{2}\big]=
        %     {\mathbb{E}\left[\parallel\mathbf{s}_{n,\mathrm{G}}^{t+1}\circ\left(\tilde{\boldsymbol{\omega}}_{\mathrm{L}}^{t}+\boldsymbol{\zeta}_{n,\mathrm{G}}^{t+1}\right) \right.} \\
        %     &{\left.+(\mathbf{1}_{|\boldsymbol{\omega}|}-\mathbf{s}_{n,\mathrm{G}}^{t+1})\circ\left(\tilde{\boldsymbol{\omega}}_{\mathrm{L}}^{t}+\mathbf{E}_{n,\mathrm{G}}^{t}\right)-\boldsymbol{\omega}^{\ast}\parallel^{2}\right] }
        % \end{split}
        % \\
        % \label{d_glm2}
        %\begin{split}
% &=
\mathbb{E}\left[\parallel\Lambda_1\!+\!\tilde{\boldsymbol{\omega}}_{\mathrm{L}}^{t}\!\!-\!\boldsymbol{\omega}^{\ast}\!\parallel^{2}\right] \\
        %\end{split} \\
        % \label{d_glm3}
        % % \begin{split}
        %     &=\mathbb{E}\left[\parallel\Lambda_1\parallel^{2}\right]+\mathbb{E}\left[\parallel\tilde{\boldsymbol{\omega}}_{\mathrm{L}}^{t}-\boldsymbol{\omega}^{\ast}\parallel^{2}\right] +2\mathbb{E}\left\langle \Lambda_1,\tilde{\boldsymbol{\omega}}_{\mathrm{L}}^{t}-\boldsymbol{\omega}^{\ast}\right\rangle 
        % \end{split}
        % \\
        \label{d_glm4}
        % \begin{split}
            &\leq\big(1\!\!+\!\!\frac{1}{\varphi_{1}}\big)\mathbb{E}\left[\parallel\Lambda_1\parallel^{2}\right] \!\!+\!\!\left(1\!\!+\!\!\varphi_{1}\right)\mathbb{E}\left[\parallel\tilde{\boldsymbol{\omega}}_{\mathrm{L}}^{t}\!\!-\!\!\boldsymbol{\omega}^{\ast}\parallel^{2}\right]
        % \end{split}
        % \\
        % \label{d_glm5}
        % &\triangleq             \left(1+\frac{1}{\varphi_{1}}\right)\Lambda_1+\left(1+\varphi_{1}\right)\mathbb{E}\left[\parallel\tilde{\boldsymbol{\omega}}_{\mathrm{L}}^{t}-\boldsymbol{\omega}^{\ast}\parallel^{2}\right]
    \end{align}
\end{subequations}%
where $\Lambda_1=\mathbf{s}_{n,\mathrm{G}}^{t+1}\circ\boldsymbol{\zeta}_{n,\mathrm{G}}^{t+1}+(\mathbf{1}_{|\boldsymbol{\omega}|}-\mathbf{s}_{n,\mathrm{G}}^{t+1})\circ\mathbf{E}_{n,\mathrm{G}}^{t}$ is defined for brevity; (\ref{d_glm4}) is obtained by exploiting $(a+b)^2=a^2+b^2+2\left\langle a,b\right\rangle\leq a^2+b^2+v^2a^2+\frac{1}{v^2}b^2$ with $a=\Lambda_1$, $b=\tilde{\boldsymbol{\omega}}_{\mathrm{L}}^{t}-\boldsymbol{\omega}^{\ast}$, and $v=\varphi_{1}\neq 0$.

Further, we establish the upper bound of $\mathbb{E}[\parallel\Lambda_1\parallel^{2}] $ as
\begin{subequations}\footnotesize
    \begin{align}
        \label{Omega1}
        \mathbb{E}&\left[\parallel\!\!\Lambda_1\!\!\parallel^{2}\right] \!\!=\!\!\mathbb{E}\Big[\!{\sum}_{k=1}^{|\boldsymbol{\omega}|}\big(\!s_{n,k,\mathrm{G}}^{t+1}\zeta_{n,k,\mathrm{G}}^{t+1}\!\!+\!\!(1\!\!-\!\!s_{n,k,\mathrm{G}}^{t+1})E_{n,k,\mathrm{G}}^{t}\!\big)^{2}\!\Big] \\
        \label{Omega2}
        &\leq\mathbb{E}\Big[{\sum}_{k=1}^{|\boldsymbol{\omega}|}\left(s_{n,k,\mathrm{G}}^{t+1}\left(\left|\tilde{\omega}_{k,\mathrm{L}}^{t}\right|\!\!+\!\!C\right)\!\!+\!\!(1\!\!-\!\!s_{n,k,\mathrm{G}}^{t+1})E_{\mathrm{G}}^{\mathrm{max}}\right)^{2}\Big] \\
        % \label{Omega3}
        % \begin{split}
        %     &\leq2\mathbb{E}\left[\sum_{k=1}^{|\boldsymbol{\omega}|}\left(\left(s_{n,k,\mathrm{G}}^{t+1}\left|\tilde{\omega}_{k,\mathrm{L}}^{t}\right|\right)^{2} \right.\right.\\
        %     &\left.\left.+\left(\left(\beta_{\mathrm{G}}+(1-\beta_{\mathrm{G}})s_{n,k,\mathrm{G}}^{t+1}\right)C\right)^{2}\right)\right]
        % \end{split}
        % \\
        % \label{Omega4}
        % \begin{split}
        %     &=2\mathbb{E}\left[\sum_{k=1}^{|\boldsymbol{\omega}|}\left(s_{n,k,\mathrm{G}}^{t+1}\left|\tilde{\omega}_{k,\mathrm{L}}^{t}\right|\right)^{2}\right]\\
        %     &+2\mathbb{E}\left[\sum_{k=1}^{|\boldsymbol{\omega}|}\left(\left(\beta_{\mathrm{G}}+(1-\beta_{\mathrm{G}})s_{n,k,\mathrm{G}}^{t+1}\right)C\right)^{2}\right]
        % \end{split}
        % \\
        \label{Omega5}
        &\leq 2\rho_{n,\mathrm{G}}^{t+1}\mathbb{E}\left[\parallel\tilde{\boldsymbol{\omega}}_{\mathrm{L}}^{t}\parallel^{2}\right]\!\!+\!\!2|\boldsymbol{\omega}|\left(\beta_{\mathrm{G}}^{2}\!\!+\!\!(1\!\!-\!\!\beta_{\mathrm{G}}^{2})\rho_{n,\mathrm{G}}^{t+1}\right){C}^{2} ,
    \end{align}
\end{subequations}%
{where (\ref{Omega2}) is based on triangle inequality, and (\ref{Omega5}) is due to the Cauchy-Schwarz inequality.} 

The upper bound of $\mathbb{E}\left[\|\tilde{\boldsymbol{\omega}}_{\mathrm{L}}^{t}\|^{2}\right]$ is established as 

\vspace{-\baselineskip}
\begin{subequations}\footnotesize
    \begin{align}
        \label{up_glb2_1}
        % \begin{split}
\mathbb{E}&\left[\parallel\tilde{\boldsymbol{\omega}}_{\mathrm{L}}^{t}\parallel^{2}\right]
% {\mathbb{E}\left[\parallel\frac{1}{|\mathcal{N}_{t}|}\underset{n\in\mathcal{N}_{t}}{\sum}\left(\mathbf{s}_{n}^{t}\circ\left(\boldsymbol{u}_{n}^{t}+\boldsymbol{\zeta}_{n,\mathrm{L}}^{t}\right)\right.\right.}\\
%         &{\left.\left.+\left(1-\mathbf{s}_{n}^{t}\right)\circ\left(\boldsymbol{u}_{n}^{t}+\mathbf{z}_{n}^{t}+\mathbf{E}_{n,\mathrm{L}}^{t}\right)\right)\parallel^{2}\right] }
        % \end{split}
        % \\
        % \label{up_glb2_2}
        % \begin{split}
            % &
            =\mathbb{E}\Big[\parallel\frac{1}{|\mathcal{N}_{t}|}{\sum}_{n\in\mathcal{N}_{t}}\left(\boldsymbol{u}_{n}^{t}+\Lambda_0\right)\parallel^{2}\Big]
        % \end{split}
        \\
        \label{up_glb2_3}
            % \begin{split}
&\leq\frac{1}{|\mathcal{N}_{t}|}{\sum}_{n\in\mathcal{N}_{t}}\!\!\Big(\big(\!1\!\!+\!\!\frac{1}{\varphi_{2}}\!\big)\mathbb{E}\left[\parallel\boldsymbol{u}_{n}^{t}\parallel^{2}\right]\!\!+\!\!\left(1\!\!+\!\!\varphi_{2}\right) \mathbb{E}\left[\parallel\Lambda_0\parallel^{2}\right]\!\!\Big)
            \\
        % \label{up_glb2_4}
        % \begin{split}
        % &{\leq \left(1+\frac{1}{\varphi_{2}}\right)\frac{1}{|\mathcal{N}_{t}|}\underset{n\in\mathcal{N}_{t}}{\sum}\mathbb{E}\left[\parallel\boldsymbol{u}_{n}^{t}\parallel^{2}\right] }\\
        % &{+ \left(1+\varphi_{2}\right) \left(\Theta_{\mathrm{L}}^{t}+|\boldsymbol{\omega}|\left(\sigma_{\mathrm{DP}}^{2}+\left(E_{\mathrm{L}}^{\mathrm{max}}\right)^{2}\right)\right) }
        % \end{split} 
        % \\
        \label{up_glb2_5}
        &\leq \!\Big(\!1\!\!+\!\!\frac{1}{\varphi_{2}}\!\Big)\!C^2\!\!+\!\! \left(1\!\!+\!\!\varphi_{2}\right)\!\left(\Theta_{\mathrm{L}}^{t}\!\!+\!\!|\boldsymbol{\omega}|\left(\!\sigma_{\mathrm{DP}}^{2}\!\!+\!\!\left(E_{\mathrm{L}}^{\mathrm{max}}\!\right)^{2}\right)\!\right)\! \,,
    \end{align}
\end{subequations}%
where (\ref{up_glb2_3}) is due to the Cauchy–Schwarz inequality and $(a+b)^2\leq a^2+b^2+v^2a^2+\frac{1}{v^2}b^2$ with $a=\Lambda_0$, $b=\boldsymbol{u}_{n}^{t}$, and $v=\varphi_2 \neq 0$; {(\ref{up_glb2_5}) is obtained by substituting (\ref{Upper_Omega_0_2}) into (\ref{up_glb2_3}) and exploiting $\parallel\boldsymbol{u}_{n}^{t}\parallel \leq C$.}
 % \begin{spacing}{1}
Substituting (\ref{up_glb2_5}) into (\ref{Omega5}) yields
% \end{spacing}

\vspace{-\baselineskip}
\begin{footnotesize}
 \begin{align}
     \label{Omega_fi}
     \mathbb{E}\left[\parallel\!\!\Lambda_1\!\!\parallel^{2}\right] \!\! &\leq \!2\rho_{n,\mathrm{G}}^{t+1}\!\left(\!\!\left(\!1\!\!+\!\!\frac{1}{\varphi_{2}}\right)\!\!C^2\!\!\!+ \!\!\left(\!1\!\!+\!\!\varphi_{2}\!\right)\!\left(\!\Theta_{\mathrm{L}}^{t}\!\!+\!|\boldsymbol{\omega}|\left(\sigma_{\mathrm{DP}}^{2}\!\!+\!\!\left(E_{\mathrm{L}}^{\mathrm{max}}\!\right)^{2}\!\right)\!\right)\!\right)\!\!\nonumber\\        &+\!\!2|\boldsymbol{\omega}|\left(\beta_{\mathrm{G}}^{2}\!\!+\!\!(1\!\!-\!\!\beta_{\mathrm{G}}^{2})\rho_{n,\mathrm{G}}^{t\!+\!1}\right){C}^{2} \,.
 \end{align} 
 \end{footnotesize}%
 % \begin{spacing}{1}
    Based on \textbf{Lemma \ref{Lemma0_0}}, we finally obtain (\ref{down_glb_di}) by substituting (\ref{Glb_Up_OneStep}) and (\ref{Omega_fi}) into (\ref{d_glm4}).
 According to (\ref{down_glb_di}), we further have
% \end{spacing}

\vspace{-\baselineskip}
{\begin{subequations}\footnotesize
      \begin{align}
     \mathbb{E}&\left[\parallel\!\!\hat{\boldsymbol{\omega}}_{\mathrm{n,G}}^{t+1}\!\!-\!\!\boldsymbol{\omega}^{\ast}\!\!\parallel^{2}\right]\!\! \leq \varepsilon_{\mathrm{F}}^{\max}\mathbb{E}\left[\parallel\!\!\hat{\boldsymbol{\omega}}_{\mathrm{n,G}}^{t}\!\!-\!\!\boldsymbol{\omega}^{\ast}\!\!\parallel^{2}\right]\!\!+\!\Gamma^{\max}\\
     &\leq \!\!\left(\varepsilon_{\mathrm{F}}^{\max}\right)^{t+1}\mathbb{E}\!\left[\!\parallel\hat{\boldsymbol{\omega}}_{\mathrm{n,G}}^{0}\!\!-\!\!\boldsymbol{\omega}_{n}^{\ast}\parallel^{2}\!\right]\!\!+\!\!\Gamma^{\max}{\sum}_{i=0}^{t}\left(\varepsilon_{\mathrm{F}}^{\max}\right)^{i}.
 \end{align}
 \end{subequations}}%
 By utilizing the geometric series,
%$\sum_{i=0}^{t}a^i=\frac{a^{t+1}-1}{a-1}$ with $a=\varepsilon_{\mathrm{F}}^{\max}\neq 1$,
(\ref{down_glb_overall}) is obtained. With $\varepsilon_{\mathrm{F}}^{\max}\in (0,1)$, the global FL converges.

\subsection{Proof of \textbf{Theorem \ref{theorem_PL_con}}}
\label{theorem_PL_con_proof}
Let $g_{n}(\boldsymbol{\varpi}_{n}^{t};\hat{\boldsymbol{\omega}}_{n,\mathrm{G}}^{t+1})$ denote the stochastic gradient of $f_{n}(\boldsymbol{\varpi}_{n}^{t};\hat{\boldsymbol{\omega}}_{n,\mathrm{G}}^{t+1})$. Then,

\vspace{-\baselineskip}
\begin{equation} \footnotesize
\label{gradient}
    g_{n}(\boldsymbol{\varpi}_{n}^{t};\hat{\boldsymbol{\omega}}_{n,\mathrm{G}}^{t+1})\!\!=\!\!\Big(\!1\!\!-\!\!\frac{\lambda_{n}^{t+1}}{2}\!\Big)\nabla \!F_{n}(\boldsymbol{\varpi}_{n}^{t})\!+\!\lambda_{n}^{t+1}(\boldsymbol{\varpi}_{n}^{t}\!\!-\!\hat{\boldsymbol{\omega}}_{n,\mathrm{G}}^{t+1}) .
\end{equation} 
% \begin{spacing}{1}
As per the PL model of client $n$ at the $(t+1)$-th model update, we have~\cite[Eq.~(96)]{li2021ditto}

\vspace{-\baselineskip}
% \end{spacing}
% \begin{subequations} 
\begin{footnotesize}
	\begin{align}
        %\label{PerOneStep0}
        \label{PerOneStep00_2}
&\mathbb{E}\!\left[\parallel\!\!\tilde{\boldsymbol{\varpi}}_{n}^{t+1}\!\!\!\!-\!\boldsymbol{\varpi}_{n}^{\ast}\!\!\parallel^{2}\right]
% \!=\!\mathbb{E}\left[\parallel\!\!\tilde{\boldsymbol{\varpi}}_{n}^{t}\!\!-\!\eta_{\mathrm{P},n}^{t+1}g_{n}(\tilde{\boldsymbol{\varpi}}_{n}^{t};\hat{\boldsymbol{\omega}}_{n,\mathrm{G}}^{t})\!-\!\boldsymbol{\varpi}_{n}^{\ast}\!\!\parallel^{2}\right] \label{PerOneStep00_1}\\
=\mathbb{E}\left[\parallel\tilde{\boldsymbol{\varpi}}_{n}^{t}-\boldsymbol{\varpi}_{n}^{\ast}\parallel^{2}\right]+\left(\eta_{\mathrm{P},n}^{t+1}\right)^{2}\times \nonumber\\
&\quad\mathbb{E}\!\left[\parallel g_{n}(\tilde{\boldsymbol{\varpi}}_{n}^{t};\hat{\boldsymbol{\omega}}_{n,\mathrm{G}}^{t+1})\parallel^{2}\right]\!\!+\!2\eta_{\mathrm{P},n}^{t+1}\mathbb{E}\!\left\langle g_{n}(\tilde{\boldsymbol{\varpi}}_{n}^{t};\hat{\boldsymbol{\omega}}_{n,\mathrm{G}}^{t+1}),\boldsymbol{\varpi}_{n}^{\ast}\!\!-\!\!\tilde{\boldsymbol{\varpi}}_{n}^{t}\right\rangle.
	\end{align} 
 \end{footnotesize}
	% \end{subequations} 
The third term on the RHS of (\ref{PerOneStep00_2}) can be rewritten as 
{
        \begin{subequations} \footnotesize
	\begin{align}
        \begin{split}
&2\eta_{\mathrm{P},n}^{t+1}\mathbb{E}\left\langle g_{n}(\tilde{\boldsymbol{\varpi}}_{n}^{t};\hat{\boldsymbol{\omega}}_{n,\mathrm{G}}^{t+1}),\boldsymbol{\varpi}_{n}^{\ast}-\tilde{\boldsymbol{\varpi}}_{n}^{t}\right\rangle 
\\
=&2\eta_{\mathrm{P},n}^{t+1}\!\mathbb{E}\Big\langle\!\! \big(1\!\!-\!\!\frac{\lambda_{n}^{t+1}}{2}\!\big)\!\nabla F_{n}\!(\!\tilde{\boldsymbol{\varpi}}_{n}^{t}\!)\!\!+\!\!\lambda_{n}^{t+1}(\tilde{\boldsymbol{\varpi}}_{n}^{t}\!\!-\!\!\hat{\boldsymbol{\omega}}_{n,\mathrm{G}}^{t\!+\!1}),\boldsymbol{\varpi}_{n}^{\ast}\!\!-\!\!\tilde{\boldsymbol{\varpi}}_{n}^{t}\Big\rangle \!\!
        \end{split}
            \label{divengence1}
\\
% \begin{split}
% =&2\eta_{\mathrm{P},n}^{t+1}\mathbb{E}\left\langle \left(1-\frac{\lambda_{n}^{t+1}}{2}\right)\nabla F_{n}(\tilde{\boldsymbol{\varpi}}_{n}^{t}),\boldsymbol{\varpi}_{n}^{\ast}-\tilde{\boldsymbol{\varpi}}_{n}^{t}\right\rangle \\
% &+2\eta_{\mathrm{P},n}^{t+1}\mathbb{E}\left\langle \lambda_{n}^{t+1}(\tilde{\boldsymbol{\varpi}}_{n}^{t}-\hat{\boldsymbol{\omega}}_{n,\mathrm{G}}^{t}),\boldsymbol{\varpi}_{n}^{\ast}-\tilde{\boldsymbol{\varpi}}_{n}^{t}\right\rangle 
% \end{split}
% \label{divengence2}
% \\
\begin{split}
\leq& 2\eta_{\mathrm{P},n}^{t+1}\Big(1-\frac{\lambda_{n}^{t+1}}{2}\Big)\mathbb{E}\left[F_{n}\left(\boldsymbol{\varpi}_{n}^{\ast};\hat{\boldsymbol{\omega}}_{n,\mathrm{G}}^{t+1}\right)-F_{n}\left(\tilde{\boldsymbol{\varpi}}_{n}^{t};\hat{\boldsymbol{\omega}}_{n,\mathrm{G}}^{t+1}\right)\right] \\
&-\eta_{\mathrm{P},n}^{t+1}\Big(1-\frac{\lambda_{n}^{t+1}}{2}\Big)\mu E\left[\parallel\boldsymbol{\varpi}_{n}^{\ast}-\tilde{\boldsymbol{\varpi}}_{n}^{t}\parallel^{2}\right]\\
&+2\eta_{\mathrm{P},n}^{t+1}\mathbb{E}\left\langle \lambda_{n}^{t+1}(\tilde{\boldsymbol{\varpi}}_{n}^{t}-\hat{\boldsymbol{\omega}}_{n,\mathrm{G}}^{t+1}),\boldsymbol{\varpi}_{n}^{\ast}-\tilde{\boldsymbol{\varpi}}_{n}^{t}\right\rangle 
\end{split}
\label{divengence3}
\\
\begin{split}
=&2\eta_{\mathrm{P},n}^{t+1}\mathbb{E}\left[f_{n}\left(\boldsymbol{\varpi}_{n}^{\ast};\hat{\boldsymbol{\omega}}_{n,\mathrm{G}}^{t+1}\right)-f_{n}\left(\tilde{\boldsymbol{\varpi}}_{n}^{t};\hat{\boldsymbol{\omega}}_{n,\mathrm{G}}^{t+1}\right)\right]\\
&-\eta_{\mathrm{P},n}^{t+1}\Big(\big(1\!\!-\!\!\frac{\lambda_{n}^{t+1}}{2}\big)\mu\!\!+\!\!\lambda_{n}^{t+1}\Big)\mathbb{E}\left[\parallel\boldsymbol{\varpi}_{n}^{\ast}\!\!-\!\!\tilde{\boldsymbol{\varpi}}_{n}^{t}\parallel^{2}\right]        ,
\end{split}
\label{divengence4}
	\end{align}
	\end{subequations} }%
where (\ref{divengence1}) is based on (\ref{gradient}), (\ref{divengence3}) is obtained by first considering the $\mu$-strong convexity of $F_n(\cdot)$, followed by substituting (\ref{fn}) into (\ref{divengence3}). 
By substituting (\ref{divengence4}) into (\ref{PerOneStep00_2}), we obtain the upper bound of $\mathbb{E}\left[\parallel\tilde{\boldsymbol{\varpi}}_{n}^{t+1}-\boldsymbol{\varpi}_{n}^{\ast}\parallel^{2}\right]$ as 
% \vspace{-\baselineskip}
{
    \begin{subequations} \footnotesize
	\begin{align}
        \label{PerOneStep1_1}
        \begin{split}
\!\!\!\mathbb{E}&\left[\parallel\!\!\tilde{\boldsymbol{\varpi}}_{n}^{t+1}\!\!-\!\!\boldsymbol{\varpi}_{n}^{\ast}\!\!\parallel^{2}\!\right]\!\leq\!\varepsilon_{\mathrm{P},n}^{t+1}\mathbb{E}\left[\parallel\!\!\tilde{\boldsymbol{\varpi}}_{n}^{t}\!\!-\!\!\boldsymbol{\varpi}_{n}^{\ast}\!\!\parallel^{2}\!\right]\!\!+\!\!\big(\!\eta_{\mathrm{P},n}^{t+1}\!\big)^{2}\mathbb{E}\!\left[\!\parallel \!\!g(\tilde{\boldsymbol{\varpi}}_{n}^{t};\hat{\boldsymbol{\omega}}_{n,\mathrm{G}}^{t+1})\!\!\parallel^{2}\!\right] \\
&+2\eta_{\mathrm{P},n}^{t+1}\mathbb{E}\Big[f_{n}\big(\boldsymbol{\varpi}_{n}^{\ast};\hat{\boldsymbol{\omega}}_{n,\mathrm{G}}^{t+1}\big)-f_{n}\big(\tilde{\boldsymbol{\varpi}}_{n}^{t};\hat{\boldsymbol{\omega}}_{n,\mathrm{G}}^{t+1}\big)\Big] 
\end{split}
\\
% \label{PerOneStep1_2}
% \begin{split}
% =&\varepsilon_{\mathrm{P},n}^{t+1}\!\mathbb{E}\!\left[\parallel\tilde{\boldsymbol{\varpi}}_{n}^{t}\!\!-\!\!\boldsymbol{\varpi}_{n}^{\ast}\parallel^{2}\right]\!\!+\!\!2\eta_{\mathrm{P},n}^{t+1}\mathbb{E}\left[f_{n}\left(\boldsymbol{\varpi}_{n}^{\ast};\hat{\boldsymbol{\omega}}_{n,\mathrm{G}}^{t+1}\right)\!\!-\!\!f_{n}\left(\tilde{\boldsymbol{\varpi}}_{n}^{t};\hat{\boldsymbol{\omega}}_{n,\mathrm{G}}^{t+1}\right)\right] \\
% &+\!\left(\!\eta_{\mathrm{P},n}^{t\!+\!1}\!\right)^{2}\!\mathbb{E}\!\left[\parallel\left(1\!-\!\frac{\lambda_{n}^{t\!+\!1}}{2}\right)\nabla F_{n}(\tilde{\boldsymbol{\varpi}}_{n}^{t})\!\!+\!\!\lambda_{n}^{t\!+\!1}(\tilde{\boldsymbol{\varpi}}_{n}^{t}\!-\!\hat{\boldsymbol{\omega}}_{n,\mathrm{G}}^{t\!+\!1})\parallel^{2}\right]
% % \\
% % &
% \end{split}
% \\
\label{PerOneStep1_3}
\begin{split}
\leq&\varepsilon_{\mathrm{P},n}^{t+1}\mathbb{E}\left[\parallel\tilde{\boldsymbol{\varpi}}_{n}^{t}-\boldsymbol{\varpi}_{n}^{\ast}\parallel^{2}\right]+\big(\eta_{\mathrm{P},n}^{t+1}\big)^{2}\mathbb{E}\left[\parallel g(\tilde{\boldsymbol{\varpi}}_{n}^{t};\boldsymbol{\omega}^{\ast})\parallel^{2}\right]\\
&+\big(\eta_{\mathrm{P},n}^{t+1}\lambda_{n}^{t+1}\big)^{2}\mathbb{E}\big[\parallel\hat{\boldsymbol{\omega}}_{n,\mathrm{G}}^{t+1}-\boldsymbol{\omega}^{\ast}\parallel^{2}\big] \\
&+\!2\big(\eta_{\mathrm{P},n}^{t\!+\!1}\big)^{2}\lambda_{n}^{t\!+\!1}\sqrt{\mathbb{E}\left[\parallel g(\tilde{\boldsymbol{\varpi}}_{n}^{t};\boldsymbol{\omega}^{\ast})\parallel^{2}\right]}\sqrt{\mathbb{E}\big[\parallel\hat{\boldsymbol{\omega}}_{n,\mathrm{G}}^{t\!+\!1}\!-\!\boldsymbol{\omega}^{\ast}\parallel^{2}\big]} \\
&+2\eta_{\mathrm{P},n}^{t+1}\lambda_{n}^{t+1}\sqrt{\mathbb{E}\left[\parallel\tilde{\boldsymbol{\varpi}}_{n}^{t}-\boldsymbol{\varpi}_{n}^{\ast}\parallel^{2}\right]\mathbb{E}\big[\parallel\hat{\boldsymbol{\omega}}_{n,\mathrm{G}}^{t+1}-\boldsymbol{\omega}^{\ast}\parallel^{2}\big]} 
\end{split} 
\\
\label{PerOneStep1_4}
\begin{split}
\leq&\varepsilon_{\mathrm{P},n}^{t+1}\mathbb{E}\left[\parallel\tilde{\boldsymbol{\varpi}}_{n}^{t}-\boldsymbol{\varpi}_{n}^{\ast}\parallel^{2}\right]+\Psi_{n}^{t+1}\mathbb{E}\left[\parallel\hat{\boldsymbol{\omega}}_{n,\mathrm{G}}^{t+1}-\boldsymbol{\omega}^{\ast}\parallel^{2}\right]\\
&+\Big(1+(\lambda_{n}^{t+1})^{3}\Big)\big(\eta_{\mathrm{P},n}^{t+1}\big)^{2}\mathbb{E}\left[\parallel g(\tilde{\boldsymbol{\varpi}}_{n}^{t};\boldsymbol{\omega}^{\ast})\parallel^{2}\right],
\end{split}
	\end{align}
	\end{subequations} }%
 %where (\ref{PerOneStep1_3}) is due to $\mathbb{E}\left[ab\right]\leq \sqrt{\mathbb{E}\left[a^2\right]\mathbb{E}\left[b^2\right]}$, (\ref{PerOneStep1_4}) is due to $2ab \leq a^2+b^2$, 
 where (\ref{PerOneStep1_3}) and (\ref{PerOneStep1_4}) exploit Cauchy–Schwarz inequality.

We further establish the upper bounds for the squared distances between the PL model and the FL local model, and between the PL model and the optimal FL global model, and for the squared norm of the gradient of the PL model:

 \vspace{-\baselineskip}
 \begin{footnotesize}
       	\begin{align} 
        % \footnotesize
        \label{E1}
\mathbb{E}&\left[\parallel\tilde{\boldsymbol{\varpi}}_{n}^{t}-\boldsymbol{u}_{n}^{\ast}\parallel^{2}\right]\leq\frac{1}{\mu^{2}}\mathbb{E}\left[\parallel\nabla F_{n}(\tilde{\boldsymbol{\varpi}}_{n}^{t})\parallel^{2}\right]\leq\frac{G_{0}^{2}}{\mu^{2}} ;\\
	% \end{equation}
% \begin{align} 
        % \begin{aligned}
        \label{E2}
\mathbb{E}&\left[\parallel\tilde{\boldsymbol{\varpi}}_{n}^{t}-\boldsymbol{\omega}^{\ast}\parallel^{2}\right]=\mathbb{E}\left[\parallel\tilde{\boldsymbol{\varpi}}_{n}^{t}-\boldsymbol{u}_{n}^{\ast}+\boldsymbol{u}_{n}^{\ast}-\boldsymbol{\omega}^{\ast}\parallel^{2}\right]\nonumber\\
&\leq\mathbb{E}\left[\parallel\!\tilde{\boldsymbol{\varpi}}_{n}^{t}\!\!-\!\!\boldsymbol{u}_{n}^{\ast}\parallel^{2}\right]\!\!+\!\!\mathbb{E}\left[\parallel\! \boldsymbol{u}_{n}^{\ast}\!\!-\!\!\boldsymbol{\omega}^{\ast}\!\parallel^{2}\right]\!\!+\!\!2\mathbb{E}\left[\parallel\!\tilde{\boldsymbol{\varpi}}_{n}^{t}\!\!-\!\!\boldsymbol{u}_{n}^{\ast}\parallel\!\times \!\parallel \!\boldsymbol{u}_{n}^{\ast}\!\!-\!\!\boldsymbol{\omega}^{\ast}\!\parallel\right]\nonumber\\
&\leq\frac{G_{0}^{2}}{\mu^{2}}+M^{2}+\frac{2MG_{0}}{\mu} ;\\
% \end{aligned}
	% \end{align} 
 % \end{footnotesize}
%  \begin{align}
%         \label{E4}
%         {
% \mathbb{E}\left[\parallel\hat{\boldsymbol{\omega}}_{n,\mathrm{G}}^{t}-\boldsymbol{\omega}^{\ast}\parallel^{2}\right]&=\mathbb{E}\left[\parallel\hat{\boldsymbol{\omega}}_{n,\mathrm{G}}^{t}-\boldsymbol{u}_{n}^{\ast}+\boldsymbol{u}_{n}^{\ast}-\boldsymbol{\omega}^{\ast}\parallel^{2}\right]\nonumber\\
% \leq&\mathbb{E}\left[\parallel\tilde{\boldsymbol{\varpi}}_{n}^{t}-\boldsymbol{u}_{n}^{\ast}\parallel^{2}\right]+\mathbb{E}\left[\parallel \boldsymbol{u}_{n}^{\ast}-\boldsymbol{\omega}^{\ast}\parallel^{2}\right]\nonumber\\
% &+2\mathbb{E}\left[\parallel\tilde{\boldsymbol{\varpi}}_{n}^{t}-\boldsymbol{u}_{n}^{\ast}\parallel\times \parallel \boldsymbol{u}_{n}^{\ast}-\boldsymbol{\omega}^{\ast}\parallel\right]\nonumber\\
% \leq&\frac{G_{0}^{2}}{\mu^{2}}+M^{2}+\frac{2MG_{0}}{\mu} ;
% }
% 	\end{align} 
% \begin{footnotesize}
         % \begin{align}
         % \begin{aligned}
\mathbb{E}&\!\left[\!\parallel \!\!g(\tilde{\boldsymbol{\varpi}}_{n}^{t};\boldsymbol{\omega}^{\ast})\!\!\parallel^{2}\!\right]\!=\!\mathbb{E}\!\Big[\!\parallel\!\!\big(1\!-\!\frac{\lambda_{n}^{t+1}}{2}\!\big)\nabla F_{n}(\tilde{\boldsymbol{\varpi}}_{n}^{t})\!\!+\!\!\lambda_{n}^{t+1}(\!\tilde{\boldsymbol{\varpi}}_{n}^{t}\!\!-\!\!\boldsymbol{\omega}^{\ast}\!)\!\!\parallel^{2}\!\Big]\nonumber \\
        &\leq\!\!\big(\!1\!\!-\!\!\frac{\lambda_{n}^{t\!+\!1}}{2}\!\big)^{2}G_{0}^{2}\!\!+\!\!\big(\!\lambda_{n}^{t\!+\!1}\!\big)^{2}(\frac{G_{0}}{\mu}\!\!+\!\!M)^{2}
    % \nonumber\\
    %     &
        \!\!+\!\!2\big(\!1\!\!-\!\!\frac{\lambda_{n}^{t+1}}{2}\!\big)G_{0}\lambda_{n}^{t\!+\!1}(\frac{G_{0}}{\mu}\!\!+\!\!M)
            \nonumber\\
        &
    \triangleq G_n^{t+1} ,
        % \end{aligned}
        \label{E3} 
	\end{align}  
 \end{footnotesize}%
where (\ref{E1}) is due to the convexity of $F_n(\cdot)$ and the assumption that $\mathbb{E}\left[\parallel\nabla F_{n}({\boldsymbol{\omega}}^{t})\parallel^{2}\right]\leq G_{0}^{2}$.
 (\ref{E2}) is based on the Cauchy-Schwarz inequality and (\ref{E1}). Likewise, (\ref{E3}) is based on (\ref{E2}). {Similarly, we have $\mathbb{E}\left[\parallel\hat{\boldsymbol{\omega}}_{n,\mathrm{G}}^{t}-\boldsymbol{\omega}^{\ast}\parallel^{2}\right]\leq (\frac{G_{0}^{2}}{\mu}+M)^{2}$.}

By plugging (\ref{E1})--(\ref{E3}) into (\ref{PerOneStep1_4}), it readily follows that 
	\begin{equation} \footnotesize
	\begin{aligned}
        \label{PerOneStep2}
        \mathbb{E}\!&\left[\parallel\!\!\tilde{\boldsymbol{\varpi}}_{n}^{t+1}\!\!\!\!-\!\boldsymbol{\varpi}_{n}^{\ast}\!\!\parallel^{2}\right]
        \!\leq \varepsilon_{\mathrm{P},n}^{t+1}\mathbb{E}\left[\parallel\!\!\tilde{\boldsymbol{\varpi}}_{n}^{t}-\boldsymbol{\varpi}_{n}^{\ast}\!\!\parallel^{2}\right]\!\\
&+\!\big(1\!\!+\!\!(\lambda_{n}^{t+1})^{3}\big)\left(\eta_{\mathrm{P},n}^{t+1}\right)^{2}G_n^{t+1}\!\!+\!\!\Psi_{n}^{t+1}\mathbb{E}\left[\parallel\hat{\boldsymbol{\omega}}_{n,\mathrm{G}}^{t+1}\!\!-\!\!\boldsymbol{\omega}^{\ast}\parallel^{2}\right],
\end{aligned}
	\end{equation} 
	
% where $\varepsilon_{\mathrm{P}}=1-\eta_{\mathrm{P},n}^{t+1}\left(\left(1-\frac{\lambda_{n}^{t+1}}{2}\right)\mu+\lambda_{n}^{t+1}\right)+\eta_{\mathrm{P},n}^{t+1}$.
    Based on \textbf{Theorem \ref{theorem1}} and (\ref{PerOneStep2}), it follows that
    {
\begin{subequations} \footnotesize
\label{per_con}
\begin{align}
    \label{di_per_t2_1}
\begin{split}
\mathbb{E}&\left[\parallel\tilde{\boldsymbol{\varpi}}_{n}^{t+1}-\boldsymbol{\varpi}_{n}^{\ast}\parallel^{2}\right]\leq \varepsilon_{\mathrm{P},n}^{t+1}\mathbb{E}\left[\parallel\tilde{\boldsymbol{\varpi}}_{n}^{t}-\boldsymbol{\varpi}_{n}^{\ast}\parallel^{2}\right]\\
&\!+\!\!\big(\!1\!\!+\!\!(\lambda_{n}^{t+1})^{3}\!\big)\!\!\left(\!\eta_{\mathrm{P},n}^{t+1}\!\right)^{2}\!\!G_n^{t+1}\!+\!\!\Psi_{n}^{t+1}\big( \!h_1(\rho_{n,\mathrm{G}}^{t+1})\Theta_{\mathrm{L}}^{t}\\
&+\Gamma_{0}\rho_{n,\mathrm{G}}^{t+1}+\Gamma_{1}\!\!+\!\!\frac{1}{|\mathcal{N}_{t}|}\!{\sum}_{n\in\mathcal{N}_{t}} \varepsilon_{\mathrm{F},n}^t\mathbb{E}\left[\!\parallel\!\hat{\boldsymbol{\omega}}_{n,\mathrm{G}}^{t}\!\!-\!\!\boldsymbol{\omega}^{\ast}\!\parallel^{2}\right]\big) ,
\end{split}
\\
\label{di_per_t2_2}
% \begin{split}
 &=\varepsilon_{\mathrm{P},n}^{t+1}\mathbb{E}\left[\parallel\tilde{\boldsymbol{\varpi}}_{n}^{t}-\boldsymbol{\varpi}_{n}^{\ast}\parallel^{2}\right]\!+ \Phi_n^{t+1}.
% \end{split}
\end{align}
\end{subequations}
}
Then, \textbf{Theorem \ref{theorem_PL_con}} follows.

\subsection{Proof of \textbf{Theorem \ref{Convergence_t2}}}
\label{Convergence_t2_proof}
For the brevity of notation, we define
% \begin{subequations}
    \begin{equation} \footnotesize
         \label{ACDE_t}
B_{t+1}\triangleq\big(1+(\lambda_{n}^{t+1})^{3}\big)\big(\eta_{\mathrm{P},n}^{t+1}\big)^{2}G_{n}^{t+1}.
         % \\
% D_{t+1}&\triangleq\left(\left(\eta_{\mathrm{P},n}^{t+1}\right)^{2}+1\right)\left(\lambda_{n}^{t+1}\right)^{2}+\frac{\left(\eta_{\mathrm{P},n}^{t+1}\right)^{3}}{\lambda_{n}^{t+1}},\\
        % \Gamma_{t+1}&\triangleq h_{1}(\rho_{n,\mathrm{G}}^{t+1})\Theta_{\mathrm{L}}^{t}+\Gamma_{0}\rho_{n,\mathrm{G}}^{t+1}+\Gamma_{1}.
    \end{equation}%
% \end{subequations}
Let $B^{\max}$ and $\Psi^{\max}$ denote the maxima of $B_{t+1}$ and $\Psi_{n}^{t+1}$, respectively. 
Based on (\ref{lambda}) and (\ref{element_error_pr}), $B^{\max}$, $\Psi^{\max}$, and $\Gamma^{\max}$ exist and are unique since 
%$\eta_{\mathrm{P},n}^{t+1} \in \Omega_0$ and 
$e_{n,k,\mathrm{L}}^{t}<1$, $\forall n,t$ in (\ref{SER}).
% for any round $t$ and client~$n$. 
The maximum of $\Phi_{n}^{t+1}$, denoted as $\Phi^{\max}$, exists and is unique; i.e., $\Phi^{\max}=B^{\max}+\Psi^{\max}(\Gamma^{\max}+(\frac{G_0^2}{\mu}+M)^2 \varepsilon_{\mathrm{F}}^{\max})$.
% due to $\Phi_{n}^{t+1}=B_{t+1}+\Psi_{n}^{t+1}(\Gamma_{t+1}+\frac{(G_{0}^{2}+M\mu)^{2}}{|\mathcal{N}_{t}|\mu^{2}}\!\!\underset{n\in\mathcal{N}_{t}}{\sum} \!\!\varepsilon_{\mathrm{F},n}^t)$ and $\frac{(G_{0}^{2}+M\mu)^{2}}{|\mathcal{N}_{t}|\mu^{2}}\!\!\underset{n\in\mathcal{N}_{t}}{\sum} \!\!\varepsilon_{\mathrm{F},n}^t\leq (\frac{G_0^2}{\mu}+M)^2 \varepsilon_{\mathrm{F}}^{\max}$.

By substituting (\ref{ACDE_t}) into (\ref{PerOneStep2}) and (\ref{down_glb_di}), it follows that

\vspace{-\baselineskip}
\begin{subequations} \footnotesize
    \begin{align}
        \label{con_t2_1}
        \mathbb{E}&\left[\parallel\tilde{\boldsymbol{\varpi}}_{n}^{t+1}\!\!-\!\!\boldsymbol{\varpi}_{n}^{\ast}\parallel^{2}\right] \leq \!\varepsilon_{\mathrm{P}}^{\max}\mathbb{E}\left[\parallel\tilde{\boldsymbol{\varpi}}_{n}^{t}\!\!-\!\!\boldsymbol{\varpi}_{n}^{\ast}\parallel^{2}\right] \!\!+\!\!\Phi^{\max}\\
        % &\leq \!\varepsilon_{\mathrm{P}}^{\max}\mathbb{E}\!\left[\parallel\!\!\boldsymbol{\varpi}_{n}^{t}\!\!-\!\!\boldsymbol{\varpi}_{n}^{\ast}\!\!\parallel^{2}\right]\!\!+\!\!h_{0}^{\max}\mathbb{E}\!\left[\parallel\!\!\hat{\boldsymbol{\omega}}_{\mathrm{G}}^{t}\!\!-\!\!\boldsymbol{\omega}^{\ast}\!\!\parallel^{2}\right]\!\!+\!\!\Phi^{\max}
        % \\
        &\leq \left(\varepsilon_{\mathrm{P}}^{\max}\right)^{t+1}\mathbb{E}\left[\parallel{\boldsymbol{\varpi}}_{n}^{0}\!\!-\!\!\boldsymbol{\varpi}_{n}^{\ast}\parallel^{2}\right]\!\!+\!\!\Phi^{\max}{\sum}_{i=0}^{t}\left(\varepsilon_{\mathrm{P}}^{\max}\right)^{i}
        \label{con_t2_2}
        \\
        \label{con_t2_3}
        &=\!\left(\varepsilon_{\mathrm{P}}^{\max}\right)^{t+1}\mathbb{E}\left[\parallel{\boldsymbol{\varpi}}_{n}^{0}\!\!-\!\!\boldsymbol{\varpi}_{n}^{\ast}\parallel^{2}\right]\!\!+\!\!\frac{(\varepsilon_{\mathrm{P}}^{\max})^{t+1}\!\!-\!\!1}{\varepsilon_{\mathrm{P}}^{\max}-1}\Phi^{\max}.
    \end{align}
\end{subequations}%
According to (\ref{con_t2_3}), with $\varepsilon_{\mathrm{P}}^{\max}<1$, WPFL under imperfect channels converges as $t$ increases. After $T$ aggregations, the convergence upper bound of the PL model is (\ref{T_convergence}).

\subsection{Proof of \textbf{Theorem \ref{theorem4}}}
\label{theorem4_proof}
% \begin{spacing}{1}
    For conciseness, $\eta_{\mathrm{P},n}^{t+1}$ and  $\lambda_{n}^{t+1}$ are written as $\eta$ and $\lambda$, respectively. Based on (\ref{Phi_n}) and (\ref{lambda}), the second derivative of $\Phi_{n}^{t+1}$ with respect to $\eta$ is given by

    \vspace{-\baselineskip}
    % \end{spacing}
    \begin{footnotesize}
    \begin{align}
        &\frac{\partial^{2}\Phi_{n}^{t\!+\!1}}{\partial \eta^{2}}\!=\!\frac{6{a_0}{\lambda}^{2}}{{\eta}^{2}}\Big(MH_1
        \!\!+\!\!\big((2{\eta}^{2}\!\!-\!\!\frac{3}{4}{b_0})^{2}\!\!+\!\!\frac{7}{16}b_{0}^{2}\big)(\frac{2\!-\!\mu}{2\mu}G_{0})\Big)\nonumber\\
        &\!+\!{2(1\!+\!{\lambda}^{3})}H_2\!+\!\frac{2{\eta}^{6}H_3\!+\!2a_{0}^{3}({b_0}+{\eta}(-{\mu}+{\eta}))^{3}H_4}{{a_0}{\eta}^{4}({b_0}-{\mu}{\eta}+{\eta}^{2})^{3}}  (\Gamma_{2}\rho_{n,\mathrm{G}}^{t+1}+\Gamma_{3}) , \nonumber
    \end{align}
    \end{footnotesize}%
    where $a=(1-\frac{\mu}{2})^{-1}$, $b=1-\varepsilon_{\mathrm{P}}^{t+1}$, and $Q=(\frac{1}{\mu}-\frac{1}{2})G_{0}+M$, and $H_i$, $i=1,2,3,4$, are given by
    {
    \begin{subequations} \footnotesize
        \begin{align}
            \label{H1}
            H_1&=4{\eta}^{4}-2\mu {\eta}^{3}-3{b_0}{\eta}^{2}+{b_0}\mu {\eta}+{b_0}^{2}; \\
            \label{H2}
            \begin{split}
                H_2&=6({a_0}Q)^{2}{\eta}^{2}+\left(6G_{0}({a_0}Q)-6\mu ({a_0}Q)^{2}\right){\eta} 
                \\
        &+\left(\left(2{b_0}+\mu^{2}\right)({a_0}Q)^{2}-2({a_0}Q)\mu G_{0}+G_{0}^{2}\right);
            \end{split}
            \\
            \label{H3}
            H_3&=6{b_0}^{2}+{b_0}{\eta}(-8{\mu}+3{\eta})+{\eta}^{2}(3{\mu}^{2}-3{\mu}{\eta}+{\eta}^{2});\\
            \label{H4}
            H_4&=3{b_0}^{2} \!\!+\!\!{\eta}^{4}(1\!\!+\!\!{\mu}^{2}\!\!-\!\!6{\mu}{\eta}\!\!+\!\!6{\eta}^{2})\!\!+\!\!{b_0}(-2{\mu}{\eta}\!\!+\!\!2{\eta}^{4}). 
        \end{align}
    \end{subequations}}%
    By analyzing the monotonicity of $H_i$, $i=1,2,3,4$, and comparing their minima with respect to $\eta \in \Omega_0^{t+1}\cup\Omega_1^{t+1}$, it can be found that $H_1$-$H_4$ are positive with $\eta \in \Omega_0^{t+1}\cup\Omega_1^{t+1}$. Therefore, the second derivative of $\Phi_n^t$ is positive in $\eta \in \Omega_0^{t+1}\cup\Omega_1^{t+1}$. This concludes this proof.

% \bibliography{DittoDP}

\end{document}